\documentclass[10pt,a4paper,reqno]{amsart}
\usepackage{amsthm,amssymb,amstext,graphicx,comment}
\usepackage[foot]{amsaddr}
\usepackage[a4paper, margin=2.5cm]{geometry}
\usepackage{amscd}
\usepackage{mathtools}
\usepackage[bibliography= common]{apxproof}
\usepackage[linesnumbered,ruled,vlined]{algorithm2e}
\SetKwInput{KwInput}{Input}                
\SetKwInput{KwOutput}{Output}              

\usepackage{hyperref}
\usepackage{booktabs}
\usepackage{float}
\usepackage[table]{xcolor}
\usepackage{placeins}
\usepackage{mathabx}
\usepackage{bm}

\usepackage{bbm}
\usepackage{subcaption}
\usepackage{array}
\usepackage{empheq}
\usepackage{siunitx}
\usepackage{textcomp}
\usepackage{multirow}
\hypersetup{
  colorlinks=true,
  pdfauthor={© Sven Karbach},
  pdfsubject={},
  pdfkeywords={}
}
\makeatletter
\newsavebox{\@brx}
\newcommand{\llangle}[1][]{\savebox{\@brx}{\(\m@th{#1\langle}\)}%
  \mathopen{\copy\@brx\kern-0.5\wd\@brx\usebox{\@brx}}}
\newcommand{\rrangle}[1][]{\savebox{\@brx}{\(\m@th{#1\rangle}\)}%
  \mathclose{\copy\@brx\kern-0.5\wd\@brx\usebox{\@brx}}}
\makeatother
\DeclareMathOperator{\diag}{diag}

\usepackage{enumitem}
\usepackage[utf8]{inputenc}
\usepackage[T1]{fontenc}
\usepackage{lmodern}
\usepackage{bibentry}
\usepackage{soul}
\usepackage{tikz}
\usepackage{eso-pic}
\usepackage[capitalise]{cleveref}
\newtheorem{theorem}{Theorem}[section]
\newtheorem{corollary}[theorem]{Corollary}
\newtheorem{assumption}[theorem]{Assumption}
\newtheorem{example}[theorem]{Example}

\newtheorem{definition}[theorem]{Definition}
\newtheorem{remark}[theorem]{Remark}
\newtheorem{proposition}[theorem]{Proposition}
\newtheoremstyle{example}
  {.3\baselineskip}
  {.3\baselineskip}
  {\normalsize}  
  {0pt}       
  {\bfseries} 
  {.}         
  {5pt plus 1pt minus 1pt} 
  {}          
\theoremstyle{example}
\DeclareMathOperator\Tr{Tr}

\newcommand*\D{\mathop{}\mathrm{d}}

\newcommand{\e}{\mathrm{e}}

\newcommand{\ve}{\mathbf{e}}

\begin{document}

\title[Semi-Static Hedging of Covariance Risk in Multi-Asset Derivatives]{Semi-Static Variance-Optimal Hedging of Covariance Risk in Multi-Asset Derivatives}

\address{\today. Korteweg-de Vries Institute for Mathematics and Informatics Institute, University of Amsterdam, The Netherlands}

\thanks{The authors gratefully acknowledge financial support
from the Amsterdam University Fund through the AUF Impact Call (Fall 2024)}

\maketitle
\begin{center}
  \begin{tabular}{cc}
    \textsc{Konstantinos Chatziandreou} & \textsc{Sven Karbach} \\
    \small[\texttt{\MakeLowercase{k.chatziandreou@uva.nl}}] &
                                                              \small[\texttt{\MakeLowercase{sven@karbach.org}}]
    \\ 
  \end{tabular}
\end{center}

\begin{abstract}
We develop a semi-static framework for the variance-optimal hedging of multi-asset derivatives exposed to correlation and covariance risk. The approach combines continuous-time dynamic trading in the underlying assets with a static portfolio of auxiliary contingent claims. Using a multivariate Galtchouk--Kunita--Watanabe decomposition, we show that the resulting global mean-variance problem decouples naturally into an inner continuous-time projection onto the space spanned by the underlying assets and an outer finite-dimensional quadratic optimization over the static hedging instruments. To systematically select suitable auxiliary claims, we leverage multidimensional functional spanning theory, establishing that otherwise unhedgeable cross-gamma exposures can be structurally mitigated through static strips of vanilla, product, and spread options. As a central application, we derive explicit semi-static replication formulas for covariance swaps and geometric dispersion trades. Our framework accommodates a broad class of asset dynamics, including quadratic and stochastic Volterra covariance models, as well as affine stochastic covariance models with jumps, yielding tractable semi-closed-form solutions via Fourier transform techniques. Extensive numerical experiments demonstrate that incorporating optimally weighted static strips of cross-asset instruments substantially reduces the mean-squared hedging error relative to purely dynamic benchmark strategies across various model classes.

\end{abstract}

\section{Introduction}
\label{sec:introduction}

The valuation and hedging of multi-asset derivatives remains a central challenge in modern financial mathematics, largely driven by the presence of correlation and covariance risk. Across equity, fixed-income, commodity, and foreign-exchange markets, the joint dynamics of asset prices give rise to complex dependence structures and, in particular, non-trivial cross-gamma exposures. These exposures cannot, in general, be eliminated through dynamic trading in the underlying assets alone so that markets with multi-asset contingent claims are inherently incomplete. The need to manage such joint risk factors has both theoretical and practical implications. From a modeling perspective, stochastic dependence substantially increases both analytical and computational complexity, particularly in high-dimensional settings. From a market perspective, this complexity has fostered the emergence of a rich class of correlation-sensitive instruments and trading strategies, including dispersion trades, correlation swaps, and covariance swaps, which play a central role in portfolio hedging, volatility arbitrage, and relative-value trading. Despite their growing importance, the development of analytically tractable and economically effective hedging frameworks for such instruments remains a challenging and, to a large extent, open problem.

Historically, the mathematical finance literature on multi-asset derivatives has largely focused on extending the complete-market paradigm of the Black--Scholes model. Early contributions include Margrabe’s seminal pricing formula for exchange options~\cite{Margabe1978}, along with subsequent generalizations to more complex cross-asset payoffs~\cite{ExchangeOptions}. Building on these foundations, Stulz~\cite{STULZ1982161} and Cox et al.~\cite{COX1979229} derived closed-form solutions for options on the maximum or minimum of two assets. Even for standard instruments such as basket options, where the sum of lognormal random variables fails to remain lognormal, an extensive body of work has developed, ranging from moment-matching approximations~\cite{Kolpakov2014} and spectral methods~\cite{Dempster2002} to static replication bounds~\cite{multiassetCarr}. More recently, Pellegrino~\cite{Pellegrino2016} proposed closed-form approximations for spread options that retain much of the classical tractability. While these classical and approximation-based approaches offer robust pricing mechanisms within deterministic volatility frameworks, they obscure the profound hedging challenges introduced by stochastic co-movements. When volatilities and correlations follow joint stochastic dynamics, the classical replication arguments underlying geometric Brownian motion break down, and market participants must resort to probabilistic frameworks capable of consistently capturing the evolution of the full instantaneous covariance matrix. 

In this paper, we develop a unified framework for the pricing and variance-optimal (VO) hedging of multi-asset contingent claims in incomplete markets driven by stochastic instantaneous covariance processes. Our approach combines multivariate quadratic hedging theory, a broad class of stochastic covariance models (including affine, quadratic, and Volterra specifications), and multidimensional replication formulas within a tractable Fourier-analytic setting. Particular emphasis is placed on correlation- and covariance-linked derivatives, and on the structural role of auxiliary instruments in reducing residual hedging risk. 

\subsection{Related Literature and Methodological Context}

A related perspective on the hedging of correlation risk is provided by Cont et al.~\cite{RamaContEquityCorrelation}, who construct an arbitrage-free multi-asset pricing model consistent with observed index and constituent option prices via a random mixture of reference models. The distribution of mixture weights is determined through a well-posed convex optimization problem, yielding a flexible calibration framework that quantifies the residual model uncertainty associated with implied correlation structures. As a by-product, they derive static hedging strategies that reduce the sensitivity of multi-asset derivatives to this uncertainty. While their objective differs from ours (as they focus on an exponential hedging formulation~\cite{exphedging} via purely static strategies rather than variance-optimal semi-static replication), their work is closely related in spirit, underscoring the structural role of option-based static positions in mitigating unhedgeable correlation risk.

To operationalize the static hedging component of our framework, we rely on the theory of functional payoff replication. In the univariate setting, the inverse problem, i.e. replicating an arbitrary payoff $h(S_T)$ using a portfolio of vanilla options, is well studied. Under mild regularity conditions, payoffs can be decomposed into an affine position in the underlying asset alongside a continuum of out-of-the-money (OTM) calls and puts, as established by Breeden and Litzenberger~\cite{Breeden} and Carr and Madan~\cite{Carr}. This functional representation has two profound consequences: European option surfaces structurally encode the entire risk-neutral marginal density, and highly non-linear functionals (such as the log-contracts underlying variance swaps) admit robust, model-free static representations via OTM option strips. In market practice, this theoretical continuum of strikes is approximated by finite strips, where discrete portfolio weights are optimized to minimize the replication error (see, e.g., Carr and Wu~\cite{CarrStaticHedgingStandardOptions} and Wu~\cite{WuRobustHedging}). A related one-dimensional semi-static perspective is provided by the robust replication and correlation-immunization methodology of Carr--Lee type strategies, where suitably chosen static European claims are combined with dynamic trading in the underlying so as to neutralize first-order leverage effects; in this sense, our multivariate variance-optimal framework may be viewed as an exact projection-based counterpart of that idea for covariance-sensitive claims \cite{Lin04032019}.

However, for derivatives sensitive to correlation and cross-asset co-movement, e.g., dispersion trades, covariance swaps, or best-of/worst-of structure, a univariate option strip is fundamentally insufficient. The payoff structurally depends on the \emph{joint} distribution of the asset vector $(S_T^1,\dots,S_T^d)$. A multivariate spanning theorem is therefore required to provide a systematic methodology for (approximately) completing the market. Several multidimensional extensions have been proposed: Bossu et al.~\cite{Bossu2021,Bossu04052022} derive spanning representations utilizing basket options via integral equations and inverse Radon transforms, while Schmutz~\cite{Schmutz01072011} investigates semi-static replication for Margrabe-type payoffs. Furthermore, Cui et al.~\cite{Cui04052022} and Madan et al. ~\cite{Madan2021Pricing} establish Carr--Madan-type extensions relying on multiple integrals of product options.

\subsection{Summary of Main Contributions}

Building upon this theoretical foundation, our paper makes the following main contributions:

\subsubsection*{General Semi-Static VO Hedging Framework.} 
We formulate a general semi-static variance-optimal hedging framework in a stochastic covariance setting. Using a multivariate Galtchouk--Kunita--Watanabe (GKW) decomposition, we use the fact that the global mean-variance problem decomposes naturally into an \emph{inner} dynamic projection problem and an \emph{outer} static quadratic optimization over auxiliary instruments. The dynamic component corresponds to the orthogonal projection of the target claim onto the space of stochastic integrals generated by the underlying asset prices, while the static component reduces to a finite-dimensional convex problem driven by the covariance structure of the residual martingales. This yields a transparent interpretation of semi-static hedging as a projection in an enlarged space of attainable claims, extending the semi-static framework of~\cite{semistaticsparse} to a genuinely multivariate stochastic covariance setting.

\subsubsection*{Instrument Selection via Multidimensional Spanning.} 
Developed in Section~\ref{sec:replicationtheory}, our second contribution concerns the systematic construction of auxiliary instruments. Making systematic use of the aforementioned multidimensional spanning formulas, we show that sufficiently regular multi-asset payoffs can be decomposed into static portfolios of vanilla and product (quanto) or spread options. This characterization provides a theoretically sound methodology for selecting static hedging instruments, relating the structural features of a multi-asset payoff (most notably its cross-gamma dependence) to concrete classes of OTC-traded options. 

\subsubsection*{Multivariate Fourier Representation.} 
In Sections~\ref{sec:fourierrepresentation} and~\ref{sec:affinecase}, we establish a multivariate Fourier representation of the GKW decomposition. Extending earlier univariate results, we show that Fourier integration and GKW projection commute under suitable integrability conditions. Consequently, both the optimal dynamic hedging strategy and the residual orthogonal martingale admit explicit Fourier representations in terms of exponential basis functions. This provides a tractable analytical alternative to BSDE-based numerical approaches and enables explicit hedging formulas in high-dimensional settings.

\subsubsection*{Analytical Solutions in Affine Models.} 
Our fourth contribution, presented in Section~\ref{sec:affinecase}, is the explicit analytical treatment of the variance-optimal semi-static hedging problem in affine stochastic covariance (ASC) models. Specializing to this class (which includes the Wishart stochastic covariance model alongside jump extensions) we derive semi-closed-form expressions for the variance-optimal dynamic hedge ratios, the minimal mean-squared hedging error, and the covariance matrices governing the optimal static portfolio weights. 

\subsubsection*{Applications to Covariance Swaps and Dispersion Trades and its Numerics} 
As a central application, Section~\ref{sec:replicationtheory} analyzes covariance swaps and geometric dispersion-type products. We derive a pathwise decomposition of realized covariance and establish that the terminal payoff admits a semi-static representation involving dynamic trading in the underlying assets, together with static positions in log-contracts and a continuum of out-of-the-money product options. Embedding this into the variance-optimal framework allows us to characterize the residual unhedgeable variance risk explicitly, these results are complemented by Section~\ref{sec:covswap-semistatic} where we tailor the setting to affine stochastic covariance models. Finally, Section~\ref{sec:numerics-and-examples} provides extensive numerical illustrations. The results demonstrate that appropriately chosen static instruments can substantially reduce the mean-squared hedging error. Owing to the affine structure, all quantities of interest can be computed via low-dimensional matrix-valued Riccati equations and Fourier inversion, ensuring computational tractability for practical implementation and calibration.

\subsection{Organization of the Paper}

The remainder of the paper is organized as follows. Section~\ref{sec:VO-framework} introduces the market setup, formulates the semi-static variance-optimal hedging problem, and derives its solution via a multivariate Galtchouk--Kunita--Watanabe decomposition. Section~\ref{sec:replicationtheory} develops the multidimensional spanning theory for instrument selection and establishes exact static replication formulas for covariance swaps and geometric dispersion trades. In Section~\ref{sec:fourierrepresentation}, we derive explicit Fourier-analytic representations for the semi-static hedging problem for suitably integrable payoffs. Building on these results, Section~\ref{sec:affinecase} specializes the framework to affine stochastic covariance models, yielding semi-closed-form expressions for the optimal dynamic hedge ratios and static portfolio weights, with Section \ref{sec:covswap-semistatic} complementing these results for the semi-static hedging of covariance swaps. Section~\ref{sec:numerics-and-examples} presents comprehensive numerical experiments illustrating the efficacy of the proposed hedging strategies across various model classes. 

\section{Variance-Optimal Pricing and Hedging of Multi-Asset Contingent Claims}\label{sec:VO-framework}

In this section, we formulate and solve the \emph{semi-static} variance-optimal (VO) hedging problem in a multi-asset financial market. The hedging agent is permitted to trade \emph{dynamically} in the underlying risky assets and, simultaneously, to take \emph{static} positions in a finite family of auxiliary hedging instruments (e.g., vanilla options, spread or quanto options, and covariance-linked products). The resulting global mean-variance problem admits an elegant and computationally tractable solution via a multivariate Galtchouk-Kunita-Watanabe (GKW) decomposition. Specifically, the dynamic trading component is determined by the orthogonal projection of the claim onto the stable subspace generated by the underlying assets, while the optimal static weights are obtained by solving a finite-dimensional quadratic program.

\subsection{Market Setup and Model Dynamics}\label{sec:MarketSetup}

We consider a financial market defined on a filtered probability space $(\Omega,\mathcal{F},\mathbb{F},\mathbb{Q})$ satisfying the usual conditions. The market consists of a risk-free cash account $S^0=(S_t^0)_{t\geq 0}$ and $d\in\mathbb{N}$ risky assets with price processes denoted by $\bm{S}_t=(S_t^1,\dots,S_t^d)^\top\in\mathbb{R}^{d}$. The cash account evolves according to
\begin{align}\label{eq:bank-account}
    \D S_t^0 &= S_t^0 r_t\,\D t,\qquad S_0^0=1,
\end{align}
where we assume throughout that the risk-free rate $(r_t)_{t\in[0,T]}$ is a deterministic function of time. 

Let $\mathbb{S}_+^d$ denote the cone of symmetric, positive semidefinite $d\times d$ matrices.  
The central modeling assumption of this paper is that the risky assets evolve under a \emph{stochastic covariance} specification. More precisely, the matrix-valued quadratic covariation of the log-returns is modeled by an adapted process
\begin{align}\label{eq:inst-cov}
    \bm{\Sigma} &= (\bm{\Sigma}_t)_{t\ge 0}, \qquad \bm{\Sigma}_t \in \mathbb{S}_+^d,
\end{align}
which we refer to as the \emph{instantaneous covariance process}. Throughout the present section, our results rely on the semimartingale structure and square-integrability of the price processes only. Nevertheless, for concreteness and to facilitate later explicit computations, we introduce three principal model classes that motivate the subsequent analysis.

\paragraph{Class 1: Affine Stochastic Covariance Models with Jumps.}
In the affine stochastic covariance specification, the instantaneous covariance process $\bm{\Sigma}$ takes values in the positive semidefinite cone $\mathbb{S}_+^d$. The most prominent continuous representative is the Wishart process, which we extend here to accommodate jump dynamics. The covariance state matrix is governed by the stochastic differential equation (SDE)
\begin{align}\label{eq:AffineDynamics}
    \D\bm{\Sigma}_t &= \big(\bm{\Omega}+\mathcal{L}(\bm{\Sigma}_t)\big)\,\D t + \sqrt{\bm{\Sigma}_t}\,\D\bm{W}_t \bm{A} + \bm{A}^\top \D\bm{W}_t^\top \sqrt{\bm{\Sigma}_t} + \D\bm{J}_t,\qquad \bm{\Sigma}_0=\bm{x}_0,
\end{align}
where $\bm{\Omega}\in\mathbb{S}_+^d$, $\mathcal{L}\colon\mathbb{S}^d \to \mathbb{S}^d$ is a linear map (e.g., $\mathcal{L}(\bm{Y})=\bm{M}\bm{Y}+\bm{Y}\bm{M}^\top$ for some $\bm{M} \in \mathbb{R}^{d \times d}$), $\bm{A}\in\mathbb{R}^{d\times d}$, $\bm{W}$ is a $d\times d$ matrix of independent standard Brownian motions, and $\bm{J}$ is an $\mathbb{S}_+^d$-valued pure jump process of finite variation. The risky assets follow
\begin{align}\label{eq:StockDynamicsWishart}
    \D\bm{S}_t &= \diag(\bm{S}_t)\sqrt{\bm{\Sigma}_t}\,\D\bm{B}_t,
\end{align}
yielding the continuous cross quadratic variation structure $\D\llangle\log \bm{S}\rrangle_t= \bm{\Sigma}_t\D t$. This class is highly tractable owing to the availability of closed-form Fourier transforms.

\paragraph{Class 2: Quadratic Covariance Models.}
To avoid the stringent Feller-type existence conditions and numerical complexities associated with matrix square roots, the quadratic class models the covariance via an unconstrained matrix-valued factor process. Let $\bm{X}$ be a $d\times d$ continuous Markovian process, typically an Ornstein-Uhlenbeck (OU) process of the form
\begin{align}\label{eq:QuadraticDynamics}
    \D\bm{X}_t &= (\bm{M} - \bm{\Lambda}\bm{X}_t)\,\D t + \bm{Q}\,\D\bm{W}_t, \qquad \bm{X}_0=\bm{x}_0,
\end{align}
where $\bm{M}, \bm{\Lambda}, \bm{Q}$ are parameter matrices of appropriate dimensions. The instantaneous covariance is defined via the quadratic map $\bm{\Sigma}_t = \bm{X}_t \bm{X}_t^\top$, guaranteeing that $\bm{\Sigma}_t \in \mathbb{S}_+^d$ almost surely without imposing boundary reflections. The asset prices are governed by
\begin{align}\label{eq:StockDynamicsQuad}
    \D\bm{S}_t &= \diag(\bm{S}_t)\,\bm{X}_t\,\D\bm{B}_t,
\end{align}
yielding $\D\llangle \log \bm{S}\rrangle_t = (\bm{X}_t \bm{X}_t^\top)\,\D t$. 

\paragraph{Class 3: Stochastic Volterra Covariance Models.}
To encompass path-dependent and rough volatility phenomena, one can extend the unconstrained state process to the non-Markovian setting. Let $\bm{X}$ be a $d\times d$ matrix-valued Volterra process of the form
\begin{align}\label{eq:VolterraDynamics}
    \bm{X}_t &= \bm{g}_0(t) + \int_0^t \bm{K}(t,s)\,\D\bm{W}_s,
\end{align}
where $\bm{g}_0:[0,T]\to\mathbb{R}^{d\times d}$ is a deterministic initial curve and $\bm{K}:[0,T]^2\to\mathbb{R}^{d\times d}$ is a deterministic Volterra kernel. Crucially, allowing $\bm{K}(t,s)$ to exhibit singularities on the diagonal $t=s$ (e.g., fractional kernels $K(t,s) \propto (t-s)^{H-1/2}$ for Hurst index $H < 1/2$) natively captures the rough behavior of instantaneous covariance. As in the quadratic case, the asset dynamics are given by $\D\bm{S}_t = \diag(\bm{S}_t)\,\bm{X}_t\,\D\bm{B}_t$, ensuring positive semidefinite instantaneous covariance $\bm{\Sigma}_t = \bm{X}_t \bm{X}_t^\top$. 

\paragraph{\textbf{Introducing leverage.}} Within the continuous diffusion setting above, leverage and, more generally, intrinsic market incompleteness are incorporated by allowing return shocks to be correlated with covariance shocks. To this end, define the linear map
\begin{equation}
\label{eq:VW-R-map-clean}
    \mathcal R:\mathbb R^{d\times d}\to\mathbb R^d,
    \qquad
    \mathcal R(\bm U)
    :=
    \bigl(
    \Tr(\bm U\bm\rho_1^\top),\dots,\Tr(\bm U\bm\rho_d^\top)
    \bigr)^\top,
\end{equation}
and its Frobenius adjoint
\begin{equation}
\label{eq:VW-R-adjoint-clean}
    \mathcal R^\ast:\mathbb R^d\to\mathbb R^{d\times d},
    \qquad
    \mathcal R^\ast(\bm a)
    =
    \sum_{j=1}^d a_j \bm\rho_j.
\end{equation}
Next, introduce the Gram matrix
\begin{equation}
\label{eq:VW-Q-clean}
    \bm Q_{ij}:=\Tr(\bm\rho_i\bm\rho_j^\top),
    \qquad i,j=1,\dots,d,
\end{equation}
and assume $\bm Q\preceq \bm I_d$. Let $\bm\Lambda\in\mathbb R^{d\times d}$ satisfy $\bm\Lambda\bm\Lambda^\top=\bm I_d-\bm Q$, let $\bm W^\perp$ be a $d$-dimensional Brownian motion independent of $\bm W$, and define
\begin{equation}
\label{eq:VW-B-clean}
    \D\bm B_t
    =
    \mathcal R(\D\bm W_t)+\bm\Lambda\,\D\bm W_t^\perp.
\end{equation}
This construction yields a $d$-dimensional Brownian motion $\bm B$ whose instantaneous covariation with $\bm W$ is encoded by the family $(\bm\rho_1,\dots,\bm\rho_d)$, and therefore provides a convenient and analytically tractable way to model leverage.

A particularly convenient special case which retains the linear-affinity of the factor model under the affine stochastic covariance models, as shown in \cite{Fonseca2007}, is the vector correlation specification
\begin{align}\label{eq:corr}
    \D\bm{B}_t &= \D\bm{W}_t\,\bm{\rho} + \sqrt{1-\bm{\rho}^\top\bm{\rho}}\,\D\bm{Z}_t,
\end{align}
where the constant vector $\bm{\rho}\in\mathbb{R}^d$ satisfies $\bm{\rho}^\top\bm{\rho}\le 1$, and $\bm{Z}$ is a $d$-dimensional Brownian motion independent of $\bm{W}$.

\subsection{The Semi-Static Variance-Optimal Hedging Problem}

Let $H\in L^2(\mathcal{F}_T,\mathbb{Q})$ denote the target payoff of the contingent claim, where we assume without loss of generality that discounting has been absorbed into the definition of $H$ and the asset prices. Furthermore, let $\bm{\eta}=(\eta^1,\dots,\eta^n)^\top\in \big(L^2(\mathcal{F}_T,\mathbb{Q})\big)^n$ represent the payoffs of $n\in\mathbb{N}$ auxiliary claims, which are restricted to be held \emph{statically} from inception $t=0$ to maturity $T$. 

A semi-static trading strategy is characterized by an initial capital $c\in\mathbb{R}$, a vector of static weights $\bm{u}\in\mathbb{R}^n$ allocated to the auxiliary claims, and a predictable, dynamic trading strategy $\bm{\theta}$ in the underlying risky assets. The objective is to minimize the mean-squared hedging error (MSHE) under the martingale measure $\mathbb{Q}$:
\begin{equation}\label{eq:fulloptimization}
    \epsilon^2 = \min_{c\in\mathbb{R},\, \bm{u}\in\mathbb{R}^n,\, \bm{\theta}\in\Theta_S}
    \mathbb{E}\bigg[\Big(c + \bm{u}^\top(\bm{\eta}-\mathbb{E}[\bm{\eta}]) + \int_0^T \bm{\theta}_t^\top \D\bm{S}_t - H\Big)^2\bigg].
\end{equation}
Note that centering the auxiliary payoffs $\bm{\eta}$ by their expectation is a standard normalization that entails no loss of generality, serving to mathematically decouple the optimal initial capital from the static weights in the optimization. The space of admissible dynamic strategies is defined as the set of square-integrable integrands with respect to $\bm{S}$:
\begin{equation}\label{eq:admissible}
    \Theta_S := L^2(\bm{S})
    = \bigg\{\bm{\theta} \text{ predictable }\, \mathbb{R}^{d}\text{-valued process} \colon \mathbb{E}\bigg[\int_0^T \bm{\theta}_t^\top\,\D\llangle \bm{S},\bm{S}\rrangle_t\,\bm{\theta}_t\bigg] < \infty\bigg\}.
\end{equation}

\subsection{Galtchouk-Kunita-Watanabe Decomposition}

To solve the semi-static problem, we first recall the Galtchouk-Kunita-Watanabe (GKW) decomposition in the purely dynamic setting (i.e., $\bm{u} = \bm{0} \in \mathbb{R}^n$). Let $H_t:=\mathbb{E}[H\mid\mathcal{F}_t] \in \mathbb{R}$ denote the scalar intrinsic value process of the claim. A square-integrable payoff $H$ admits a GKW decomposition with respect to the $d$-dimensional underlying asset process $\bm{S}$ if there exists a constant $H_0\in\mathbb{R}$, an admissible dynamic strategy $\bm{\theta}^H\in\Theta_S$ taking values in $\mathbb{R}^d$, and a square-integrable, scalar martingale $L^H$ strongly orthogonal to $\bm{S}$, such that
\begin{equation}\label{eq:GKW}
    H = H_0 + \int_0^T (\bm{\theta}_t^H)^\top \D\bm{S}_t + L_T^H.
\end{equation}

The optimal $\mathbb{R}^d$-valued integrand $\bm{\theta}^H$ is uniquely characterized by the projection identity arising from the strong orthogonality condition $L^H \perp \bm{S}$. To formulate this rigorously for general semimartingales, let $A = (A_t)_{t\in[0,T]}$ be a scalar, predictable, increasing process dominating the predictable covariations (for instance, $A_t = t + \sum_{i=1}^d \langle S^i, S^i \rangle_t$). By the Kunita-Watanabe inequality, the matrix-valued predictable covariation of the asset prices and the vector-valued predictable covariation between the underlying assets and the claim admit Radon-Nikodym densities with respect to $A$:
\begin{equation}\label{eq:proj}
    \D\llangle \bm{S},\bm{S}\rrangle_t = \bm{c}_t^{\bm{S}} \,\D A_t, \qquad \D\langle \bm{S},H\rangle_t = \bm{c}_t^{\bm{S},H} \,\D A_t,
\end{equation}
where $\bm{c}^{\bm{S}}$ takes values in $\mathbb{S}_+^d$ and $\bm{c}^{\bm{S},H}$ takes values in $\mathbb{R}^d$. The projection identity then yields the vector equality $\bm{c}_t^{\bm{S},H} = \bm{c}_t^{\bm{S}} \bm{\theta}_t^H$ almost everywhere with respect to the measure $\mathbb{Q} \times \D A$. Consequently, the variance-minimal dynamic hedge is given by:
\begin{equation}\label{eq:OptimalStrategy}
    \bm{\theta}_t^* = \big(\bm{c}_t^{\bm{S}}\big)^{\dagger} \bm{c}_t^{\bm{S},H} \in \mathbb{R}^d,
\end{equation}
where $\dagger$ denotes the Moore-Penrose pseudoinverse, effectively accommodating potential rank deficiencies or singularities in the $d \times d$ covariance matrix of the multidimensional asset price process.

\subsubsection*{Decomposition of the Semi-Static Hedging Problem}

We now extend this machinery to incorporate the static instruments. For notational convenience, we set $\eta^0:=H$ and define the associated scalar martingale value processes $H_t^i:=\mathbb{E}[\eta^i\mid\mathcal{F}_t] \in \mathbb{R}$ for $i=0,\dots,n$. Applying the GKW decomposition \eqref{eq:GKW} individually to each payoff $\eta^i$ yields a system of orthogonal decompositions:
\begin{equation}\label{eq:SystemGKW}
    H_t^i = H_0^i + \int_0^t (\bm{\theta}_s^i)^\top \D\bm{S}_s + L_t^i,\qquad i=0,\dots,n,
\end{equation}
where each scalar residual martingale satisfies $L^i \perp \bm{S}$. The corresponding $d$-dimensional dynamic hedge ratios are determined via
\begin{equation}\label{eq:OptimalStrategyAll}
    \bm{\theta}_t^i = \big(\D\llangle \bm{S},\bm{S}\rrangle_t\big)^{\dagger}\,\D\llangle \bm{S},H^i\rrangle_t \in \mathbb{R}^d.
\end{equation}

We aggregate the auxiliary hedge ratios into a time-dependent weighting matrix $\bm{\Theta}_t \in \mathbb{R}^{d\times n}$ and collect the strongly orthogonal scalar residuals into an $n$-dimensional vector $\bm{L}_t \in \mathbb{R}^n$:
\begin{equation}\label{eq:ThetaDef}
    \bm{\Theta}_t := \big[\bm{\theta}_t^1, \bm{\theta}_t^2, \dots, \bm{\theta}_t^n\big] \in \mathbb{R}^{d \times n}, \qquad \bm{L}_t := (L_t^1,\dots,L_t^n)^\top \in \mathbb{R}^n.
\end{equation}
It is crucial to note that even if an individual auxiliary payoff $\eta^i$ is structurally dependent on only a single asset, its variance-optimal hedge $\bm{\theta}^i$ is generally $d$-dimensional. This cross-hedging phenomenon occurs because the orthogonal projection in \eqref{eq:proj} is mediated by the full, $d \times d$ matrix-valued covariation process $\D\llangle \bm{S},\bm{S}\rrangle$, capturing the interconnected correlation dynamics between all market assets.

\vspace{0.2cm}
The general multivariate Galtchouk--Kunita--Watanabe projection in \eqref{eq:OptimalStrategyAll} permits cross-hedging across all traded assets, even when a given auxiliary claim $\eta^k$ depends only on a single component $S^m$. In some of the applications below, however, we intentionally impose the economically more restrictive convention that such a univariate claim is dynamically hedgeable only through its own underlying. In that case, the relevant projection space is no longer the full space
\begin{equation}
    \mathcal G
    :=
    \left\{
        \int_0^T \bm{\theta}_t^\top \D\bm{S}_t
        :
        \bm{\theta}\in\Theta_S
    \right\},
\end{equation}
but the one-dimensional closed subspace
\begin{equation}
    \mathcal G^{m}
    :=
    \left\{
        \int_0^T \vartheta_t\,\D S_t^m
        :
        \vartheta\in L^2(S^m)
    \right\}
    \subseteq \mathcal G.
\end{equation}
Equivalently, the admissible integrands are restricted to
\begin{equation}
    \Theta_{S^m}
    :=
    \left\{
        \vartheta\,\ve_m
        :
        \vartheta\in L^2(S^m)
    \right\}
    \subseteq \Theta_S,
\end{equation}
where $\ve_m\in\mathbb{R}^d$ denotes the $m$-th canonical basis vector. Writing, as in \eqref{eq:proj},
\begin{equation}
    \D\llangle \bm{S},\bm{S}\rrangle_t
    =
    \bm{c}_t^{\bm{S}}\,\D A_t,
    \qquad
    \D\langle \bm{S},H\rangle_t
    =
    \bm{c}_t^{\bm{S},H}\,\D A_t,
\end{equation}
the constrained variance-optimal hedge of a square-integrable claim $H$ is the unique element of $\Theta_{S^m}$ of the form
\begin{equation}
    \bm{\theta}_t^{H,(m)}
    =
    \vartheta_t^{H,(m)}\ve_m,
\end{equation}
with
\begin{equation}
    \vartheta_t^{H,(m)}
    =
    \bigl(\ve_m^\top \bm{c}_t^{\bm{S}}\ve_m\bigr)^\dagger
    \ve_m^\top \bm{c}_t^{\bm{S},H},
\end{equation}
that is,
\begin{equation}
    \bm{\theta}_t^{H,(m)}
    =
    \ve_m
    \bigl(\ve_m^\top \bm{c}_t^{\bm{S}}\ve_m\bigr)^\dagger
    \ve_m^\top \bm{c}_t^{\bm{S},H}.
\end{equation}

\subsection{Solution to the Semi-Static Hedging Problem}

Following the structural approach of \cite{Semi_static_Fourier}, we observe that the joint optimization problem mathematically separates. We decompose the global problem into an inner dynamic projection (for a given, fixed vector of static weights $\bm{u} \in \mathbb{R}^n$) and an outer, finite-dimensional quadratic minimization over those weights:
\begin{align}\label{eq:two-stage}
    \epsilon^2(\bm{u}) &= \min_{\bm{\theta}\in\Theta_S,\, c\in\mathbb{R}} \mathbb{E}\bigg[\Big(c - \bm{u}^\top\mathbb{E}[\bm{\eta}] + \int_0^T \bm{\theta}_t^\top \D\bm{S}_t - (H - \bm{u}^\top\bm{\eta})\Big)^2\bigg], \nonumber \\[1ex]
    \epsilon_*^2 &= \min_{\bm{u}\in\mathbb{R}^n} \epsilon^2(\bm{u}).
\end{align}
By isolating the static allocation from the continuous-time dynamic tracking, \eqref{eq:two-stage} reduces the intrinsically high-dimensional backward stochastic differential equation (BSDE) control problem into a highly tractable two-stage geometry.

\subsubsection*{Inner Optimization: Optimal Dynamic Strategy}

Fix $\bm{u}\in\mathbb{R}^n$ and define the static residual claim as
\begin{align}\label{eq:residual-claim}
    H^{\bm{u}} &:= H - \bm{u}^\top\bm{\eta}.
\end{align}
To analyze this residual, we first apply the GKW decomposition to the vector of auxiliary claims $\bm{\eta}$. There exists an initial value vector $\bm{\eta}_0 \in \mathbb{R}^n$, an admissible $d \times n$ matrix-valued dynamic strategy $\bm{\Theta}^{\bm{\eta}}$, and an $n$-dimensional vector of square-integrable martingales $\bm{L}^{\bm{\eta}}$ strongly orthogonal to $\bm{S}$, such that $\bm{\eta} = \bm{\eta}_0 + \int_0^T (\bm{\Theta}_t^{\bm{\eta}})^\top \D\bm{S}_t + \bm{L}_T^{\bm{\eta}}$. 

By the linearity of the conditional expectation operator, the intrinsic value process of the residual claim is given by
\begin{align}\label{eq:residual-GKW}
    H_t^{\bm{u}} &= \big(H_0 - \bm{u}^\top \bm{\eta}_0\big) + \int_0^t \big(\bm{\theta}_s^H - \bm{\Theta}_s^{\bm{\eta}}\bm{u}\big)^\top \D\bm{S}_s + \big(L_t^H - \bm{u}^\top \bm{L}_t^{\bm{\eta}}\big).
\end{align}
We introduce the residual orthogonal martingale $L^{\bm{u}} := L^H - \bm{u}^\top\bm{L}^{\bm{\eta}}$, which strictly maintains strong orthogonality to $\bm{S}$. Consequently, for any candidate dynamic strategy $\bm{\theta} \in \Theta_S$, the strong orthogonality allows us to apply the Pythagorean theorem in $L^2(\mathcal{F}_T, \mathbb{Q})$ to obtain
\begin{align}\label{eq:pythagoras}
    \mathbb{E}\bigg[\Big(\int_0^T \bm{\theta}_t^\top \D\bm{S}_t - (H_T^{\bm{u}} - H_0^{\bm{u}})\Big)^2\bigg] &= \mathbb{E}\bigg[\Big(\int_0^T (\bm{\theta}_t - \bm{\theta}_t^{\bm{u}})^\top \D\bm{S}_t\Big)^2\bigg] + \mathbb{E}\big[(L_T^{\bm{u}})^2\big],
\end{align}
where $\bm{\theta}^{\bm{u}} := \bm{\theta}^H - \bm{\Theta}^{\bm{\eta}}\bm{u} \in \mathbb{R}^d$ represents the dynamically adjusted hedge ratio. 

The first term on the right-hand side of \eqref{eq:pythagoras} is uniquely minimized to zero by selecting the optimal dynamic tracking strategy $\bm{\theta}^* = \bm{\theta}^{\bm{u}}$, while the optimal initial capital evaluates to $c^*(\bm{u}) = H_0 - \bm{u}^\top(\bm{\eta}_0 - \mathbb{E}[\bm{\eta}])$. The corresponding minimal mean-squared hedging error for a fixed static portfolio $\bm{u}$ is therefore isolated entirely to the variance of the orthogonal residual:
\begin{align}\label{eq:inner-error-variance}
    \epsilon^2(\bm{u}) &= \mathbb{E}\big[(L_T^{\bm{u}})^2\big] = \mathbb{E}\bigg[\Big(L_T^H - \bm{u}^\top \bm{L}_T^{\bm{\eta}}\Big)^2\bigg].
\end{align}
To operationalize this, we define the following structural covariance components of the orthogonal martingales:
\begin{align}\label{eq:maincomponenetsstatic}
    A &:= \mathbb{E}[(L_T^H)^2] \in \mathbb{R}, \nonumber \\
    B &:= \mathbb{E}[\bm{L}_T^{\bm{\eta}} L_T^H] \in \mathbb{R}^n, \\
    C &:= \mathbb{E}[\bm{L}_T^{\bm{\eta}} (\bm{L}_T^{\bm{\eta}})^\top] \in \mathbb{S}_+^n. \nonumber
\end{align}
Expanding \eqref{eq:inner-error-variance} yields a strictly quadratic form in $\bm{u}$:
\begin{align}\label{eq:quadratic-form}
    \epsilon^2(\bm{u}) &= \bm{u}^\top C \bm{u} - 2\bm{u}^\top B + A.
\end{align}

\subsubsection*{Outer Optimization: Optimal Static Weights}

Minimizing the quadratic form in \eqref{eq:quadratic-form} with respect to $\bm{u}$ yields the first-order condition, which takes the form of the standard normal equations:
\begin{align}\label{eq:normal-eq}
    C\,\bm{u}^* &= B.
\end{align}
If the covariance matrix $C$ is strictly positive definite (and thus invertible), the unique optimal static allocation is $\bm{u}^* = C^{-1}B$. In the general case, redundant static instruments or collinearity in the residual risk may render $C$ singular. To ensure a robust, well-posed solution, we select the minimum-norm static allocation via the Moore-Penrose pseudoinverse:
\begin{align}\label{eq:nustar}
    \bm{u}^* &= C^{\dagger}B \in \mathbb{R}^n.
\end{align}
Substituting \eqref{eq:nustar} back into the objective \eqref{eq:quadratic-form}, we obtain the minimal global variance-optimal hedging error:
\begin{align}\label{eq:epsstar}
    \epsilon_*^2 &= A - B^\top C^{\dagger}B.
\end{align}

The singular case is well posed because $B\in\operatorname{Range}(C)$. Indeed, if $x\in\ker(C)$, then
\[
0=x^\top Cx=\mathbb E\big[(x^\top \bm L_T^{\bm\eta})^2\big],
\]
hence $x^\top \bm L_T^{\bm\eta}=0$ in $L^2(\mathcal F_T,\mathbb Q)$ and therefore
\[
x^\top B=\mathbb E\big[(x^\top \bm L_T^{\bm\eta})L_T^H\big]=0.
\]
Thus $\ker(C)\subseteq B^\perp$, equivalently $B\in(\ker C)^\perp=\operatorname{Range}(C)$, so the quadratic problem admits minimizers and its minimum-norm minimizer is exactly $C^\dagger B$.

\section{Multidimensional Spanning and the Static Replication of Covariance Risk}
\label{sec:replicationtheory}

This section complements the abstract semi-static variance-optimal framework developed in Section~\ref{sec:VO-framework} by addressing the corresponding \emph{instrument-selection} problem, namely the identification of static claims that constitute natural candidates for mitigating the residual unhedgeable risk of multi-asset contingent claims. In the one-dimensional setting, the inverse problem of replicating a sufficiently regular payoff by a static strip of vanilla options is classical; see, e.g., \cite{Breeden,Carr}. More precisely, under appropriate differentiability and integrability assumptions, a payoff $h:\mathbb{R}_+\to\mathbb{R}$ admits a representation in terms of a continuum of calls indexed by the strike variable $K\geq 0$, namely
\begin{equation}
\label{eq:spanning-formula}
    h(x)\stackrel{\text{a.e.}}{=}\int_0^\infty \varphi(K)(x-K)^+\,\mathrm{d}K, \qquad x\in\mathbb{R}_+.
\end{equation}
In analytical terms, this means that the ramp family $\bigl(x\mapsto (x-K)^+\bigr)_{K\in\mathbb{R}_+}$ spans a broad class of admissible payoff functions, up to a Lebesgue-null set. Formally differentiating \eqref{eq:spanning-formula} twice shows that the corresponding density is $\varphi(K)=h''(K)$. In particular, fixing an anchor $\alpha\geq 0$ and assuming that $h'$ is absolutely continuous on $\mathbb{R}_+$, so that $h''$ exists almost everywhere and is locally integrable, one obtains the refined Carr--Madan-type decomposition
\begin{equation}
\label{eq:replication1d-refined}
    h(x)
    = h(\alpha)+h'(\alpha)(x-\alpha)
    +\int_0^\alpha h''(K)(K-x)^+\,\mathrm{d}K
    +\int_\alpha^\infty h''(K)(x-K)^+\,\mathrm{d}K,
\end{equation}
valid for every $x\geq 0$. Equivalently, introducing the kernel
\begin{equation}
\label{eq:psi-1d}
    \psi_{\alpha}(x,K)
    :=
    \bigl(x-\alpha+(\alpha-K)+2(K-x)\mathbbm{1}_{\{K\leq \alpha\}}\bigr)^+,
    \qquad (x,K)\in\mathbb{R}_+\times\mathbb{R}_+,
\end{equation}
so that $\psi_\alpha(x,K)=(K-x)^+$ for $K\leq \alpha$ and $\psi_\alpha(x,K)=(x-K)^+$ for $K\geq \alpha$, the same identity can be written in the compact form
\begin{equation}
\label{eq:replication1d-psi}
    h(x)=h(\alpha)+h'(\alpha)(x-\alpha)+\int_0^\infty h''(K)\psi_\alpha(x,K)\,\mathrm{d}K.
\end{equation}
This representation is naturally interpreted as a fixed-point identity generated by a linear spanning operator: for each $\alpha\geq 0$, define
\begin{equation}
\label{eq:operator-1d}
    (\mathcal{S}_{\alpha}h)(x)
    :=
    h(\alpha)+h'(\alpha)(x-\alpha)+\int_0^\infty h''(K)\psi_\alpha(x,K)\,\mathrm{d}K,
\end{equation}
so that \eqref{eq:replication1d-psi} is precisely the operator identity $h=\mathcal{S}_{\alpha}h$ on $\mathbb{R}_+$.

In the multivariate stochastic covariance setting, the relevant source of incompleteness is no longer purely marginal convexity, but rather the joint law of the terminal asset vector and, in particular, the mixed sensitivity carried by cross-variation terms. This structural feature strongly motivates the use of static instruments whose payoffs load directly on cross-asset interactions, or more generally factorize across coordinates. In practice, such payoffs appear as quanto or product options, as well as spread- and basket-type claims. Our framework therefore relies on multidimensional spanning formulas. Related multivariate extensions of the Carr--Madan construction have already been developed in \cite{Cui04052022,Madan2021Pricing}, where sufficiently regular multi-asset payoffs are represented by multiple integrals of products of calls and puts. An alternative high-dimensional approach based on integral equations and inverse Radon transforms is given in \cite{Bossu2021}, where the spanning family is expressed in terms of basket options. 

Here we follow the product-option approach, which is particularly well aligned with the covariance-sensitive structures studied later in this paper. To formalize the multidimensional extension, fix $\bm a=(a_1,\dots,a_n)\in\mathbb{R}_+^n$, and for each coordinate $i\in\{1,\dots,n\}$ define the lifted one-dimensional spanning operator $\mathcal{S}^{(i)}_{a_i}$ acting on $h:\mathbb{R}_+^n\to\mathbb{R}$ by
\begin{equation}
\label{eq:operator-lift}
    (\mathcal{S}^{(i)}_{a_i}h)(\bm x)
    :=
    h(\bm x)\big|_{x_i=a_i}
    +\partial_i h(\bm x)\big|_{x_i=a_i}(x_i-a_i)
    +\int_0^\infty \partial_{ii}h(\bm x)\big|_{x_i=K_i}\,\psi_{a_i}(x_i,K_i)\,\mathrm{d}K_i,
\end{equation}
where the notation $\big|_{x_i=\xi}$ means that the $i$-th coordinate of $\bm x$ is replaced by $\xi$, while all remaining coordinates are kept fixed. Under the regularity assumptions stated below, each $\mathcal{S}^{(i)}_{a_i}$ is a well-defined linear operator on the relevant class of payoff functions, and the full spanning identity is obtained by coordinate-wise iteration:
\begin{equation}
\label{eq:operator-factorization}
    h(\bm x)
    =
    \Bigl(\mathcal{S}^{(1)}_{a_1}\circ\mathcal{S}^{(2)}_{a_2}\circ\cdots\circ\mathcal{S}^{(n)}_{a_n}\Bigr)h(\bm x),
    \qquad \bm x\in\mathbb{R}_+^n.
\end{equation}
Moreover, since mixed partial derivatives commute under the imposed smoothness assumptions, the operators $\mathcal{S}^{(i)}_{a_i}$ commute pairwise, so the order of composition in \eqref{eq:operator-factorization} is immaterial. To operationalize this for our static hedging framework, we rely on the following multidimensional spanning formulation:

\begin{proposition}[Multidimensional Spanning, cf.\ \cite{Cui04052022}]
\label{n-dimensionalspanning}
Let $h:\mathbb{R}_+^n\to\mathbb{R}$ and fix a reference point $\bm a=(a_1,\dots,a_n)\in\mathbb{R}_+^n$. Assume that $h$ admits continuous partial derivatives $\partial^\alpha h$ for all multi-indices $\alpha=(\alpha_1,\dots,\alpha_n)$ in
\begin{equation*}
    J:=\bigl\{\alpha:\alpha_i\in\{0,1,2\}\ \text{for all } i=1,\dots,n\bigr\},
\end{equation*}
and satisfies appropriate growth and decay conditions at infinity such that the required infinite-strike integrals are well-defined (see \cite{Cui04052022} for explicit integrability criteria). For $\bm x=(x_1,\dots,x_n)\in\mathbb{R}_+^n$ and $\alpha\in J$, define
\begin{equation*}
    \partial^\alpha:=\partial_1^{\alpha_1}\cdots\partial_n^{\alpha_n},
    \qquad
    \partial_i^{\alpha_i}:=\frac{\partial^{\alpha_i}}{\partial x_i^{\alpha_i}},
    \qquad
    \bm x^\alpha:=\prod_{i=1}^n x_i^{\alpha_i}.
\end{equation*}
Define the indicator maps $L(\alpha)\in\{0,1\}^n$ and $I(\alpha)\in\{0,1,+\}^n$ componentwise by
\begin{equation}
\label{eq:LI-def}
    L(\alpha_i):=
    \begin{cases}
        1, & \alpha_i=2,\\
        0, & \alpha_i\neq 2,
    \end{cases}
    \qquad
    I(\alpha_i):=
    \begin{cases}
        0, & \alpha_i=0,\\
        1, & \alpha_i=1,\\
        +, & \alpha_i=2,
    \end{cases}
\end{equation}
and let $\bm K=(K_1,\dots,K_n)\in\mathbb{R}_+^n$. Then, with $\odot$ denoting the Hadamard product and $\mathbbm{1}_{\bm K\leq \bm a}$ the componentwise indicator, one has the spanning identity
\begin{equation}
\label{eq:spanning-nd-refined}
    h(\bm x)
    =
    \sum_{\alpha\in J}
    \int_{\mathbb{R}_+^n}
    \left(\partial^\alpha h\right)(\bm u)\Big|_{\bm u=\bm a+(\bm K-\bm a)\odot L(\alpha)}
    \!
    \left[
        \bm x-\bm a
        +
        \left(
            \bm a-\bm K
            +
            2(\bm K-\bm x)\odot\mathbbm{1}_{\bm K\leq \bm a}
        \right)\odot L(\alpha)
    \right]^{I(\alpha)}
    \!
    \left[\mathrm{d}\bm K\right]^{L(\alpha)}.
\end{equation}
By convention, $\left[\mathrm{d}K_i\right]^{L(\alpha_i)}=1$ whenever $L(\alpha_i)=0$, that is, no integration is taken in coordinate $i$ unless $\alpha_i=2$. Moreover, for $\bm z\in\mathbb{R}^n$, the expression $\bm z^{I(\alpha)}$ denotes componentwise application of $I(\alpha)$, namely $\prod_{i:\alpha_i=1} z_i \prod_{i:\alpha_i=2} z_i^+$, with the convention that coordinates with $\alpha_i=0$ contribute the factor $1$.
\end{proposition}

For $n=2$, let $\bm a=(\alpha_1,\alpha_2)\in\mathbb{R}_+^2$ and $\bm x=(x_1,x_2)\in\mathbb{R}_+^2$. Then $J=\{0,1,2\}^2$ contains $9$ multi-indices, and \eqref{eq:spanning-nd-refined} expands into the explicit bivariate representation
\begin{align}
\label{functionalreplication2d}
    h(x_1,x_2)
    &=
    h(\alpha_1,\alpha_2)
    + h_{x_1}(\alpha_1,\alpha_2)(x_1-\alpha_1)
    + h_{x_2}(\alpha_1,\alpha_2)(x_2-\alpha_2)
    + h_{x_1x_2}(\alpha_1,\alpha_2)(x_1-\alpha_1)(x_2-\alpha_2) \nonumber\\
    &\quad
    + \int_{0}^{\alpha_2} h_{x_2x_2}(\alpha_1,K_2)(K_2-x_2)^+\,\mathrm{d}K_2
    + \int_{\alpha_2}^{\infty} h_{x_2x_2}(\alpha_1,K_2)(x_2-K_2)^+\,\mathrm{d}K_2 \nonumber\\
    &\quad
    + \int_{0}^{\alpha_2} h_{x_1x_2x_2}(\alpha_1,K_2)(x_1-\alpha_1)(K_2-x_2)^+\,\mathrm{d}K_2 \nonumber\\
    &\quad
    + \int_{\alpha_2}^{\infty} h_{x_1x_2x_2}(\alpha_1,K_2)(x_1-\alpha_1)(x_2-K_2)^+\,\mathrm{d}K_2 \nonumber\\
    &\quad
    + \int_{0}^{\alpha_1} h_{x_1x_1}(K_1,\alpha_2)(K_1-x_1)^+\,\mathrm{d}K_1
    + \int_{\alpha_1}^{\infty} h_{x_1x_1}(K_1,\alpha_2)(x_1-K_1)^+\,\mathrm{d}K_1 \nonumber\\
    &\quad
    + \int_{0}^{\alpha_1} h_{x_1x_1x_2}(K_1,\alpha_2)(x_2-\alpha_2)(K_1-x_1)^+\,\mathrm{d}K_1 \nonumber\\
    &\quad
    + \int_{\alpha_1}^{\infty} h_{x_1x_1x_2}(K_1,\alpha_2)(x_2-\alpha_2)(x_1-K_1)^+\,\mathrm{d}K_1 \nonumber\\
    &\quad
    + \int_{0}^{\alpha_1}\int_{0}^{\alpha_2}
        h_{x_1x_1x_2x_2}(K_1,K_2)(K_1-x_1)^+(K_2-x_2)^+\,\mathrm{d}K_2\,\mathrm{d}K_1 \nonumber\\
    &\quad
    + \int_{0}^{\alpha_1}\int_{\alpha_2}^{\infty}
        h_{x_1x_1x_2x_2}(K_1,K_2)(K_1-x_1)^+(x_2-K_2)^+\,\mathrm{d}K_2\,\mathrm{d}K_1 \nonumber\\
    &\quad
    + \int_{\alpha_1}^{\infty}\int_{0}^{\alpha_2}
        h_{x_1x_1x_2x_2}(K_1,K_2)(x_1-K_1)^+(K_2-x_2)^+\,\mathrm{d}K_2\,\mathrm{d}K_1 \nonumber\\
    &\quad
    + \int_{\alpha_1}^{\infty}\int_{\alpha_2}^{\infty}
        h_{x_1x_1x_2x_2}(K_1,K_2)(x_1-K_1)^+(x_2-K_2)^+\,\mathrm{d}K_2\,\mathrm{d}K_1.
\end{align}

Equation \eqref{functionalreplication2d} shows explicitly that the static spanning of a generic bivariate payoff requires not only marginal vanilla options, but also genuinely mixed building blocks. A canonical and particularly convenient choice for these mixed terms is the family of product options, such as $(x_1-K_1)^+(x_2-K_2)^+$, together with their corresponding out-of-the-money quadrant variants. While marginal vanillas are typically liquid in equity-index markets, genuine product options are often traded OTC. This distinction is economically important: without such cross-asset instruments, the market remains incomplete with respect to the joint risk-neutral law of the asset vector. By contrast, in markets such as energy and weather derivatives, quanto and product structures are standard instruments for managing interconnected price risks; see, for example, \cite{BenthLangeQuanto,Alfonsi02112023}.

Our aim is now to exploit these spanning identities in order to derive replication formulas for covariance swaps, thereby extending the classical variance-swap representations of \cite{Breeden,Carr} to the multivariate setting. We note that here we focus on the continuous multivariate semi-martingale case. However one could extend these results to the multivariate jump diffusion setting following the results of \cite{robustreplicationjumps}.

\subsection{Covariance Swaps}\label{sec:covariance-swaps}

\begin{example}[Geometric Covariance Swap Replication via Product Options]
\label{cov-swapreplicationQuanto}
Fix $T>0$ and a valuation time $t\in[0,T)$. For assets $i,j\in\{1,\dots,d\}$, set $Y_u^k:=\log S_u^k$ for $u\in[t,T]$. The realized geometric covariance is defined as the quadratic covariation of the log-returns,
\begin{align}\label{eq:geo-qcov-def}
    \langle Y^{i},Y^{j}\rangle_{t,T}
    :=
    p\text{-}\lim_{|\pi|\to 0}
    \sum_{m=0}^{n-1}
    \bigl(Y^i_{t_{m+1}}-Y^i_{t_m}\bigr)
    \bigl(Y^j_{t_{m+1}}-Y^j_{t_m}\bigr),
\end{align}
and the corresponding continuously monitored geometric covariance swap payoff is
\begin{align}\label{eq:geo-covswap-payoff}
    H_T^{\mathrm{g\text{-}cov}}
    :=
    \langle Y^{i},Y^{j}\rangle_{t,T}-K_{\mathrm{var}},
    \qquad
    K_{\mathrm{var}}=\mathbb{E}_t^{\mathbb{Q}}\big[\langle Y^{i},Y^{j}\rangle_{t,T}\big].
\end{align}

For the identities below we assume that $S^i$ and $S^j$ are strictly positive \emph{continuous} semimartingales.

\subsubsection*{Step 1: Pathwise Decomposition}
Applying integration by parts to the product $Y^iY^j$ yields
\begin{align}\label{eq:geo-qcov-product-rule}
    \langle Y^i,Y^j\rangle_{t,T}
    =
    \bigl(Y_T^iY_T^j-Y_t^iY_t^j\bigr)
    -\int_t^T Y_u^i\,\D Y_u^j
    -\int_t^T Y_u^j\,\D Y_u^i.
\end{align}
Using Itô's formula in the form

\begin{equation*}
    \D Y_u^k=\frac{\D S_u^k}{S_u^k}-\frac{1}{2(S_u^k)^2}\,\D\langle S^k\rangle_u,
\end{equation*}
one obtains
\begin{align}\label{eq:geo-qcov-decomp-dS}
    \langle Y^i,Y^j\rangle_{t,T}
    &=
    \underbrace{Y_T^iY_T^j-Y_t^iY_t^j}_{\text{terminal term}}
    -\underbrace{\int_t^T \frac{Y_u^i}{S_u^j}\,\D S_u^j+\int_t^T \frac{Y_u^j}{S_u^i}\,\D S_u^i}_{\text{dynamic trading}} \nonumber\\
    &\quad
    +\underbrace{\frac12\int_t^T \frac{Y_u^i}{(S_u^j)^2}\,\D\langle S^j\rangle_u
    +\frac12\int_t^T \frac{Y_u^j}{(S_u^i)^2}\,\D\langle S^i\rangle_u}_{\text{weighted variance terms}}.
\end{align}
The stochastic integrals correspond to continuous dynamic trading strategies $\phi_u^i=-Y_u^j/S_u^i$ and $\phi_u^j=-Y_u^i/S_u^j$. The remaining finite-variation terms represent weighted variance exposure and, in an equity-only implementation, generally remain unhedgeable unless additional variance-linked instruments are available.

\subsubsection*{Step 2: Static Replication of the Terminal Product}
To replicate the terminal term in \eqref{eq:geo-qcov-decomp-dS}, it is necessary to span the cross-product $Y_T^iY_T^j$. Set $(\alpha,\beta):=(S_t^i,S_t^j)$, and define the convexity error function
\begin{align}\label{eq:log-kernel-def}
    F(z,z_0)
    &:=
    \frac{z-z_0}{z_0}-\log\Big(\frac{z}{z_0}\Big)
    =
    \int_0^{z_0}\frac{(K-z)^+}{K^2}\,\D K
    +\int_{z_0}^{\infty}\frac{(z-K)^+}{K^2}\,\D K.
\end{align}
Then the log-return increment admits the pathwise decomposition
\begin{align}\label{eq:deltaY-decomp}
    \Delta Y_T^k
    :=
    Y_T^k-Y_t^k
    =
    \log\Big(\frac{S_T^k}{S_t^k}\Big)
    =
    \frac{S_T^k-S_t^k}{S_t^k}-F(S_T^k,S_t^k),
    \qquad k\in\{i,j\}.
\end{align}
Writing $\Delta S_T^i:=S_T^i-\alpha$ and $\Delta S_T^j:=S_T^j-\beta$, one has
\begin{align}\label{eq:terminal-product-basic}
    Y_T^iY_T^j-Y_t^iY_t^j
    =
    Y_t^i\,\Delta Y_T^j
    +Y_t^j\,\Delta Y_T^i
    +\Delta Y_T^i\,\Delta Y_T^j.
\end{align}
Substituting \eqref{eq:deltaY-decomp} into \eqref{eq:terminal-product-basic} yields
\begin{align}\label{eq:terminal-product-full}
    Y_T^iY_T^j-Y_t^iY_t^j
    &=
    Y_t^i\Big(\frac{\Delta S_T^j}{\beta}-F(S_T^j,\beta)\Big)
    +Y_t^j\Big(\frac{\Delta S_T^i}{\alpha}-F(S_T^i,\alpha)\Big) \nonumber\\
    &\quad
    +\Big(\frac{\Delta S_T^i}{\alpha}-F(S_T^i,\alpha)\Big)
     \Big(\frac{\Delta S_T^j}{\beta}-F(S_T^j,\beta)\Big) \nonumber\\[1mm]
    &=
    \underbrace{\frac{Y_t^i}{\beta}\Delta S_T^j+\frac{Y_t^j}{\alpha}\Delta S_T^i}_{\text{linear-in-spot terms}}
    -\underbrace{\Big(Y_t^iF(S_T^j,\beta)+Y_t^jF(S_T^i,\alpha)\Big)}_{\text{univariate log-contract strips}}
    +\underbrace{\frac{\Delta S_T^i}{\alpha}\frac{\Delta S_T^j}{\beta}}_{\text{bilinear spot term}} \nonumber\\
    &\quad
    -\underbrace{\Big(\frac{\Delta S_T^i}{\alpha}F(S_T^j,\beta)+\frac{\Delta S_T^j}{\beta}F(S_T^i,\alpha)\Big)}_{\text{mixed linear-convexity terms}}
    +\underbrace{F(S_T^i,\alpha)F(S_T^j,\beta)}_{\text{pure cross-convexity term}}.
\end{align}

\subsubsection*{Step 3: Quadrant-by-Quadrant Product Representation}

\paragraph{(i) Pure Cross-Convexity Strip.}
Multiplying the two one-dimensional log kernels in \eqref{eq:log-kernel-def} gives
\begin{align}\label{eq:F-product-double-integral-expanded}
    F(S_T^i,\alpha)F(S_T^j,\beta)
    &=
    \int_0^\alpha\int_0^\beta
    \frac{(K_1-S_T^i)^+(K_2-S_T^j)^+}{K_1^2K_2^2}\,\D K_2\,\D K_1 \nonumber\\
    &\quad
    +\int_0^\alpha\int_\beta^\infty
    \frac{(K_1-S_T^i)^+(S_T^j-K_2)^+}{K_1^2K_2^2}\,\D K_2\,\D K_1 \nonumber\\
    &\quad
    +\int_\alpha^\infty\int_0^\beta
    \frac{(S_T^i-K_1)^+(K_2-S_T^j)^+}{K_1^2K_2^2}\,\D K_2\,\D K_1 \nonumber\\
    &\quad
    +\int_\alpha^\infty\int_\beta^\infty
    \frac{(S_T^i-K_1)^+(S_T^j-K_2)^+}{K_1^2K_2^2}\,\D K_2\,\D K_1.
\end{align}
Define the generalized out-of-the-money quadrant payoff
\begin{align}\label{eq:geo-quadrant-payoffs}
    \Pi_{ij}(x,y;K_1,K_2)
    &:=
    \big((K_1-x)^+\mathbbm{1}_{\{K_1\leq \alpha\}}+(x-K_1)^+\mathbbm{1}_{\{K_1>\alpha\}}\big)\nonumber\\
    &\quad\times
    \big((K_2-y)^+\mathbbm{1}_{\{K_2\leq \beta\}}+(y-K_2)^+\mathbbm{1}_{\{K_2>\beta\}}\big),
\end{align}
so that the cross-convexity term may be written compactly as
\begin{align}\label{eq:geo-kernel-product-quanto}
    F(S_T^i,\alpha)F(S_T^j,\beta)
    =
    \iint_{\mathbb{R}_+^2}
    \frac{\Pi_{ij}(S_T^i,S_T^j;K_1,K_2)}{K_1^2K_2^2}\,\D K_1\,\D K_2.
\end{align}
Thus every product quadrant $\mathrm{CC}$, $\mathrm{CP}$, $\mathrm{PC}$, and $\mathrm{PP}$ enters this strip with the same nonnegative density $1/(K_1^2K_2^2)$.

\paragraph{(ii) Bilinear Spot Term.}
Introduce the ramp functions

\begin{equation*}
    C_\alpha(x):=(x-\alpha)^+,\qquad P_\alpha(x):=(\alpha-x)^+,\qquad
    C_\beta(y):=(y-\beta)^+,\qquad P_\beta(y):=(\beta-y)^+.
\end{equation*}
Then the identity
\begin{align}\label{eq:bilinear-quadrant-identity}
    (x-\alpha)(y-\beta)
    =
    C_\alpha(x)C_\beta(y)-C_\alpha(x)P_\beta(y)-P_\alpha(x)C_\beta(y)+P_\alpha(x)P_\beta(y)
\end{align}
holds for all $(x,y)\in\mathbb{R}_+^2$. Applying \eqref{eq:bilinear-quadrant-identity} to the bilinear term in \eqref{eq:terminal-product-full} yields
\begin{align}\label{eq:bilinear-term-cccp}
    \frac{\Delta S_T^i}{\alpha}\frac{\Delta S_T^j}{\beta}
    =
    \frac{1}{\alpha\beta}
    \Big(
        C_\alpha(S_T^i)C_\beta(S_T^j)
        -C_\alpha(S_T^i)P_\beta(S_T^j)
        -P_\alpha(S_T^i)C_\beta(S_T^j)
        +P_\alpha(S_T^i)P_\beta(S_T^j)
    \Big).
\end{align}
This alternating sign pattern directly reflects the mixed-derivative term $h_{x_1x_2}(\alpha_1,\alpha_2)(x_1-\alpha_1)(x_2-\alpha_2)$ in \eqref{functionalreplication2d}.

\subsubsection*{Step 4: Full Covariance Swap Static/Dynamic Spanning}
Combining the preceding identities, the geometric covariance swap payoff \eqref{eq:geo-covswap-payoff} admits the decomposition
\begin{align}\label{eq:full-covswap-expansion-product-simplified}
    H_T^{\mathrm{g\text{-}cov}}
    &=
    \langle Y^{i},Y^{j}\rangle_{t,T}-K_{\mathrm{var}} \nonumber\\
    &=
    \bigl(Y_T^iY_T^j-Y_t^iY_t^j\bigr)
    -\int_t^T Y_u^i\,\D Y_u^j
    -\int_t^T Y_u^j\,\D Y_u^i
    -K_{\mathrm{var}} \nonumber\\
    &=
    \underbrace{\Big(\frac{Y_t^i}{\beta}(S_T^j-\beta)+\frac{Y_t^j}{\alpha}(S_T^i-\alpha)\Big)}_{\text{static linear terms}}
    -\underbrace{\Big(Y_t^iF(S_T^j,\beta)+Y_t^jF(S_T^i,\alpha)\Big)}_{\text{static one-dimensional OTM strips}} \nonumber\\
    &\quad
    +\underbrace{\frac{1}{\alpha\beta}\Big(C_\alpha^iC_\beta^j-C_\alpha^iP_\beta^j-P_\alpha^iC_\beta^j+P_\alpha^iP_\beta^j\Big)}_{\text{bilinear mixed term}}
    -\underbrace{\frac{1}{\alpha}(C_\alpha^i-P_\alpha^i)F(S_T^j,\beta)}_{\text{mixed term 1}} \nonumber\\
    &\quad
    -\underbrace{\frac{1}{\beta}(C_\beta^j-P_\beta^j)F(S_T^i,\alpha)}_{\text{mixed term 2}}
    +\underbrace{F(S_T^i,\alpha)F(S_T^j,\beta)}_{\text{cross-convexity}}
    -\underbrace{\int_t^T \frac{Y_u^i}{S_u^j}\,\D S_u^j+\int_t^T \frac{Y_u^j}{S_u^i}\,\D S_u^i}_{\text{dynamic trading}} \nonumber\\
    &\quad
    +\underbrace{\frac12\int_t^T \frac{Y_u^i}{(S_u^j)^2}\,\D\langle S^j\rangle_u+\frac12\int_t^T \frac{Y_u^j}{(S_u^i)^2}\,\D\langle S^i\rangle_u}_{\text{weighted variance terms}}.
\end{align}

\noindent
The mixed terms in \eqref{eq:full-covswap-expansion-product-simplified} admit the explicit signed product-strip expansions
\begin{align}
    \frac{1}{\alpha}(C_\alpha^i-P_\alpha^i)F(S_T^j,\beta)
    &=
    \frac{1}{\alpha}\int_0^\beta \frac{C_\alpha^iP_{K_2}^j}{K_2^2}\,\D K_2
    +\frac{1}{\alpha}\int_\beta^\infty \frac{C_\alpha^iC_{K_2}^j}{K_2^2}\,\D K_2 \nonumber\\
    &\quad
    -\frac{1}{\alpha}\int_0^\beta \frac{P_\alpha^iP_{K_2}^j}{K_2^2}\,\D K_2
    -\frac{1}{\alpha}\int_\beta^\infty \frac{P_\alpha^iC_{K_2}^j}{K_2^2}\,\D K_2, \\[1ex]
    \frac{1}{\beta}(C_\beta^j-P_\beta^j)F(S_T^i,\alpha)
    &=
    \frac{1}{\beta}\int_0^\alpha \frac{C_\beta^jP_{K_1}^i}{K_1^2}\,\D K_1
    +\frac{1}{\beta}\int_\alpha^\infty \frac{C_\beta^jC_{K_1}^i}{K_1^2}\,\D K_1 \nonumber\\
    &\quad
    -\frac{1}{\beta}\int_0^\alpha \frac{P_\beta^jP_{K_1}^i}{K_1^2}\,\D K_1
    -\frac{1}{\beta}\int_\alpha^\infty \frac{P_\beta^jC_{K_1}^i}{K_1^2}\,\D K_1.
\end{align}
In particular, short positions in the cross-product quadrants $\mathrm{CP}$ and $\mathrm{PC}$ arise \emph{structurally}, through the alternating signs in the bilinear identity and the mixed linear-convexity terms. This is the key mechanism through which the static portfolio reproduces the signed cross-gamma geometry of the covariance payoff.

Equation \eqref{eq:full-covswap-expansion-product-simplified} therefore has the following precise meaning. The terminal cross-product $Y_T^iY_T^j-Y_t^iY_t^j$ and its mixed-convexity components lie in the linear span generated by dynamic trading in $(S^i,S^j)$, univariate log-contract strips, and signed strips of product options across the four quadrants. The full floating leg $\langle Y^i,Y^j\rangle_{t,T}$ equals this spanned component \emph{plus} the weighted-variance terms
$$
\frac12\int_t^T \frac{Y_u^i}{(S_u^j)^2}\,\D\langle S^j\rangle_u
\;+\;
\frac12\int_t^T \frac{Y_u^j}{(S_u^i)^2}\,\D\langle S^i\rangle_u,
$$
which are not statically generated by the same option families. Exact replication thus requires augmenting the market by suitable variance-linked contracts (for example (weighted)-variance swaps or log-contract portfolios on the individual assets) \cite{CarrLee}. Absent such instruments, the decomposition should be used to identify the natural static auxiliaries for variance-optimal hedging rather than to claim pathwise completeness.
\end{example}

We now provide an economic and analytical rationale for the inverse-square densities $1/K^2$ and $1/(K_1^2K_2^2)$ that appear in the one- and two-dimensional spanning formulas under a univariate and bivariate lognormal setting. The one-dimensional inverse-square kernel is the weight arising from local vega-flatness of a vanilla strip. The multivariate analogue is obtained in the same way, once one uses the
explicit analytic product-option formulas in the bivariate lognormal
model; see Madan and Wang \cite{Madan2021Pricing}. In particular, the
cross-vega of each quadrant product option is itself scale-invariant,
depending on the spot variables only through the moneyness ratios
$K_1/S_1$ and $K_2/S_2$.

\clearpage 

\begin{proposition}[Inverse-square kernels from local vega-flatness]
\label{prop:e5-copy-paste}
Fix a maturity $\tau>0$.

\begin{enumerate}
    \item[(i)] Let
    \[
        \Pi^{(1)}(S_0,\nu)
        :=
        \int_0^\infty \lambda(K)\,O(S_0,K;\nu)\,\mathrm{d} K,
        \qquad
        \lambda\in C^1((0,\infty),\mathbb{R}_{>0}),
    \]
    where $O(S_0,K;\nu)$ denotes a Black--Scholes vanilla option price
    with spot $S_0$ and total variance $\nu=\sigma^2\tau$. Then
    \begin{equation}
    \label{eq:e5-copy-uni}
        \partial_{S_0}\partial_\nu \Pi^{(1)}
        =
        \frac{S_0}{2\sqrt{\nu}}
        \int_0^\infty
        \bigl(
            2\lambda(xS_0)+xS_0\lambda'(xS_0)
        \bigr)
        n\bigl(d_+(x;\nu)\bigr)\,\mathrm{d} x,
    \end{equation}
    where
    \[
        d_+(x;\nu):=-\frac{\log x}{\sqrt{\nu}}+\frac12\sqrt{\nu}.
    \]
    Consequently, the density-level local vega-flatness condition
    $\partial_{S_0}\partial_\nu \Pi^{(1)}=0$ yields
    \[
        2\lambda(K)+K\lambda'(K)=0,
        \qquad K>0,
    \]
    whose unique positive $C^1$ solution is
    \[
        \lambda(K)=\frac{c}{K^2},
        \qquad c>0.
    \]

    \item[(ii)] Let $\alpha\in\{\mathrm{CC},\mathrm{CP},\mathrm{PC},\mathrm{PP}\}$ and consider
    \[
        \Pi^{(2)}(S_1,S_2)
        :=
        \int_0^\infty\int_0^\infty
        \lambda(K_1,K_2)\,
        C^\alpha(S_1,S_2;K_1,K_2)\,
        \mathrm{d} K_1\,\mathrm{d} K_2,
    \]
    where $C^\alpha$ denotes the corresponding bivariate Black--Scholes
    product-option price, and assume $\lambda\in C^1((0,\infty)^2,\mathbb{R}_{>0})$
    with sufficient integrability to justify differentiation under the
    strike integrals. Then the explicit analytic formulas of
    \cite{Madan2021Pricing} imply
    \begin{equation}
    \label{eq:e5-copy-crossvega}
        \partial_{\nu_{12}}C^\alpha(S_1,S_2;K_1,K_2)
        =
        S_1S_2\mathrm{e}^{r\tau+\nu_{12}}
        Q^\alpha\!\left(\frac{K_1}{S_1},\frac{K_2}{S_2};\rho\right),
    \end{equation}
    where $Q^\alpha$ is the quadrant probability kernel associated with
    the chosen product-option type. Equivalently,
    \begin{align*}
        Q^{\mathrm{CC}}(x,y;\rho) &= \Phi_2^{++}(d_1(x),e_1(y);\rho),\\
        Q^{\mathrm{CP}}(x,y;\rho) &= -\Phi_2^{+-}(d_1(x),e_1(y);\rho),\\
        Q^{\mathrm{PC}}(x,y;\rho) &= -\Phi_2^{-+}(d_1(x),e_1(y);\rho),\\
        Q^{\mathrm{PP}}(x,y;\rho) &= \Phi_2^{--}(d_1(x),e_1(y);\rho),
    \end{align*}
    with
    \[
        d_1(x):=\frac{\log x-r\tau-\tfrac12\nu_{11}-\nu_{12}}{\sqrt{\nu_{11}}},
        \qquad
        e_1(y):=\frac{\log y-r\tau-\tfrac12\nu_{22}-\nu_{12}}{\sqrt{\nu_{22}}}.
    \]
    Therefore
    \begin{align}
    \label{eq:e5-copy-bi-1}
        \partial_{S_1}\partial_{\nu_{12}}\Pi^{(2)}
        &=
        S_1S_2^2\mathrm{e}^{r\tau+\nu_{12}}
        \int_0^\infty\int_0^\infty
        \bigl(
            2\lambda(xS_1,yS_2)+xS_1\partial_1\lambda(xS_1,yS_2)
        \bigr)
        Q^\alpha(x,y;\rho)\,\mathrm{d} x\,\mathrm{d} y,\\
    \label{eq:e5-copy-bi-2}
        \partial_{S_2}\partial_{\nu_{12}}\Pi^{(2)}
        &=
        S_1^2S_2\mathrm{e}^{r\tau+\nu_{12}}
        \int_0^\infty\int_0^\infty
        \bigl(
            2\lambda(xS_1,yS_2)+yS_2\partial_2\lambda(xS_1,yS_2)
        \bigr)
        Q^\alpha(x,y;\rho)\,\mathrm{d} x\,\mathrm{d} y.
    \end{align}
    Hence the local cross-vega-flatness condition yields
    the first-order system
    \[
        2\lambda(K_1,K_2)+K_1\partial_1\lambda(K_1,K_2)=0,
        \qquad
        2\lambda(K_1,K_2)+K_2\partial_2\lambda(K_1,K_2)=0,
    \]
    whose unique positive $C^1$ solution is
    \[
        \lambda(K_1,K_2)=\frac{c}{K_1^2K_2^2},
        \qquad c>0.
    \]
\end{enumerate}
\end{proposition}

\begin{proof}
Part (i) is the standard one-dimensional scaling argument. Since
\[
    \partial_\nu O(S_0,K;\nu)
    =
    \frac{S_0}{2\sqrt{\nu}}\,n\bigl(d_+(K/S_0;\nu)\bigr),
\]
the change of variables $K=xS_0$ gives
\[
    \partial_\nu \Pi^{(1)}
    =
    \frac{S_0^2}{2\sqrt{\nu}}
    \int_0^\infty \lambda(xS_0)\,n\bigl(d_+(x;\nu)\bigr)\,\mathrm{d} x,
\]
and differentiation with respect to $S_0$ yields
\eqref{eq:e5-copy-uni}. Requiring local vega-flatness at the density
level gives the ordinary differential equation
$2\lambda(K)+K\lambda'(K)=0$, whose unique positive $C^1$ solution is
$\lambda(K)=cK^{-2}$.

For part (ii), the explicit closed-form bivariate Black--Scholes
product-option formulas imply \eqref{eq:e5-copy-crossvega}; this is the
precise multivariate analogue of the one-dimensional Black--Scholes vega
formula. Substituting \eqref{eq:e5-copy-crossvega} into the strip and
introducing the scaling variables $x=K_1/S_1$ and $y=K_2/S_2$ gives
\[
    \partial_{\nu_{12}}\Pi^{(2)}
    =
    (S_1S_2)^2\mathrm{e}^{r\tau+\nu_{12}}
    \int_0^\infty\int_0^\infty
    \lambda(xS_1,yS_2)\,
    Q^\alpha(x,y;\rho)\,\mathrm{d} x\,\mathrm{d} y.
\]
Differentiation with respect to $S_1$ and $S_2$ yields
\eqref{eq:e5-copy-bi-1}--\eqref{eq:e5-copy-bi-2}. Imposing local
cross-vega-flatness forces the bracketed terms to vanish, hence
\[
    2\lambda(K_1,K_2)+K_1\partial_1\lambda(K_1,K_2)=0,
    \qquad
    2\lambda(K_1,K_2)+K_2\partial_2\lambda(K_1,K_2)=0.
\]
Solving first in the $K_1$-variable gives
$\lambda(K_1,K_2)=a(K_2)K_1^{-2}$; substituting into the second
equation yields $2a(K_2)+K_2a'(K_2)=0$, so
$a(K_2)=cK_2^{-2}$. Therefore
$\lambda(K_1,K_2)=c/(K_1^2K_2^2)$.
\end{proof}

We now present an alternative replication formula for covariance swaps based on the bilinearity of quadratic covariation. In contrast to the product-option construction above, this approach reduces the problem to one-dimensional spanning formulas applied to suitably chosen transformed coordinates.

\begin{example}[Geometric Covariance Swap Replication via Log-Spread Options]
\label{cov-swapreplicationlogspreads}
Fix a maturity \(T>0\) and a valuation time \(t\in[0,T)\). For assets \(i,j\in\{1,\dots,d\}\), let \(Y_u^k:=\log S_u^k\) for \(u\in[t,T]\). The realized geometric covariance is defined by
\begin{align}\label{eq:geo-qcov-def-spread}
    \langle Y^{i},Y^{j}\rangle_{t,T}
    :=
    p\text{-}\lim_{|\pi|\to 0}
    \sum_{m=0}^{n-1}
    \bigl(Y^i_{t_{m+1}}-Y^i_{t_m}\bigr)
    \bigl(Y^j_{t_{m+1}}-Y^j_{t_m}\bigr),
\end{align}
and the corresponding covariance-swap payoff is
\begin{align}\label{eq:geo-covswap-payoff-spread}
    H_T^{\mathrm{g\text{-}cov}}
    :=
    \langle Y^{i},Y^{j}\rangle_{t,T}-K_{\mathrm{var}},
    \qquad
    K_{\mathrm{var}}=\mathbb{E}_t^{\mathbb{Q}}\big[\langle Y^{i},Y^{j}\rangle_{t,T}\big].
\end{align}
As above, this example considers the case of strictly positive continuous semimartingales. The construction proceeds in two steps: Step~1 reduces the covariance swap to three one-dimensional quadratic-variation terms, and Step~2 applies the classical Carr--Madan log-contract replication to each of these three variance legs.

\subsubsection*{Step 1: Polarization Formula for Quadratic Covariation}
The bilinearity of quadratic covariation yields the polarization identity
\begin{align}\label{eq:polarizationquadraticvariation}
    \langle Y^i,Y^j\rangle_{t,T}
    =
    \frac12\Big(
        \langle Y^i\rangle_{t,T}
        +\langle Y^j\rangle_{t,T}
        -\langle Y^i-Y^j\rangle_{t,T}
    \Big).
\end{align}
Thus, at the level of the continuous-path polarization identity, a covariance swap may be represented as a long position in the two marginal variance swaps together with a short position in the variance swap written on the log-spread process \(Y^i-Y^j=\log(S^i/S^j)\).

\subsubsection*{Step 2: Carr--Madan Replication of the Three Variance Legs}
Set
\[
    R_u^{ij}:=\frac{S_u^i}{S_u^j},
    \qquad u\in[t,T],
\]
so that \(Y_u^i-Y_u^j=\log R_u^{ij}\). For any strictly positive continuous semimartingale \(X\) on \([t,T]\), It\^o's formula yields
\begin{align}\label{eq:cm-log-qv-generic}
    \langle \log X\rangle_{t,T}
    =
    2\int_t^T \frac{1}{X_u}\,\D X_u
    -2\log\Big(\frac{X_T}{X_t}\Big).
\end{align}
Applying the one-dimensional Carr--Madan log-contract identity at the anchor \(X_t\),
\begin{align}\label{eq:cm-log-contract-generic}
    \log\Big(\frac{X_T}{X_t}\Big)
    =
    \frac{X_T-X_t}{X_t}
    -\int_0^{X_t}\frac{(K-X_T)^+}{K^2}\,\D K
    -\int_{X_t}^{\infty}\frac{(X_T-K)^+}{K^2}\,\D K,
\end{align}
one obtains the replication formula
\begin{align}\label{eq:cm-qv-generic}
    \langle \log X\rangle_{t,T}
    &=
    2\int_t^T \frac{1}{X_u}\,\D X_u
    -\frac{2}{X_t}(X_T-X_t) \nonumber\\
    &\quad
    +2\int_0^{X_t}\frac{(K-X_T)^+}{K^2}\,\D K
    +2\int_{X_t}^{\infty}\frac{(X_T-K)^+}{K^2}\,\D K.
\end{align}
Applying \eqref{eq:cm-qv-generic} to \(X=S^i\), \(X=S^j\), and \(X=R^{ij}\) gives
\begin{align}\label{eq:cm-three-legs}
    \langle Y^i\rangle_{t,T}
    &=
    2\int_t^T \frac{1}{S_u^i}\,\D S_u^i
    -\frac{2}{S_t^i}(S_T^i-S_t^i)
    +2\int_0^{S_t^i}\frac{(K-S_T^i)^+}{K^2}\,\D K
    +2\int_{S_t^i}^{\infty}\frac{(S_T^i-K)^+}{K^2}\,\D K, \nonumber\\[1mm]
    \langle Y^j\rangle_{t,T}
    &=
    2\int_t^T \frac{1}{S_u^j}\,\D S_u^j
    -\frac{2}{S_t^j}(S_T^j-S_t^j)
    +2\int_0^{S_t^j}\frac{(K-S_T^j)^+}{K^2}\,\D K
    +2\int_{S_t^j}^{\infty}\frac{(S_T^j-K)^+}{K^2}\,\D K, \nonumber\\[1mm]
    \langle Y^i-Y^j\rangle_{t,T}
    &=
    2\int_t^T \frac{1}{R_u^{ij}}\,\D R_u^{ij}
    -\frac{2}{R_t^{ij}}(R_T^{ij}-R_t^{ij})
    +2\int_0^{R_t^{ij}}\frac{(K-R_T^{ij})^+}{K^2}\,\D K
    +2\int_{R_t^{ij}}^{\infty}\frac{(R_T^{ij}-K)^+}{K^2}\,\D K.
\end{align}
Each of the three terms on the right-hand of \eqref{eq:cm-three-legs} side is now explicitly represented as the sum of a dynamic trading term and a static strip of one-dimensional out-of-the-money options. Hence the covariance payoff is synthesized by combining the Carr--Madan strips on \(S^i\), on \(S^j\), and on the exchange ratio \(R^{ij}=S^i/S^j\). If options on the ratio process are not traded, this representation should be interpreted as an instrument-selection identity rather than as an exact market replication.
\end{example}

\begin{remark}[Instrument Selection and Liquidity Constraints]
Although \eqref{eq:polarizationquadraticvariation} provides a transparent structural decomposition of covariance into marginal variance and spread-variance directions, practical implementation depends critically on instrument availability. In liquid exchange-traded markets, options on simple arithmetic spreads $S^i-S^j$ are often more readily available than options on the ratio $S^i/S^j$ or on the log-ratio $\log(S^i/S^j)$. However, from the perspective of the geometric covariance payoff, the natural transformed coordinate is the log-spread $Y^i-Y^j$, and the ratio $S^i/S^j$ is its multiplicative counterpart. Consequently, options on the arithmetic spread generally provide a poorer functional approximation than options written directly on $S^i/S^j$ or $\log(S^i/S^j)$. This distinction is important in the numerical section below: although simple spread options may be more liquid, they are not aligned with the exact spanning structure of the geometric covariance payoff, whereas product options and log-contract-type instruments are.
\end{remark}

\subsection{Optimal Weight Selection}
\label{preliminaryresults}
We now present a first set of numerical results whose purpose is twofold: first, to quantify the variance reduction achieved by the semi-static variance-optimal hedge relative to the purely dynamic benchmark, and second, to connect the empirically selected auxiliary instruments to the spanning formulas developed in Examples~\ref{cov-swapreplicationQuanto} and~\ref{cov-swapreplicationlogspreads}. At this stage, we deliberately do not impose a specific stochastic (covariance) model for the joint asset dynamics, since our immediate objective is to isolate the structural features of the semi-static hedging problem itself. The precise probabilistic setting, together with the concrete stochastic covariance models used for implementation, pricing, and hedging, will be introduced later in Section~\ref{sec:covswap-semistatic}. This subsection's purpose is therefore to visualize the geometry of the outer problem and the qualitative effect of different auxiliary families. A fully specified and reproducible implementation, including the simulation model and the estimators of $A$, $B$, and $C$, is given later in Section~\ref{sec:numerics-and-examples}.

Throughout this subsection, the target claim is the geometric covariance payoff

\begin{equation*}
    H_T^0=\langle Y^i,Y^j\rangle_{t,T} - K_{\mathrm{swap}},
\end{equation*}
and the dynamic-only benchmark is the GKW hedge associated with $H^0$, cf.\ \eqref{eq:GKW}--\eqref{eq:OptimalStrategy}. In the absence of static instruments, the residual hedging error is therefore

\begin{equation*}
    \epsilon(0)^2 = A = \mathbb{E}\big[(L_T^0)^2\big],
\end{equation*}
as follows directly from \eqref{eq:inner-error-variance}--\eqref{eq:quadratic-form}. For a fixed family of auxiliary instruments $\bm{\eta}^{\mathcal I}$, we then consider the sparse outer problem

\begin{equation*}
    \epsilon_{\mathcal I}^2(m)
    :=
    \min_{\bm{\nu}\in\mathbb{R}^{n_{\mathcal I}}}
    \Big\{
        \bm{\nu}^{\top}C_{\mathcal I}\bm{\nu}
        -2\bm{\nu}^{\top}B_{\mathcal I}
        +A
        \;:\;
        \|\bm{\nu}\|_0\le m
    \Big\},
\end{equation*}
together with its long-only analogue

\begin{equation*}
    \epsilon_{\mathcal I,+}^2(m)
    :=
    \min_{\bm{\nu}\in\mathbb{R}_+^{n_{\mathcal I}}}
    \Big\{
        \bm{\nu}^{\top}C_{\mathcal I}\bm{\nu}
        -2\bm{\nu}^{\top}B_{\mathcal I}
        +A
        \;:\;
        \|\bm{\nu}\|_0\le m
    \Big\},
\end{equation*}
which is the constrained counterpart of the normal equation $C\bm{\nu}^*=B$ in \eqref{eq:normal-eq}. Since, the exact solution of the cardinality-constrained problem is combinatorial, we approximate it by greedy-forward selection, following \cite{semistaticsparse}. In the present setting this is particularly natural, because the matrices $C_{\mathcal I}$ are typically close to singular: nearby strikes, neighboring product quadrants, and alternative spread coordinates generate highly collinear residual martingales $(L^1,\dots,L^n)$, so that many candidate instruments span almost the same direction in $L^2$. Consequently, the frontier should be read as a discrete approximation to the projection gain $B^{\top}C^{\dagger}B$ in \eqref{eq:epsstar}: the first few selected instruments account for the dominant projection directions, while subsequent additions produce only marginal reductions in $\epsilon_{\mathcal I}^2(m)$. This is exactly what is observed numerically in Figure \ref{fig:preliminary_frontier_gfs}: the largest decrease relative to the benchmark $\epsilon(0)^2=A$ occurs at small cardinalities, after which the curves flatten substantially. The numerical evidence also shows that vanilla options alone are not sufficient to eliminate the covariance-specific residual risk, because they primarily load on marginal convexity directions, whereas the covariance swap depends on genuinely mixed curvature; once one augments the static family either by product options, as suggested by \eqref{eq:full-covswap-expansion-product-simplified} (see Equation \eqref{functionalreplication2d} for the 2-dimensional functional replication formula), or by spread options, as suggested by \eqref{eq:polarizationquadraticvariation}, the reduction in MSHE becomes materially stronger. In this sense, the numerical frontiers provide direct empirical support for the instrument-selection principle implied by the replication theory: auxiliary claims are effective precisely when their residual martingales are well aligned with the orthogonal component $L^0$ of the target claim.

\begin{figure}
    \centering
    \includegraphics[width=1\linewidth]{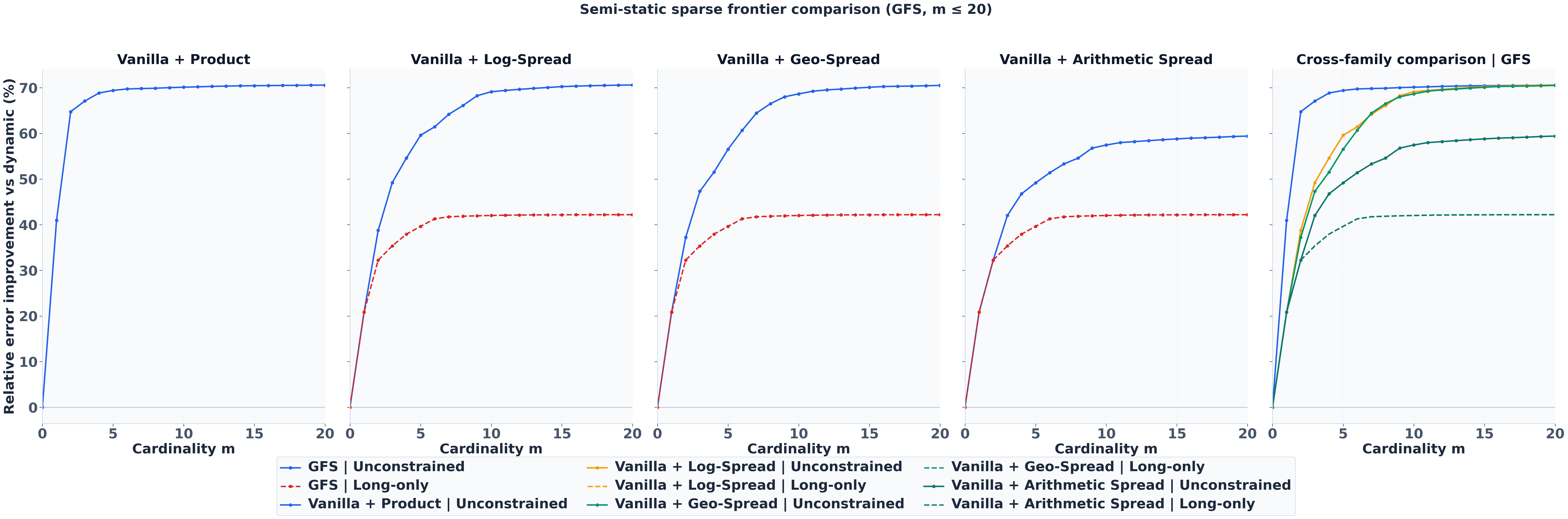}
    \caption{Greedy-forward sparse variance-optimal frontiers for the covariance-swap hedge. The figure reports the progressive improvement (decrease) of the mean-squared hedging error as the cardinality $m$ of the static portfolio increases. The steep initial improvement followed by a pronounced plateau is consistent with the quadratic structure \eqref{eq:quadratic-form}: a small number of auxiliary instruments captures most of the projection of $L_T^0$ onto $\mathrm{span}(L_T^1,\dots,L_T^n)$, while additional instruments are largely redundant because of strong collinearity across nearby strikes and related payoff families.}
    \label{fig:preliminary_frontier_gfs}
\end{figure}

\subsection*{Description of portfolio composition}

We now study in greater detail the structure of the optimal static portfolios $\bm{\nu}^*$ selected by the sparse outer optimization, with particular emphasis on how their composition reflects the spanning formulas of Section~\ref{sec:replicationtheory}. For each admissible family $\mathcal I$ and each cardinality level $m$, let $\widehat{\bm{\nu}}_{\mathcal I}^{(m)}$ denote the greedy-forward approximation to the solution of the constrained quadratic optimization above. The heatmaps therefore visualize, strike by strike and block by block, the discrete approximation of the abstract optimizer in \eqref{eq:nustar}; equivalently, they reveal how the projection of $L^0$ onto the span of the auxiliary residual martingales is implemented in practice. From this viewpoint, persistent selection of the same strikes across several values of $m$ should be interpreted as evidence that the corresponding instruments approximate robust basis directions of the residual covariance risk, while the absence of further diversification at large $m$ indicates that the dominant subspace has already been identified.

For the families \emph{vanilla + log-spread} and \emph{vanilla + geometric-spread}, the composition heatmaps presented in Figures \ref{fig:composition_logspread}-\ref{fig:composition_geospread}, show a remarkably stable pattern. In both cases, the optimizer selects a sparse strip of vanilla puts and calls on each marginal asset together with only a very small number of spread strikes concentrated near the central region of the spread coordinate. This is precisely the structure suggested by the one-dimensional Carr--Madan spanning principle \eqref{eq:replication1d-refined}--\eqref{eq:replication1d-psi} combined with the polarization identity \eqref{eq:polarizationquadraticvariation}. Indeed, the covariance payoff is decomposed into marginal variance directions and one opposing spread-variance direction, so the role of the vanillas is to span the marginal convexity terms, whereas the role of the spread options is to isolate and subtract the co-movement component.

\begin{figure}[!htbp]
    \centering
    \includegraphics[width=1\linewidth]{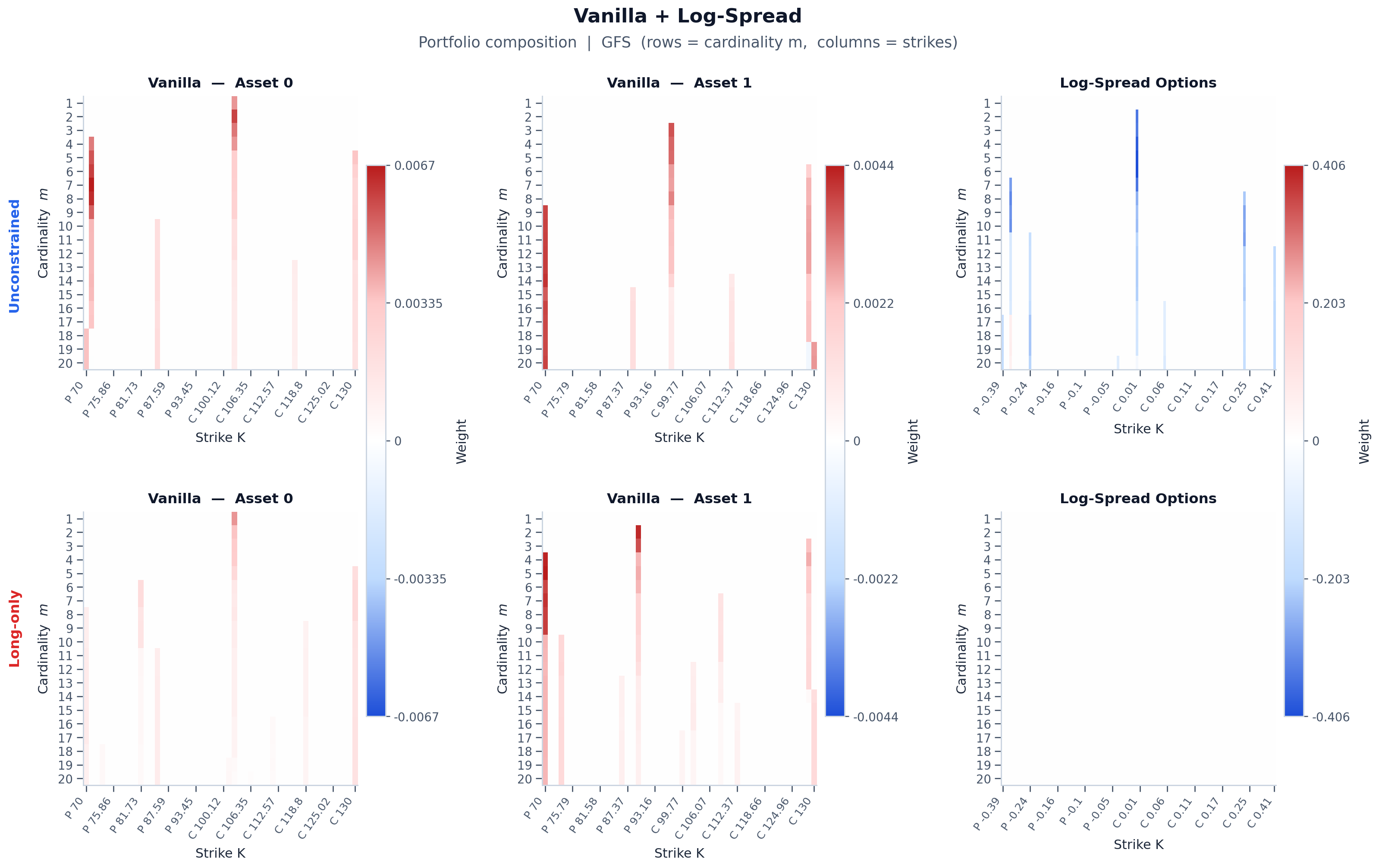}
    \caption{Stepwise composition of the sparse portfolio in the \emph{vanilla + log-spread} family. Rows correspond to the active cardinality $m$, and columns correspond to option strikes. The figure shows that the optimizer combines a small number of stable vanilla strikes on the two marginals with a highly concentrated set of log-spread options near the central spread region. This is consistent with the decomposition implied by \eqref{eq:polarizationquadraticvariation}, where the covariance direction is represented as a signed correction to marginal variance directions.}
    \label{fig:composition_logspread}
\end{figure}

\begin{figure}[!htbp]
    \centering
    \includegraphics[width=1\linewidth]{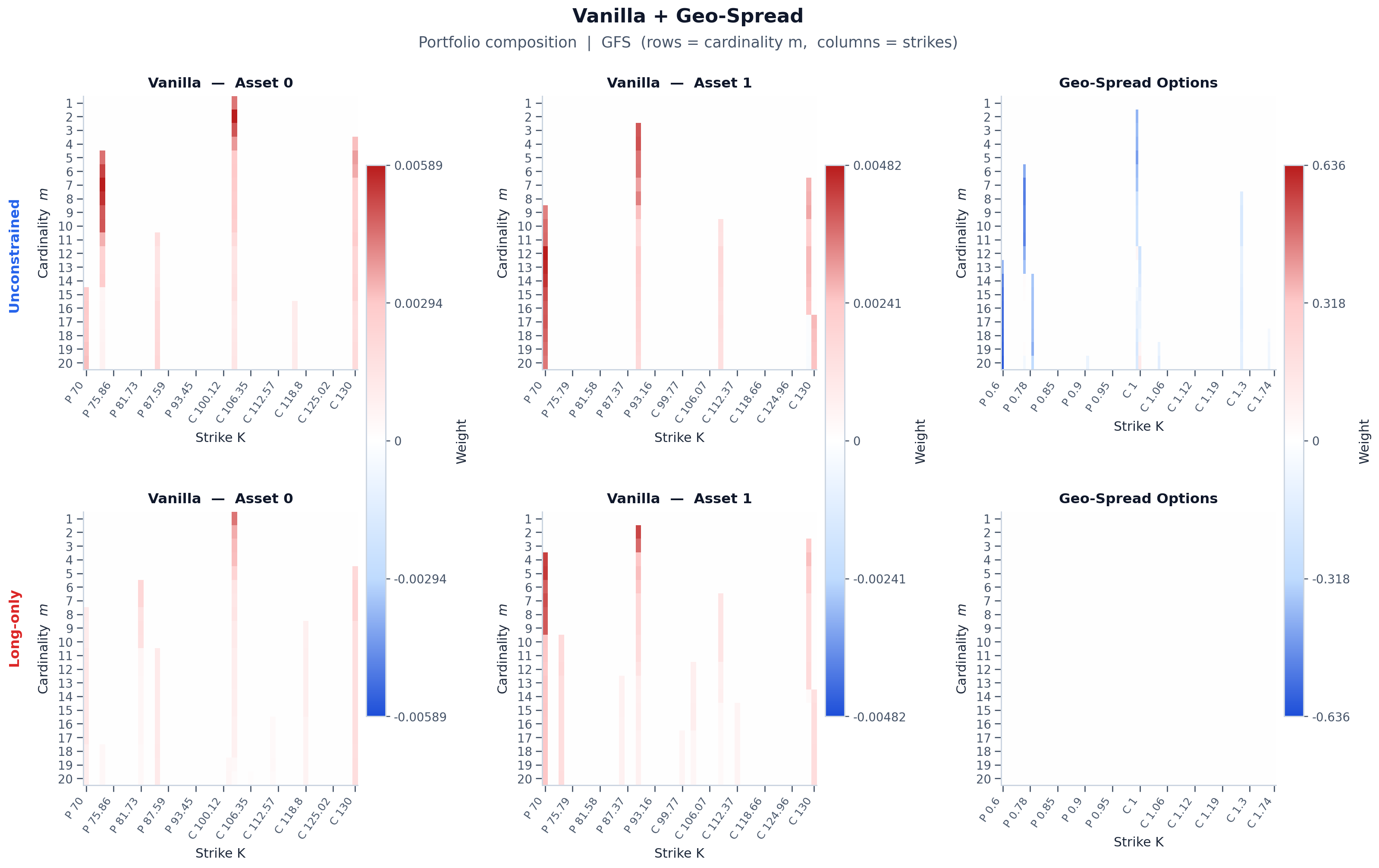}
    \caption{Stepwise composition of the sparse portfolio in the \emph{vanilla + geometric-spread} family. The structure closely mirrors the log-spread case, but now the auxiliary correction is expressed in the ratio coordinate $Z=S_1/S_2$. Numerically, this family behaves as a discrete proxy for the spread-variance correction, implemented in a multiplicative rather than additive coordinate system.}
    \label{fig:composition_geospread}
\end{figure}

The numerical payoff decompositions presented in Figures \ref{fig:payoff_logspread}-\ref{fig:payoff_geospread} confirm this interpretation: the vanilla block is a large, positive, U-shaped contribution in the diagonal variable $S_1=S_2=S$, while the spread block is centered around $Z=0$ for the log-spread case and $Z=1$ for the geometric-spread case, where it changes sign or becomes nearly flat at the pivot and contributes mainly away from the center. This is exactly the behavior one expects from the residual claim $\eta^{\bm{\nu}}=\eta^0-\bm{\nu}^{\top}\bm{\eta}$ in \eqref{eq:residual-claim}--\eqref{eq:residual-GKW}: the spread leg is not intended to reproduce the whole payoff on its own, but rather to remove the covariance-specific orthogonal component that cannot be spanned by marginal vanillas alone. By contrast, under long-only constraints the spread block becomes almost degenerate, because the replication in \eqref{eq:polarizationquadraticvariation} requires an effective short exposure to the spread-variance direction; once negative coefficients are excluded, the optimizer is forced back toward a mostly vanilla-based approximation, and the resulting payoff becomes more convex, more one-sided, and less faithful to the centered covariance geometry.

\begin{figure}[!htbp]
    \centering
    \includegraphics[width=1\linewidth]{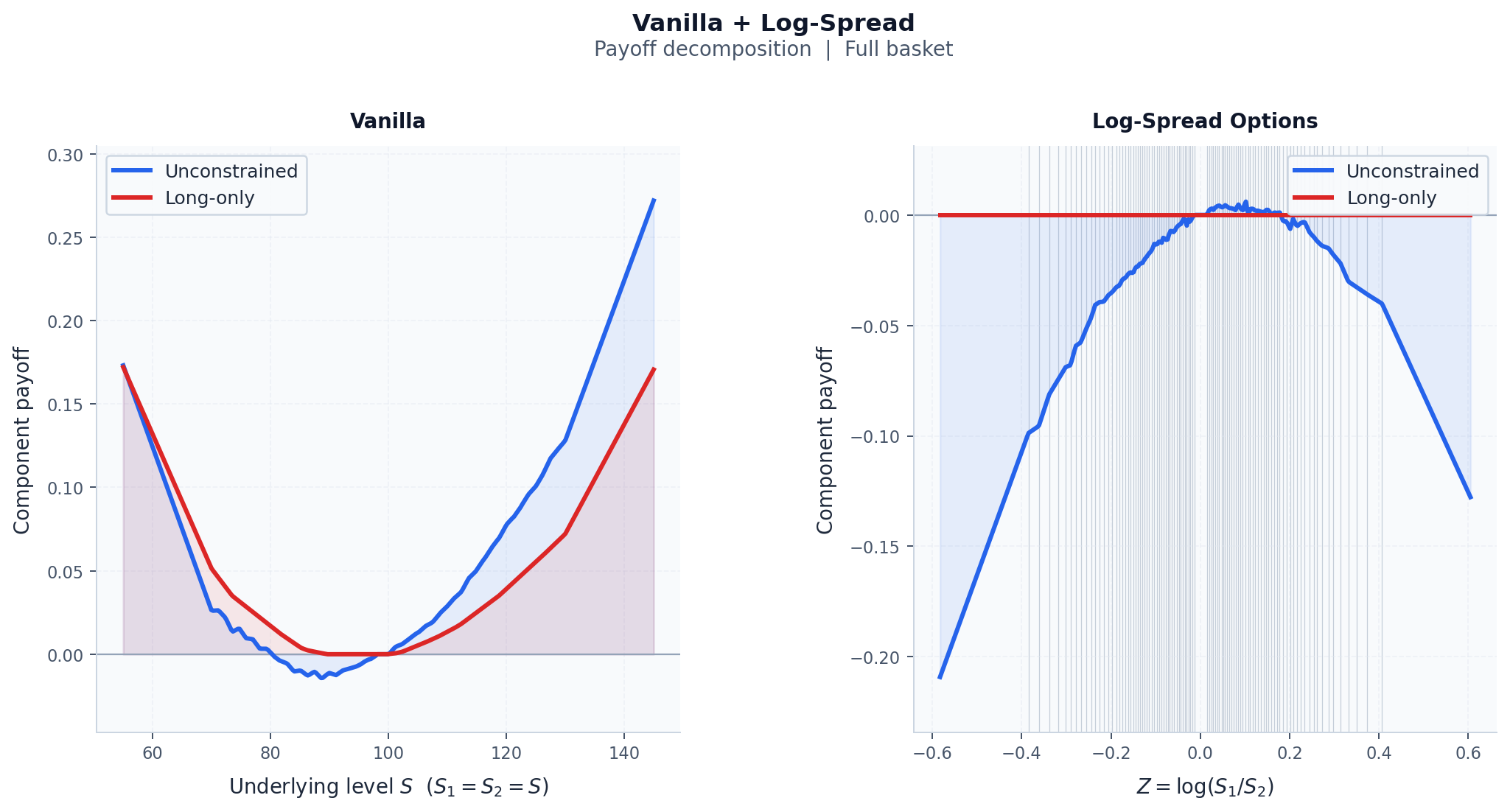}
    \caption{Payoff decomposition of the full \emph{vanilla + log-spread} basket. The left panel reports the marginal vanilla contribution as a function of the diagonal level $S_1=S_2=S$, while the right panel reports the spread contribution as a function of $Z=\log(S_1/S_2)$. The unconstrained spread component is centered and sign-sensitive, which matches its role as the corrective covariance direction in \eqref{eq:polarizationquadraticvariation}; the long-only spread component is nearly suppressed, illustrating the loss of the necessary signed hedge.}
    \label{fig:payoff_logspread}
\end{figure}

\begin{figure}[!htbp]
    \centering
    \includegraphics[width=1\linewidth]{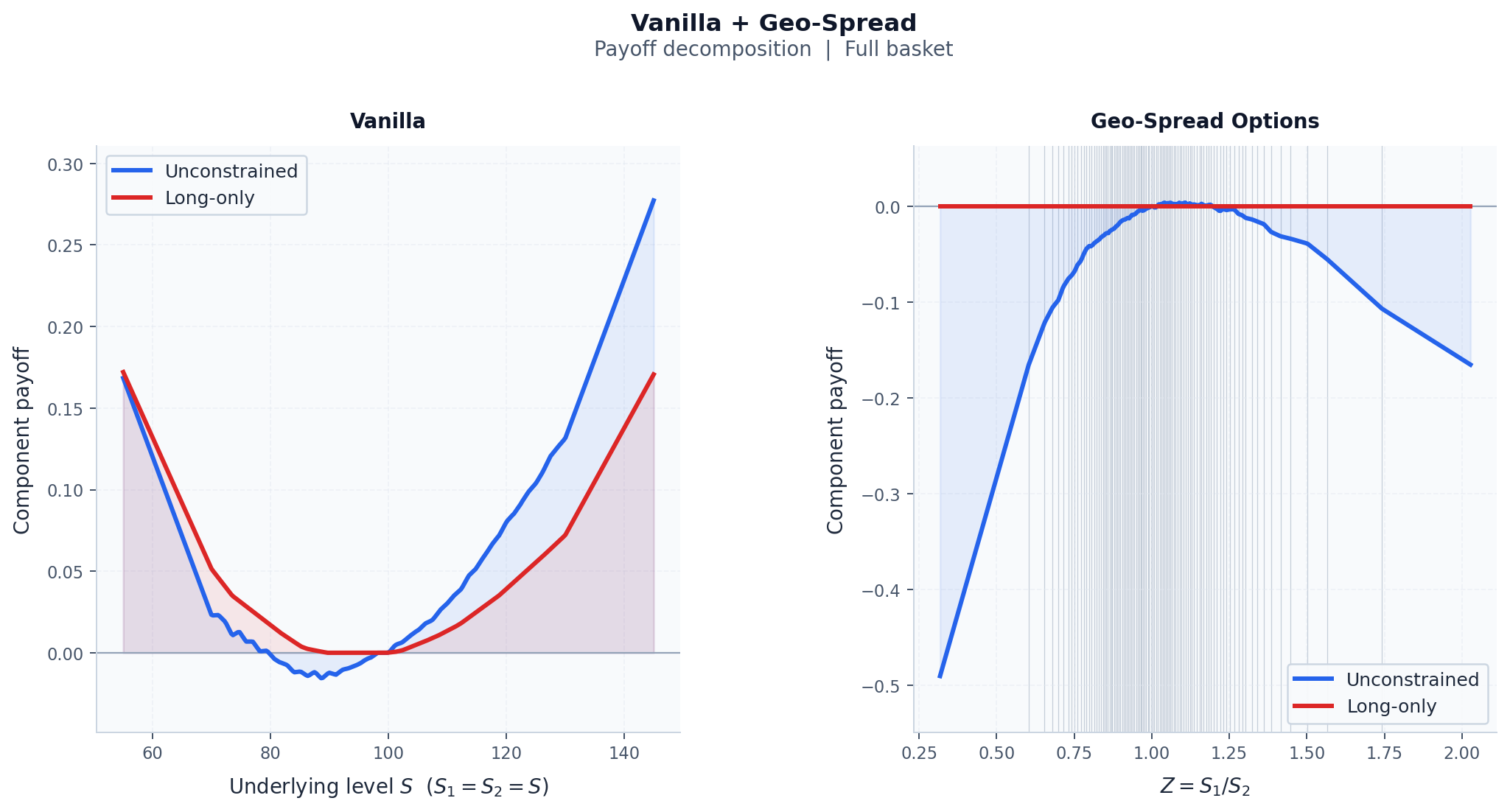}
    \caption{Payoff decomposition of the full \emph{vanilla + geometric-spread} basket. As in the log-spread case, the marginal vanilla block captures the dominant convexity along the diagonal, while the ratio-option block provides the covariance-specific correction in the relative coordinate $Z=S_1/S_2$. The unconstrained portfolio therefore preserves a centered relative-move hedge, whereas the long-only portfolio loses most of this correction.}
    \label{fig:payoff_geospread}
\end{figure}

\paragraph{Product options.}

In this case the optimizer has direct access to the four quadrant product payoffs $\mathrm{CC},\mathrm{CP},\mathrm{PC},\mathrm{PP}$, and the numerical compositions exhibit precisely the sign pattern predicted by the theory (see Equation \eqref{eq:full-covswap-expansion-product-simplified}). More specifically, the unconstrained optimizer allocates positive mass to $\mathrm{CC}$ and $\mathrm{PP}$ and negative mass to $\mathrm{CP}$ and $\mathrm{PC}$ around a tight neighborhood of central strike pairs $(K_1,K_2)$. This is the exact discrete analogue of the bilinear identity \eqref{eq:bilinear-quadrant-identity}--\eqref{eq:bilinear-term-cccp}, where the mixed term $(x-\alpha)(y-\beta)$ necessarily enters with alternating signs across quadrants, and it is further reinforced by the signed mixed terms in \eqref{eq:full-covswap-expansion-product-simplified}. The local product heatmap at $m=5$ in Figure \ref{fig:composition_product_local} makes this especially transparent: the unconstrained portfolio contains an explicit saddle-generating pattern, whereas the long-only portfolio retains only the nonnegative $\mathrm{CC}$ and $\mathrm{PP}$ directions and therefore cannot reproduce the negative mixed curvature carried by $\mathrm{CP}$ and $\mathrm{PC}$. The blockwise composition Figure \ref{fig:composition_product_blocks} shows that this phenomenon is not incidental but persists throughout the selection path: as $m$ increases, the unconstrained portfolio continues to distribute mass across all four quadrants, while the long-only portfolio concentrates almost entirely on nonnegative convex building blocks. From the viewpoint of the semi-static projection problem, this is exactly what one should expect. The unconstrained optimizer solves the normal equation $C\bm{\nu}^*=B$ in the full linear span of auxiliary residuals, so it can reproduce the centered saddle geometry of the covariance claim; the long-only optimizer instead projects onto a positive cone, and the resulting terminal payoff becomes predominantly nonnegative and nearly separable. The two-dimensional payoff surfaces in Figure~\ref{fig:surface_comparison_covariance} confirm this interpretation at the level of terminal profiles: the unconstrained hedges display the sign-changing shape characteristic of a covariance-type payoff, while the long-only surfaces are smoother, more one-sided, and systematically less able to encode negative cross-gamma exposure. The same ordering appears at the distributional level in Figure~\ref{fig:residuals_covariance}: for all three unconstrained auxiliary families, the semi-static portfolio yields the tightest concentration of terminal hedging errors around zero, whereas the single-block portfolios remain visibly more dispersed. 

\begin{figure}[!htbp]
    \centering
    \includegraphics[width=1\linewidth]{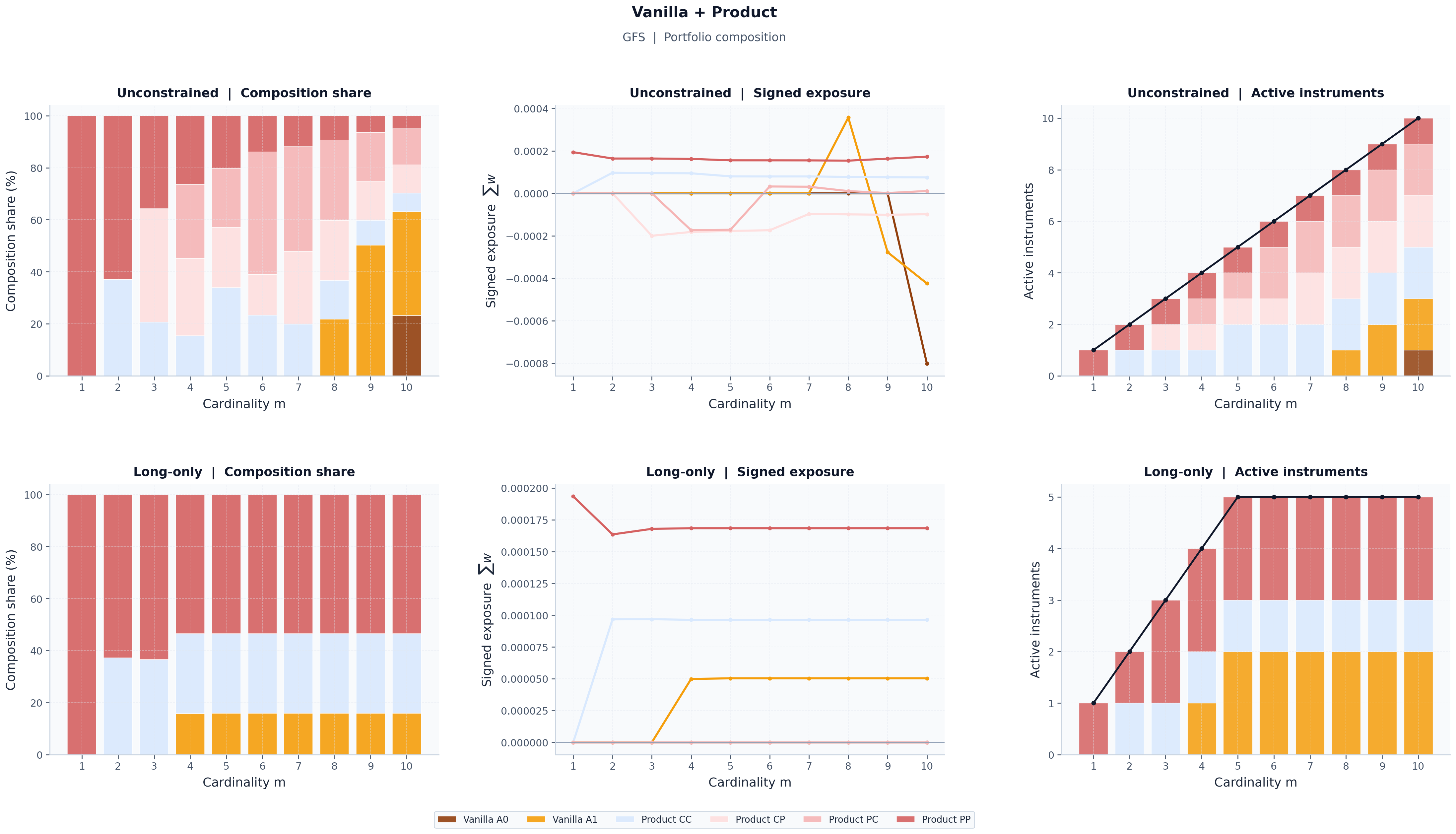}
    \caption{Blockwise composition of the sparse \emph{vanilla + product} portfolio under greedy-forward selection. The panels show the relative contribution of vanilla options and of each product quadrant $(\mathrm{CC},\mathrm{CP},\mathrm{PC},\mathrm{PP})$ as the cardinality increases. The unconstrained case uses all four quadrants, whereas the long-only case concentrates on nonnegative blocks only.}
    \label{fig:composition_product_blocks}
\end{figure}

\begin{figure}[!htbp]
    \centering
    \includegraphics[width=1\linewidth]{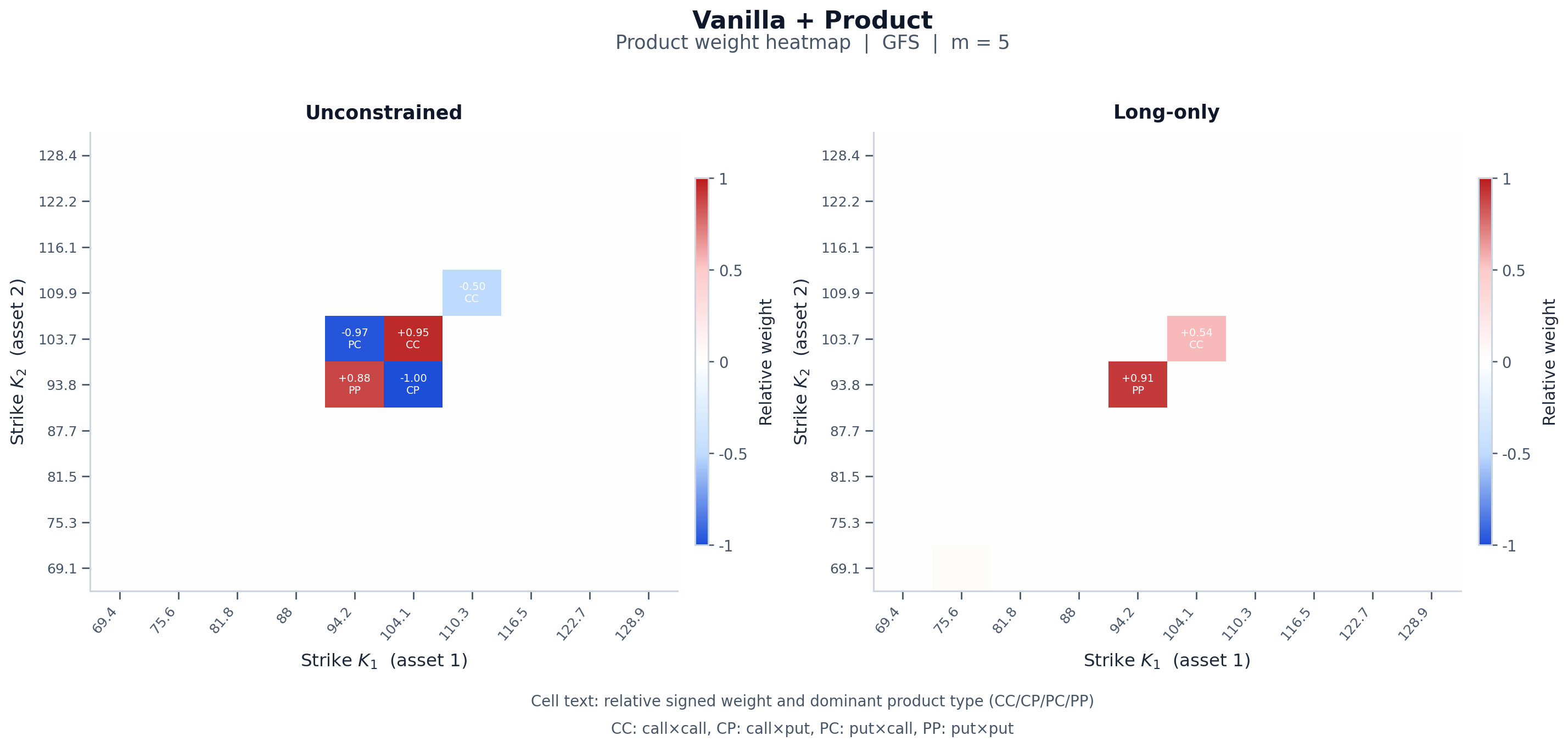}
    \caption{Local two-dimensional heatmap of selected product-option weights for the \emph{vanilla + product} family at $m=5$. The unconstrained optimizer displays the alternating sign pattern $\mathrm{CC}/\mathrm{PP}>0$ and $\mathrm{CP}/\mathrm{PC}<0$, which is the discrete signature of the bilinear covariance term in \eqref{eq:bilinear-term-cccp}. Under long-only constraints only the nonnegative quadrants remain active, so the centered saddle structure cannot be reproduced exactly.}
    \label{fig:composition_product_local}
\end{figure}

\begin{figure}[!htbp]
    \centering
    \begin{subfigure}[b]{0.88\textwidth}
        \centering
        \includegraphics[width=\textwidth]{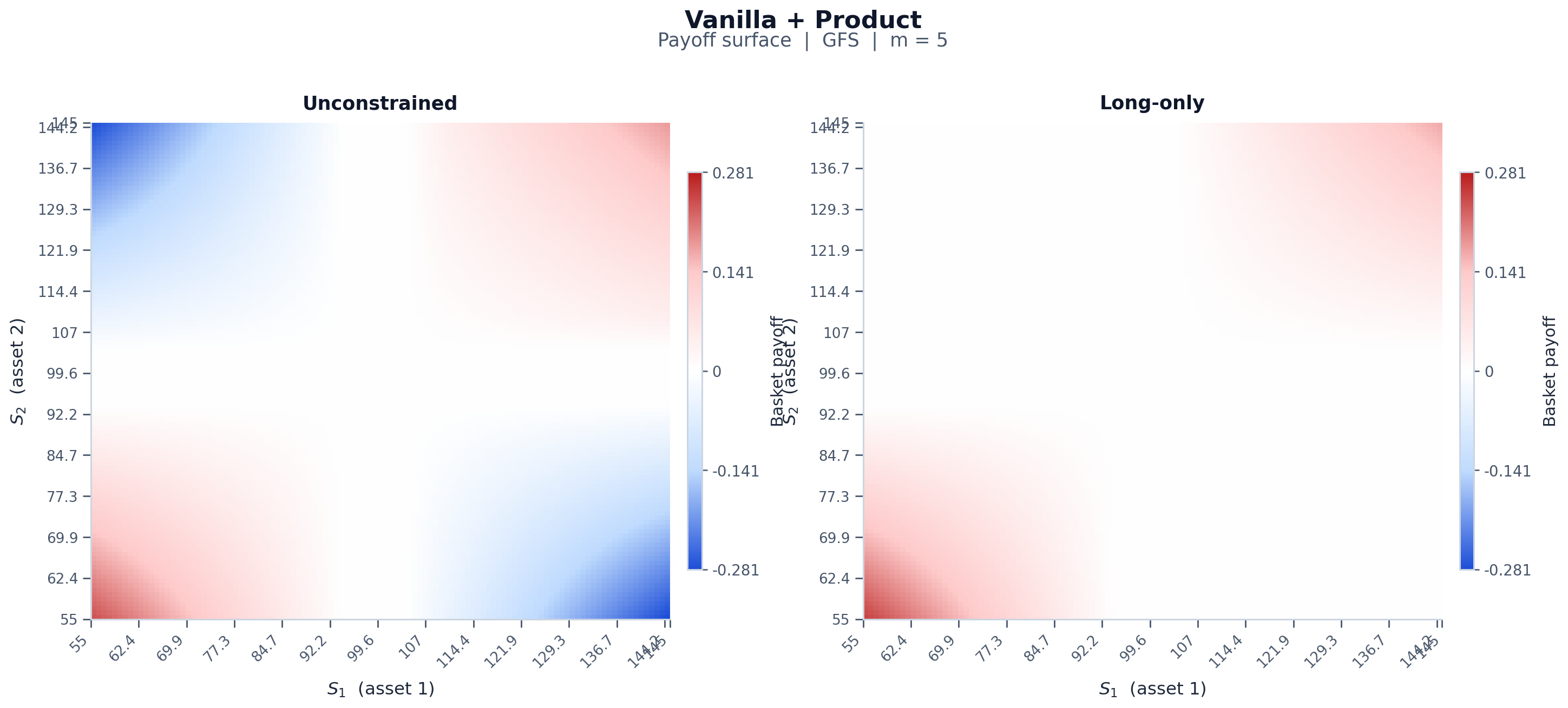}
    \end{subfigure}
    \vspace{0.5em}

    \begin{subfigure}[b]{0.88\textwidth}
        \centering
        \includegraphics[width=\textwidth]{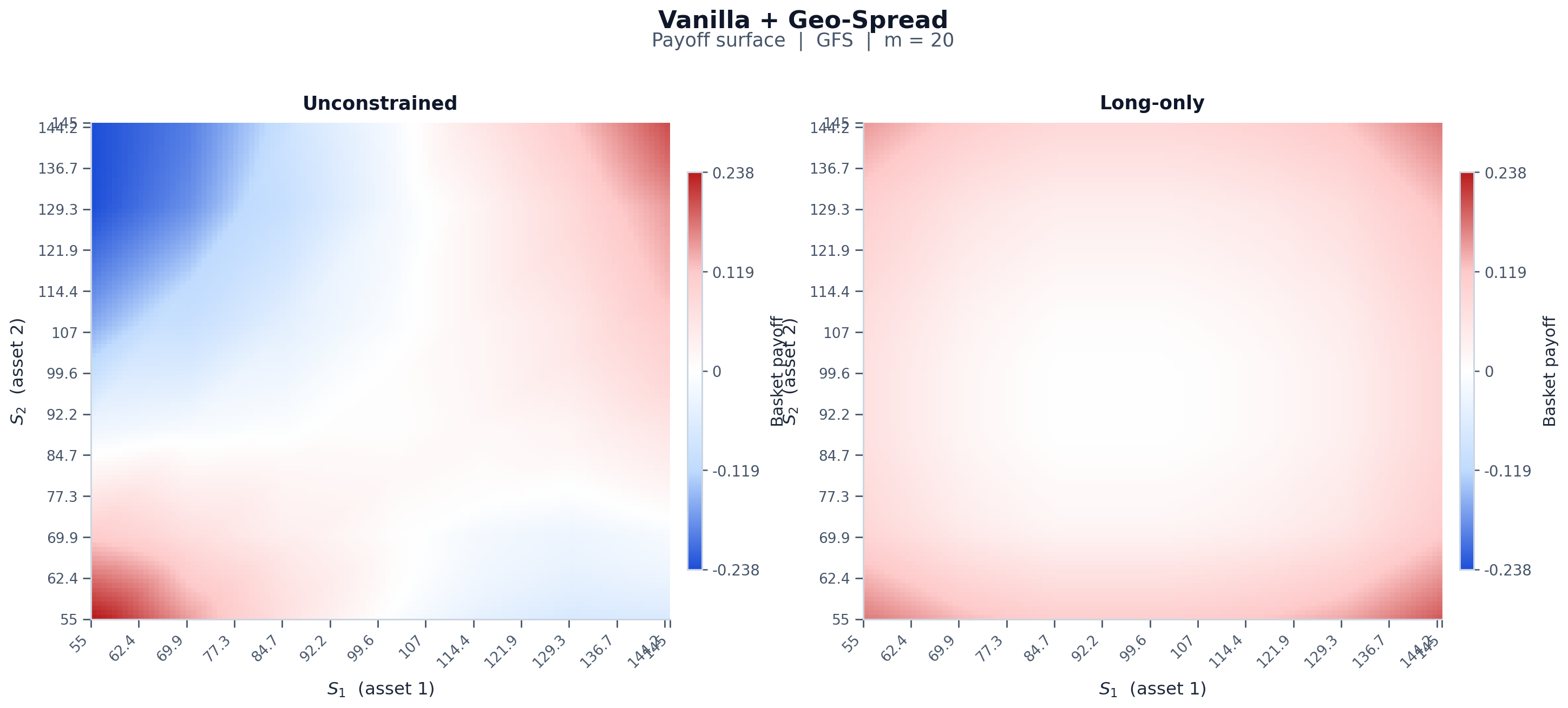}
    \end{subfigure}
    \vspace{0.5em}

    \begin{subfigure}[b]{0.88\textwidth}
        \centering
        \includegraphics[width=\textwidth]{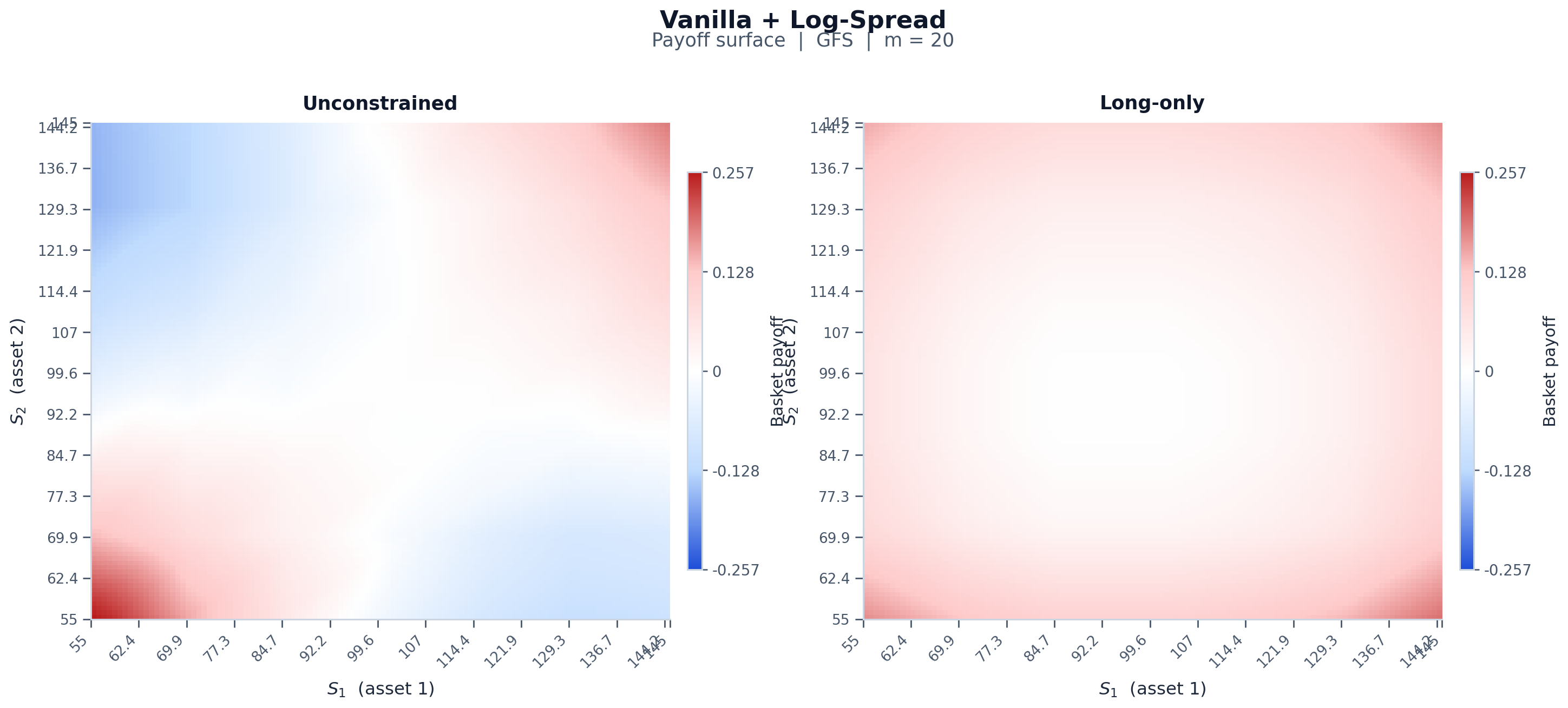}
    \end{subfigure}
    \caption{Comparison of terminal payoff surfaces generated by the three principal auxiliary families. The \emph{vanilla + product} family produces the most direct saddle-type geometry, reflecting the signed quadrant structure of \eqref{eq:bilinear-term-cccp} and the cross-convexity strip in \eqref{eq:geo-kernel-product-quanto}. The \emph{vanilla + log-spread} and \emph{vanilla + geometric-spread} families generate the same covariance direction indirectly through a relative-move coordinate, in line with \eqref{eq:polarizationquadraticvariation}. In all three cases, the long-only constraint suppresses the sign-changing part of the surface and therefore distorts the covariance hedge.}
    \label{fig:surface_comparison_covariance}
\end{figure}

\begin{figure}[!htbp]
    \centering
    \begin{subfigure}[b]{0.88\textwidth}
        \centering
        \includegraphics[width=\textwidth]{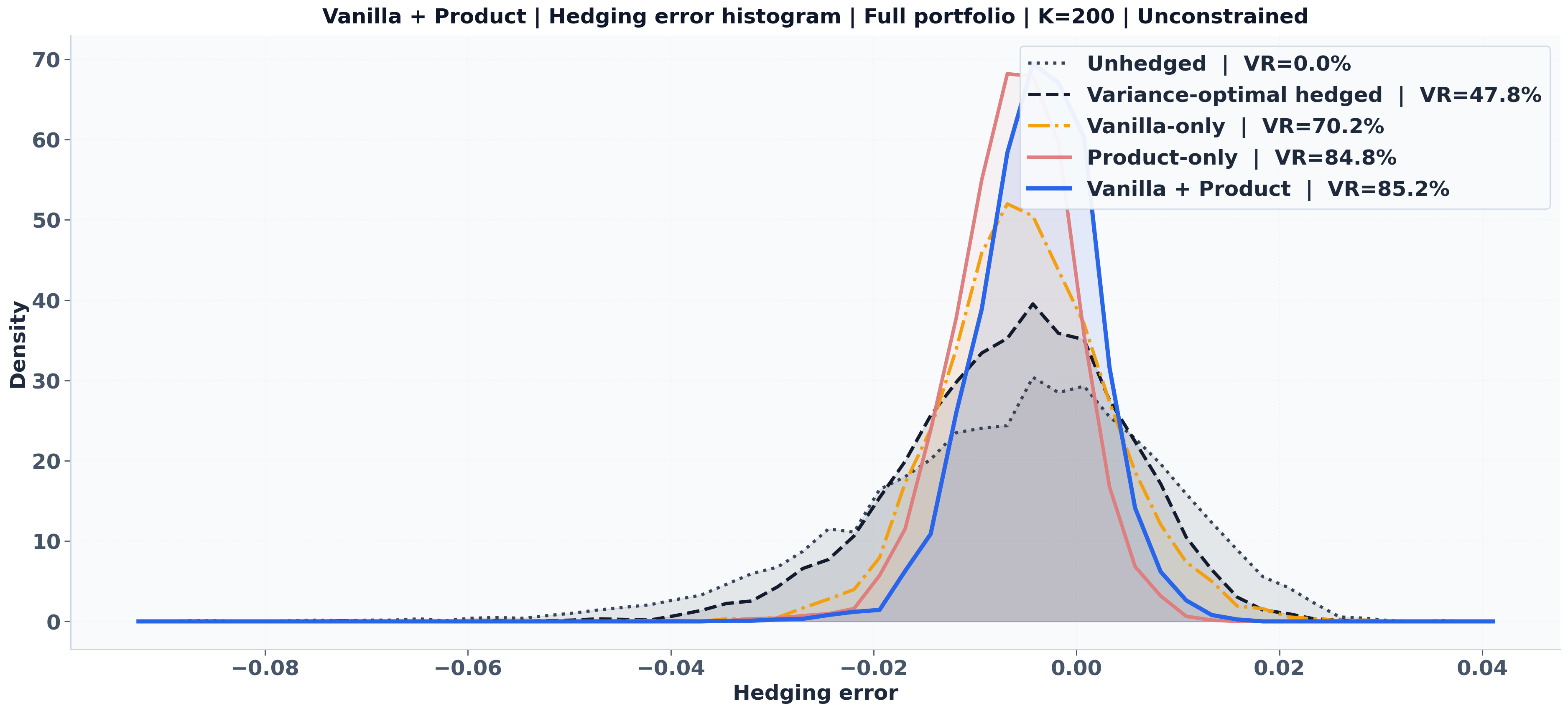}
    \end{subfigure}
    \vspace{0.5em}

    \begin{subfigure}[b]{0.88\textwidth}
        \centering
        \includegraphics[width=\textwidth]{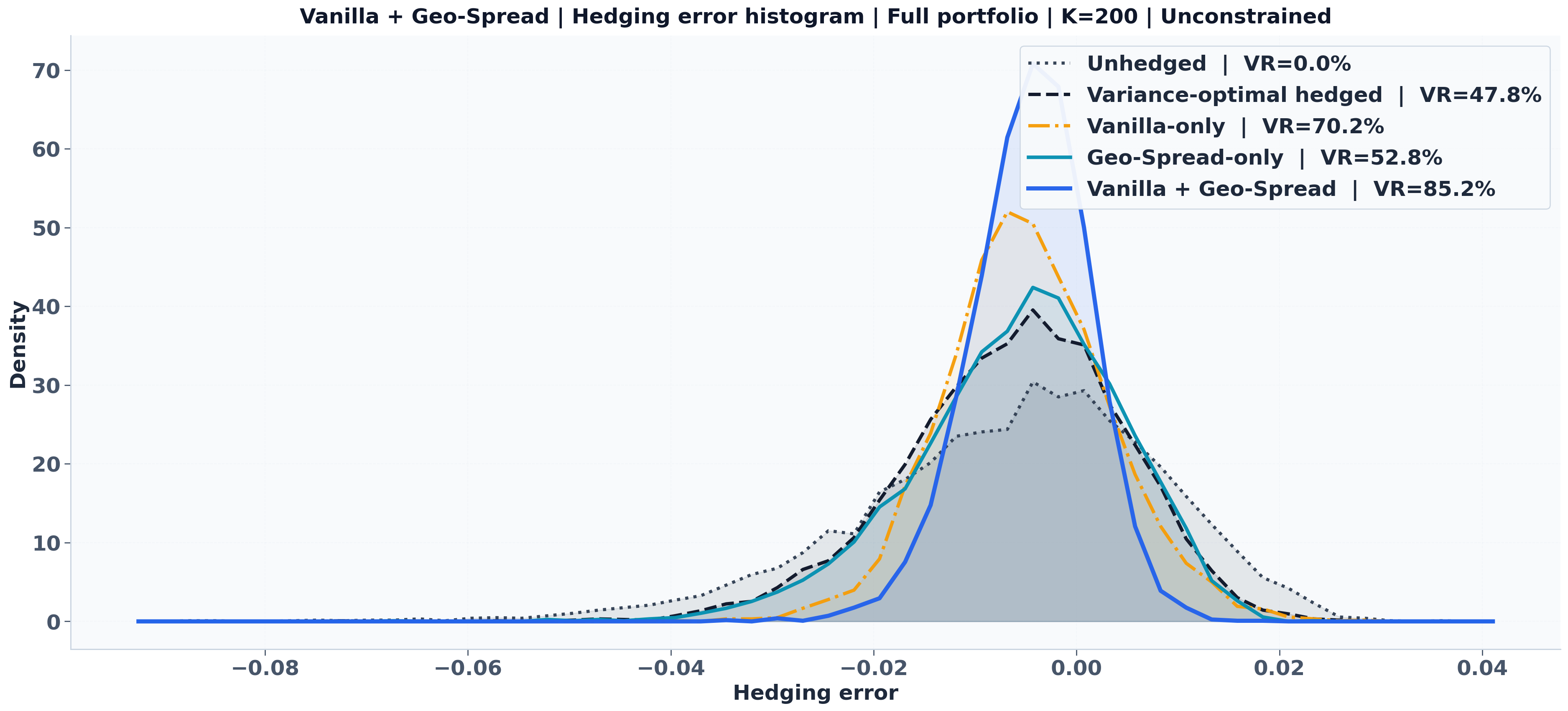}
    \end{subfigure}
    \vspace{0.5em}

    \begin{subfigure}[b]{0.88\textwidth}
        \centering
        \includegraphics[width=\textwidth]{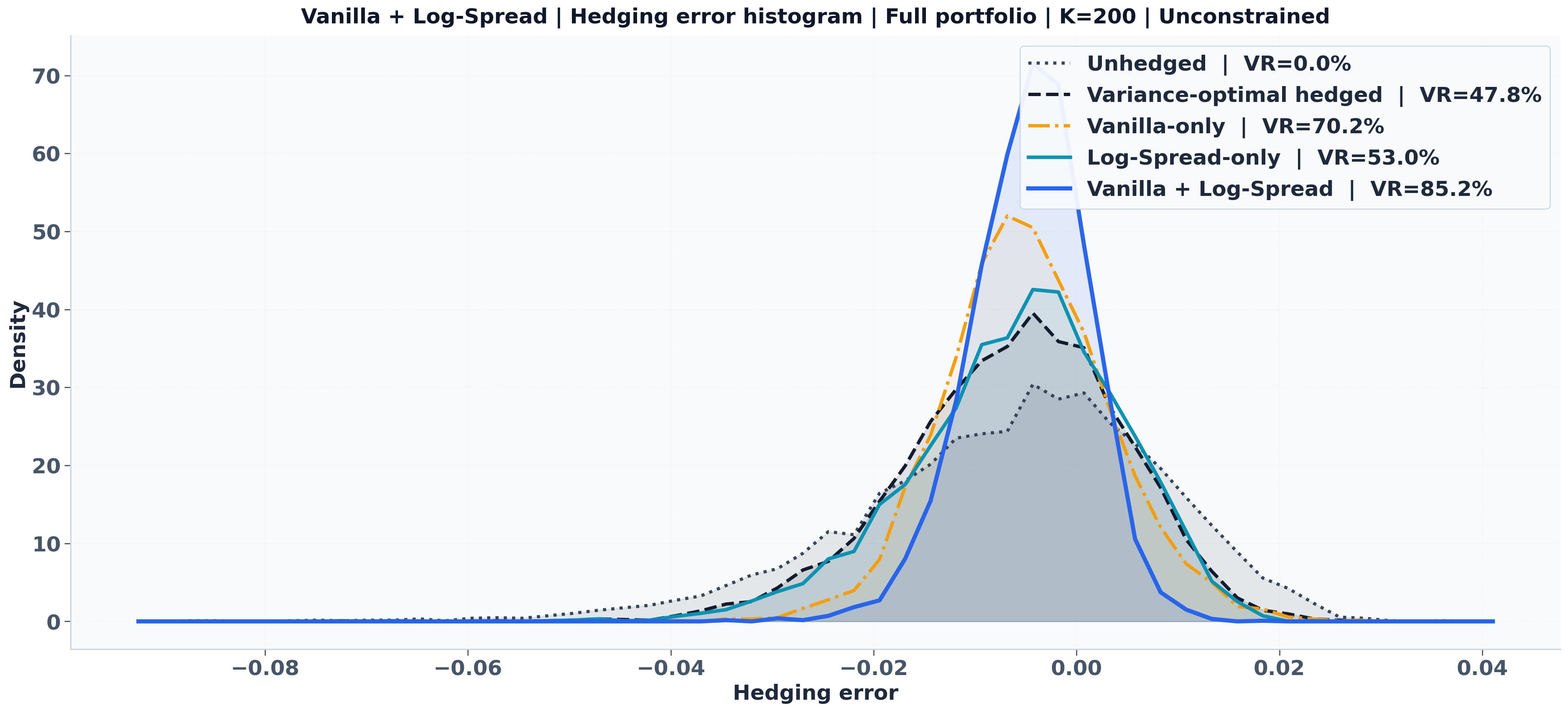}
    \end{subfigure}
\caption{Comparison of the terminal hedging-error distributions generated by the three unconstrained full semi-static portfolios: \emph{vanilla + product} (top), \emph{vanilla + geometric-spread} (middle), and \emph{vanilla + log-spread} (bottom). In each panel, the benchmark dynamic VO hedge is compared with the corresponding static subfamilies and with the combined portfolio. The combined semi-static hedge produces the strongest concentration of the terminal error around zero, indicating the largest variance reduction. }    \label{fig:residuals_covariance}
\end{figure}

\begin{remark}[On the choice of kernel in variance and covariance swaps]
A conceptually important point is that the definition of a realized variance or covariance swap is \emph{not} unique; it depends on the kernel used to aggregate return increments. Historically, early variance swaps were often defined on discretely sampled \emph{simple} returns, whereas market convention later moved toward definitions based on \emph{log}-returns. Since these contracts are typically traded OTC, both the observation grid and the return kernel are contractual specifications rather than mathematical necessities. This distinction is immaterial in a continuous diffusion setting, where simple and log returns agree to first order and both lead to the same quadratic variation in the limit. In the presence of jumps, however, the choice becomes economically and mathematically significant: if $x=\Delta S/S_{-}$ denotes the relative jump size, then the simple-return kernel contributes $x^{2}$, while the log-return kernel contributes $(\log(1+x))^{2}$. Since $(\log(1+x))^{2}>x^{2}$ for negative jumps $x\in(-1,0)$ and $(\log(1+x))^{2}<x^{2}$ for positive jumps $x>0$, the log-return convention places relatively greater weight on downward jumps; see \cite{varianceswapskernel}. The same observation applies, \emph{a fortiori}, to covariance-type contracts. Different kernel choices for the realized co-movement statistic lead to different terminal payoffs, different pathwise decompositions, and therefore different spanning formulas and semi-static hedging strategies. For this reason, the replication identities derived in this section should always be interpreted as being \emph{kernel-specific}: once the realized covariance functional is changed, the associated static instrument family and the resulting variance-optimal projection problem must be modified accordingly. We note that this asymmetry produced on the payoffs presented in Figures \ref{fig:payoff_logspread}, \ref{fig:payoff_geospread} is already visible in the model-independent variance-swap bounds of Hobson and Klimmek \cite{varianceswapskernel}, where the optimal semi-static hedge payoffs are convex and U-shaped around the forward level, but exhibit markedly different left- and right-tail growth under simple-return and log-return conventions.
\end{remark}

\clearpage 
\newpage

\subsection{Dispersion trading}
\label{sec:dispersion-trace}
By simultaneously selling a variance swap on an index and buying variance swaps on the constituents, an investor effectively takes a short position on realized correlation. This type of trade is known as a variance dispersion. A proxy for the implied correlation level sold through a variance dispersion trade is given as the squared ratio of the index variance strike to the average of the constituents’ variance strikes. Note that in order to offset the vega exposure between the two legs, we must adjust the vega notionals of the constituents by a factor equal to the square root of implied correlation. It can be shown that by dynamically trading vega-neutral variance dispersions until maturity, we would almost replicate the payoff of a correlation swap.

The covariance-swap replication results derived above rely fundamentally on the \emph{pairwise} quadratic covariation $\langle Y^{i},Y^{j}\rangle_{t,T}$ between two log-price processes. A natural and practically important extension is to study instruments whose payoff aggregates \emph{all} pairwise co-movements of a basket simultaneously. This is precisely the defining characteristic of a \emph{dispersion trade}: a structured position that isolates the spread between the realized variance of an index and the weighted sum of the realized variances of its constituents.

To ensure exact linear algebraic properties for the quadratic variation, we work within a geometric basket framework. Fix constant index weights $\bm{w}=(w_1,\dots,w_d)^{\top}\in\mathbb{R}_{+}^{d}$ with $\sum_{i=1}^{d}w_i=1$, and define the \emph{geometric index} level $I_t := \prod_{i=1}^d (S_t^i)^{w_i}$. Its log-return process is exactly
\begin{equation}\label{eq:basket-logreturn}
    Y_t^{I} := \log I_t = \sum_{i=1}^{d} w_i Y_t^{i}, \qquad Y_t^{i}:=\log S_t^{i}.
\end{equation}
By the bilinearity of the quadratic variation operator,
\begin{equation}\label{eq:basket-QV}
    \langle Y^{I}\rangle_{t,T}
    = \sum_{i=1}^{d}\sum_{j=1}^{d} w_i w_j \langle Y^{i},Y^{j}\rangle_{t,T}
    = \int_t^T \bm{w}^{\top}\bm{\Sigma}_u\bm{w}\,\mathrm{d}u.
\end{equation}
The \emph{dispersion floating leg} is defined as
\begin{equation}\label{eq:disp-def}
    \mathrm{Disp}_{t,T}
    := \langle Y^{I}\rangle_{t,T} - \sum_{i=1}^{d} w_i^2\,\langle Y^{i}\rangle_{t,T}
    = 2\sum_{1\le i<j\le d} w_i w_j \int_t^T \Sigma_{ij,u}\,\mathrm{d}u.
\end{equation}
Define the symmetric matrix $Q\in\mathbb{S}^{d}$ by
\begin{equation}\label{eq:Q-def2}
    Q_{ij} :=
    \begin{cases}
        w_i w_j, & i\neq j,\\
        0,       & i=j,
    \end{cases}
\end{equation}
so that $\mathrm{Tr}(Q\bm{\Sigma}_t)=2\sum_{i<j}w_iw_j\Sigma_{ij,t}$. Then
\begin{equation}\label{eq:Disp-TrQ}
    \mathrm{Disp}_{t,T} = \int_t^T \mathrm{Tr}(Q\bm{\Sigma}_u)\,\mathrm{d}u.
\end{equation}

Fix a maturity $T>0$ and a valuation time $t\in[0,T)$. Given the geometric index $I$ with log-price process $Y_u^I$ and constituent log-price processes $Y_u^k:=\log S_u^k$ for $u\in[t,T]$, a \emph{dispersion swap} pays at $T$:
\begin{align}\label{eq:disp-swap-payoff}
    H_T^{\mathrm{Disp}}
    &:= \mathrm{Disp}_{t,T} - K_{\mathrm{Disp}},
    \qquad
    K_{\mathrm{Disp}} = \mathbb{E}_t^{\mathbb{Q}}\bigl[\mathrm{Disp}_{t,T}\bigr]
    = \mathbb{E}_t^{\mathbb{Q}}\left[\int_t^T\mathrm{Tr}(Q\bm{\Sigma}_u)\,\mathrm{d}u\right].
\end{align}

\subsubsection*{Step 1: Decomposition of the Dispersion Floating Leg}

Using the bilinearity of quadratic variation and the definition \eqref{eq:disp-def}, we write
\begin{align}\label{eq:disp-QV-split}
    \mathrm{Disp}_{t,T}
    &= \langle Y^{I}\rangle_{t,T}
       - \sum_{i=1}^{d} w_i^2\,\langle Y^{i}\rangle_{t,T}.
\end{align}
Each realized variance $\langle Y^k\rangle_{t,T}$ is replicated by a variance swap (or equivalently a log-contract strip) in the standard Carr--Madan sense. Defining the convexity error function $F(z,z_0):= \frac{z-z_0}{z_0}-\log\bigl(\frac{z}{z_0}\bigr)$ as in \eqref{eq:log-kernel-def}, the pathwise representation of each individual log-return increment is
\begin{align}\label{eq:logret-indiv}
    \langle Y^{k}\rangle_{t,T}
    &= -2\log\Bigl(\frac{S_T^k}{S_t^k}\Bigr)
       + 2\,\frac{S_T^k - S_t^k}{S_t^k}
       - 2\int_t^T\frac{\mathrm{d}S_u^k}{S_u^k}
       + 2\,\frac{\e^{r(T-t)}}{T-t}\left[
            \int_0^{S_t^k}\frac{P^k(K)}{K^2}\,\mathrm{d}K
            + \int_{S_t^k}^{\infty}\frac{C^k(K)}{K^2}\,\mathrm{d}K
         \right],
\end{align}
where $C^k(K)$ and $P^k(K)$ are $t$-priced call and put options on $S^k$ with strike $K$ and maturity $T$. Analogously, the index realized variance $\langle Y^I\rangle_{t,T}$ is spanned by a log-contract strip on the index $I$.

Consequently, the dispersion floating leg \eqref{eq:disp-QV-split} decomposes as
\begin{align}\label{eq:disp-static-dynamic}
    \mathrm{Disp}_{t,T}
    &= \underbrace{\frac{2\e^{r(T-t)}}{T-t}\left[
            \int_0^{I_t}\frac{P^{I}(K)}{K^2}\,\mathrm{d}K
            + \int_{I_t}^{\infty}\frac{C^{I}(K)}{K^2}\,\mathrm{d}K
         \right]}_{\text{index log-contract (short)}}
       \nonumber\\
    &\quad
       - \sum_{i=1}^{d}w_i^2\,\frac{2\e^{r(T-t)}}{T-t}\left[
            \int_0^{S_t^i}\frac{P^{i}(K)}{K^2}\,\mathrm{d}K
            + \int_{S_t^i}^{\infty}\frac{C^{i}(K)}{K^2}\,\mathrm{d}K
         \right]
       \nonumber\\
    &\quad + \text{(dynamic futures rolling strategies)},
\end{align}
where the dynamic terms consist of rolling short positions in $-2\log(S^k_u/S^k_t)$ for each asset $k\in\{I,1,\dots,d\}$. This is the \emph{variance-swap replication} of the dispersion floating leg, extending the classical Carr--Madan result \eqref{eq:log-kernel-def} to the index--component framework.

\subsubsection*{Step 2: Pairwise Covariance Decomposition via Polarization}

To isolate the cross-covariance exposure directly, we apply the polarization identity \eqref{eq:polarizationquadraticvariation} pairwise:
\begin{align}\label{eq:disp-polarization}
    \mathrm{Disp}_{t,T}
    &= 2\sum_{1\le i<j\le d}
       w_i w_j\,\langle Y^i,Y^j\rangle_{t,T}
    \nonumber\\
    &= \sum_{1\le i<j\le d}
       w_i w_j
       \Bigl[
           \langle Y^i\rangle_{t,T}
           + \langle Y^j\rangle_{t,T}
           - \langle Y^i - Y^j\rangle_{t,T}
       \Bigr].
\end{align}
Each pairwise covariance term $\langle Y^i,Y^j\rangle_{t,T}$ can therefore be replicated using the semi-static covariance-swap decomposition of Example~\ref{cov-swapreplicationQuanto}. Equation \eqref{eq:disp-polarization} shows that a dispersion swap is theoretically equivalent to a portfolio of $\binom{d}{2}$ pairwise covariance swaps, each with weight $w_iw_j$, or equivalently to holding long variance swaps on each constituent and short variance swaps on their pairwise log-spreads $Y^i-Y^j=\log(S^i/S^j)$.

For a delta-hedged dispersion trade consisting of long $\alpha_i$ units
of the $i$-th component variance swap and short $1$ unit of the index
variance swap, with $\alpha_i=N_i/N_I$, the instantaneous gamma P\&L
over $(t,t+\D t)$ is
\begin{align}\label{eq:disp-gamma-PL}
    \D\Pi_{t}^{\Gamma}
    &= \sum_{i=1}^{d}\frac{1}{2}\alpha_i\Gamma_i S_{i,t}^2
       \left[\Bigl(\frac{\D S_t^i}{S_t^i}\Bigr)^2 - \sigma_{i,t}^2\,\D t\right]
       - \frac{1}{2}\Gamma_I I_t^2
       \left[\Bigl(\frac{\D I_t}{I_t}\Bigr)^2 - \sigma_{I,t}^2\,\D t\right],
\end{align}
where $\Gamma_i$ (resp.\ $\Gamma_I$) denotes the gamma of the $i$-th
component (resp.\ index) variance swap.  Expanding
$\bigl(\D I_t/I_t\bigr)^2 = \bigl(\sum_k w_k\,\D S_t^k/S_t^k\bigr)^2$ and
using $\sigma_{I,t}^2\,\D t = \bm{w}^{\top}\bm{\Sigma}_t\bm{w}\,\D t$:
\begin{align}\label{eq:disp-gamma-PL-expanded}
    \D\Pi_{t}^{\Gamma}
    &= \sum_{i=1}^{d}\frac{1}{2}S_{i,t}^2
       \left[\Bigl(\frac{\D S_t^i}{S_t^i}\Bigr)^2 - \sigma_{i,t}^2\,\D t\right]
       \bigl(\alpha_i\Gamma_i - w_i^2\Gamma_I\bigr)
    \nonumber\\
    &\quad + \frac{1}{2}\Gamma_I I_t^2\,\mathrm{Tr}(Q\bm{\Sigma}_t)
       \bigl(\hat\rho_t - \rho_t\bigr)\,\D t,
\end{align}
where $\hat\rho_t := \D S_t^i\,\D S_t^j/(S_t^iS_t^j\sigma_{i,t}\sigma_{j,t}\,\D t)$
is the instantaneous realized correlation (assumed equal across all
pairs), and $\rho_t$ is the implied correlation.  Substituting the
variance-swap gamma formula $\Gamma_k = 2/(T S_{k,t}^2)$ gives
\begin{align}\label{eq:disp-gamma-PL-VS}
    \D\Pi_{t}^{\Gamma}
    &= \frac{1}{T}\sum_{i=1}^{d}
       \left[\Bigl(\frac{\D S_t^i}{S_t^i}\Bigr)^2 - \sigma_{i,t}^2\,\D t\right]
       (\alpha_i - w_i^2)
    + \frac{\mathrm{Tr}(Q\bm{\Sigma}_t)}{T}(\rho_t - \hat\rho_t)\,\D t.
\end{align}
Integrating \eqref{eq:disp-gamma-PL-VS} over $[t,T]$ and defining
$\beta^V := (T-t)^{-1}\int_t^T\mathrm{Tr}(Q\,\bm{\sigma}_u\bm{\sigma}_u^{\top})\,\D u$, the integrated gamma P\&L of the
dispersion trade is
\begin{align}\label{eq:disp-gamma-PL-integrated}
    \Pi^{\Gamma}_{\mathrm{Disp}}
    &= \underbrace{\frac{1}{T}\sum_{i=1}^{d}(\alpha_i - w_i^2)
       \int_t^T\left[\Bigl(\frac{\D S_u^i}{S_u^i}\Bigr)^2
                     - \sigma_{i,u}^2\,\D u\right]}_{\text{idiosyncratic residual}}
    + \underbrace{\int_t^T\frac{\mathrm{Tr}(Q\bm{\Sigma}_u)}{T}
       (\rho_u - \hat\rho_u)\,\D u}_{\text{pure correlation P\&L}}.
\end{align}
Equation \eqref{eq:disp-gamma-PL-integrated} is the key structural
decomposition: it shows that the gamma P\&L of a dispersion trade
consists of a pure correlation exposure (the second term) plus an
idiosyncratic residual driven by the mismatch $\alpha_i - w_i^2$.

In the stochastic covariance setting of Section~\ref{sec:MarketSetup},
the instantaneous volatility $\sigma_{i,t}$ of each constituent follows
a diffusion with vol-of-vol coefficient $\xi_i>0$:
$\D\sigma_{i,t} = \mu_{\sigma_i,t}\,\D t + \xi_i\sigma_{i,t}\,\D W_t^{\sigma_i}$,
with $\D\langle W^i, W^{\sigma_i}\rangle_t = \rho_i\,\D t$.
Using the full P\&L expansion \eqref{eq:full-covswap-expansion-product-simplified}
adapted to the present multi-asset setting, the total P\&L of the
delta-hedged dispersion trade decomposes as
\begin{align}\label{eq:disp-total-PL}
    \D\Pi_{\mathrm{Disp},t}
    &= \D\Pi_{t}^{\Gamma}
    \nonumber\\
    &\quad + \sum_{i=1}^{d}\alpha_i
       \Bigl[
           \mathrm{Vega}_i\,\D\sigma_{i,t}
           + \tfrac{1}{2}\,\mathrm{Volga}_i\,\xi_i^2\sigma_{i,t}^2\,\D t
           + \mathrm{Vanna}_i\,\sigma_{i,t}S_{i,t}\rho_i\xi_i\,\D t
       \Bigr]
    \nonumber\\
    &\quad - \Bigl[
           \mathrm{Vega}_I\,\D\sigma_{I,t}
           + \tfrac{1}{2}\,\mathrm{Volga}_I\,\xi_I^2\sigma_{I,t}^2\,\D t
           + \mathrm{Vanna}_I\,\sigma_{I,t}I_t\rho_I\xi_I\,\D t
       \Bigr].
\end{align}
Specializing to variance swaps, for which $\mathrm{Vanna}\equiv 0$,
$\mathrm{Vega}_{\sigma,k} = 2\sigma_{k,t}\tau/T$, and
$\mathrm{Volga}_k = 2\tau/T$ (Appendix~A of \cite{JacquierSlaoui}),
equation \eqref{eq:disp-total-PL} becomes
\begin{align}\label{eq:disp-total-PL-VS}
    \D\Pi_{\mathrm{Disp},t}
    &= \D\Pi_{t}^{\Gamma}
       + \frac{2\tau}{T}
         \left[
             \sum_{i=1}^{d}\alpha_i\sigma_{i,t}^{\mathrm{imp}}\,\D\sigma_{i,t}
             - \sigma_{I,t}^{\mathrm{imp}}\,\D\sigma_{I,t}
         \right]
       + \frac{\tau}{T}
         \left[
             \sum_{i=1}^{d}\alpha_i\xi_i^2\sigma_{i,t}^2
             - \xi_I^2\sigma_{I,t}^2
         \right]\D t.
\end{align}
The first bracket in \eqref{eq:disp-total-PL-VS} is the \emph{vega
P\&L}; the second bracket is the \emph{volga (vomma) P\&L}, which
measures the net vol-of-vol exposure of the dispersion book.  In
matrix notation, defining the vol-of-vol diagonal matrix
$\bm{\Xi}_t:=\mathrm{diag}(\xi_1\sigma_{1,t},\dots,\xi_d\sigma_{d,t})$
and the weight vector $\bm{\alpha}:=(\alpha_1,\dots,\alpha_d)^{\top}$,
the full integrated P\&L of the dispersion trade is
\begin{align}\label{eq:disp-total-integrated}
    \Pi_{\mathrm{Disp}}
    &= \underbrace{\int_t^T\frac{\mathrm{Tr}(Q\bm{\Sigma}_u)}{T}
       (\rho_u - \hat\rho_u)\,\D u}_{\text{pure correlation exposure}}
    + \underbrace{\frac{1}{T}\sum_{i=1}^{d}(\alpha_i - w_i^2)
       \int_t^T\left[
                   \Bigl(\frac{\D S_u^i}{S_u^i}\Bigr)^2
                   - \sigma_{i,u}^2\,\D u
               \right]}_{\text{idiosyncratic residual}}
    \nonumber\\[4pt]
    &\quad + \underbrace{\frac{2\tau}{T}\int_t^T
         \left[
             \bm{\alpha}^{\top}(\bm{\sigma}_u^{\mathrm{imp}}\odot\D\bm{\sigma}_u)
             - \sigma_{I,u}^{\mathrm{imp}}\,\D\sigma_{I,u}
         \right]}_{\text{vega P\&L}}
    + \underbrace{\frac{\tau}{T}\int_t^T
         \left[
             \bm{\alpha}^{\top}\bm{\Xi}_u^2\mathbf{1}
             - \xi_I^2\sigma_{I,u}^2
         \right]\D u}_{\text{volga P\&L}}.
\end{align}

\subsubsection*{Step 5: Three Risk-Flat Weighting Strategies}

The choice of $\bm{\alpha}=(\alpha_1,\dots,\alpha_d)^{\top}$ determines
the risk profile of the dispersion, isolating which terms in
\eqref{eq:disp-total-integrated} survive.

\paragraph{(i) Gamma-Flat Strategy.}
Set $\alpha_i = w_i^2$ for all $i=1,\dots,d$.  Then the idiosyncratic
residual in \eqref{eq:disp-total-integrated} vanishes identically, and
the integrated gamma P\&L reduces to a \emph{pure correlation P\&L}:
\begin{align}\label{eq:gamma-flat}
    \Pi^{\Gamma}_{\mathrm{Disp}}\big|_{\alpha_i = w_i^2}
    &= \int_t^T\frac{\mathrm{Tr}(Q\bm{\Sigma}_u)}{T}(\rho_u - \hat\rho_u)\,\D u
     = \beta^V\,(\bar\rho - \bar{\hat\rho}),
\end{align}
where $\beta^V := (T-t)^{-1}\mathrm{Tr}\bigl(Q\int_t^T\bm{\sigma}_u\bm{\sigma}_u^{\top}\D u\bigr)$
and $\bar\rho$, $\bar{\hat\rho}$ denote $\beta_u$-weighted time
averages of implied and realized correlation, respectively.  The
remaining P\&L consists solely of the volga term, explaining the
empirically observed spread between the implied correlation of a
dispersion trade and the fair strike of a pure correlation swap
entirely through the vol-of-vol parameter $\xi_i$ (cf.\ \cite{JacquierSlaoui}).

\paragraph{(ii) Vega-Flat Strategy.}
The vega notional-matching condition
$N_i\Upsilon_{\sigma,i} = N_I\Upsilon_{\sigma,I}w_i$ (with
$\Upsilon_{\sigma,k} = 2\sigma_{k,t}\tau/T$ for variance swaps) gives
\begin{align}\label{eq:alpha-vega-flat}
    \alpha_i = \frac{w_i\sigma_{I,t}}{\sigma_{i,t}}, \qquad i=1,\dots,d.
\end{align}
Under \eqref{eq:alpha-vega-flat} the vega P\&L in
\eqref{eq:disp-total-integrated} vanishes and the residual gamma
approximation error is
\begin{align}\label{eq:vega-flat-error}
    \alpha_i - w_i^2
    = w_i\frac{\sigma_{I,t}}{\sigma_{i,t}} - w_i^2
    = w_i^2\left(\frac{\sigma_{I,t}}{w_i\sigma_{i,t}} - 1\right),
\end{align}
which is negligible whenever $\sigma_{I,t}\approx\sum_j w_j\sigma_{j,t}$,
i.e.\ when the index volatility is close to the weighted average of
constituent volatilities.  This strategy represents the no-arbitrage
boundary condition for the dispersion book and it satisfies
$\sum_i\alpha_i\sigma_{i,t}^2/\sigma_{I,t}^2 = 1$.

\paragraph{(iii) Theta-Flat (Gamma-Neutral) Strategy.}
Setting $\Pi^{\Gamma}_{\mathrm{Disp}}\equiv 0$ requires
\begin{align}\label{eq:alpha-theta-flat}
    \alpha_i
    = \frac{({\D I_t}/{I_t})^2 - \sigma_{I,t}^2\,\D t}
           {\sum_{k=1}^{d}\bigl[({\D S_t^k}/{S_t^k})^2
                                 - \sigma_{k,t}^2\,\D t\bigr]},
    \qquad i=1,\dots,d.
\end{align}
Under \eqref{eq:alpha-theta-flat} the gamma P\&L is zero and the
entire P\&L is driven by the volatility greeks in
\eqref{eq:disp-total-PL-VS}, making this a pure vol-of-vol carry
trade.

\section{Fourier Representation of the Multivariate GKW Decomposition}
\label{sec:fourierrepresentation}

The preliminary numerical results of Section~\ref{sec:replicationtheory} already illustrate the central qualitative message of our framework: once the auxiliary instruments are chosen in a manner consistent with the spanning formulas, the residual hedging error can be reduced substantially relative to the purely dynamic benchmark, and the resulting optimal static portfolios exhibit the structures predicted by the replication identities. Sections~\ref{sec:VO-framework}--\ref{sec:replicationtheory} show that the semi-static variance-optimal hedging problem reduces, in general, to two fundamental tasks: first, the computation of the dynamic Galtchouk--Kunita--Watanabe integrands $\bm{\theta}^i$ associated with the target claim and the auxiliary claims; second, the evaluation of the quadratic quantities $A$, $B$, and $C$ in \eqref{eq:maincomponenetsstatic}, which determine the optimal static allocation through the finite-dimensional outer problem. In the present section, we develop a Fourier-based representation that makes both tasks analytically tractable in the affine and quadratic stochastic covariance models considered later. Of course, one could also estimate the relevant conditional expectations and transition laws of $\bm{S}$, as well as the value processes $H^0,\dots,H^n$, by Monte Carlo simulation and then approximate the GKW decomposition by sequential backward regression; however, such an approach is, in principle, considerably more expensive, especially in higher dimensions \cite{Schweizer2008}.

The main idea is a multivariate extension of the Fourier--GKW methodology of \cite{Semi_static_Fourier}. We represent the relevant payoff class through an inverse Laplace--Fourier transform and then exploit both the linearity of conditional expectation and the linearity of the GKW projection. This yields a transfer principle: the GKW decomposition of a sufficiently regular claim may be obtained by integrating the GKW decompositions of the exponential basis payoffs $\e^{\bm{u}^\top \bm{Y}_T}$. In the model classes studied below, these exponential claims are analytically tractable because their conditional expectations, covariations, and hence the quantities entering the semi-static hedging problem admit explicit transform representations.

\subsection{Probabilistic setup and complex domains}

Let $\bm{S}=(S^1,\dots,S^d)^\top$ be the discounted square-integrable $\mathbb{Q}$-martingale price vector and let

\begin{equation*}
    \bm{Y}_t:=\log \bm{S}_t
\end{equation*}
denote the log-price process, where the logarithm is taken componentwise. We assume that the pair $(\bm{Y}_t,\bm{V}_t)_{t\in[0,T]}$ is Markov, where $\bm{V}$ denotes the relevant volatility or covariance state variable. In particular, $\bm{V}_t=\bm{X}_t$ in the quadratic Gaussian specification and $\bm{V}_t=\bm{\Sigma}_t$ in the affine/Wishart specification introduced in Section~\ref{sec:MarketSetup}. Fourier--Laplace methods require exponential moments in a complex strip. Fix a damping vector $\bm{R}\in\mathbb{R}^d$ such that

\begin{equation*}
    \mathbb{E}\big[\e^{2\bm{R}^\top \bm{Y}_T}\big]<\infty.
\end{equation*}
Define the vertical strip

\begin{equation*}
    \mathcal{S}(\bm{R})
    :=
    \bigl\{\bm{u}\in\mathbb{C}^d:\Re(\bm{u})=\bm{R}\bigr\}.
\end{equation*}
For $\bm{u}\in\mathcal{S}(\bm{R})$, set
\begin{equation}\label{eq:cond-mgf}
    M_t(\bm{u})
    :=
    \mathbb{E}\big[\e^{\bm{u}^\top \bm{Y}_T}\mid\mathcal{F}_t\big],
    \qquad t\in[0,T].
\end{equation}
Then $M(\bm{u})=(M_t(\bm{u}))_{t\in[0,T]}$ is a complex-valued square-integrable martingale for each
$\bm{u}\in\mathcal{S}(\bm{R})$.

\subsection{Integral representation of payoffs and prices}

Let $\eta=h(\bm{Y}_T)$ be a European payoff with $h:\mathbb{R}^d\to\mathbb{R}$. Assume that the bilateral
multivariate Laplace transform $\hat h$ exists on $\mathcal{S}(\bm{R})$ and is absolutely integrable along the
strip. By the multidimensional inverse Laplace transform,
\begin{equation}
\label{eq:LaplaceMVpayoffs}
    h(\bm{y})
    =
    \frac{1}{(2\pi i)^d}
    \int_{\bm{R}-i\infty}^{\bm{R}+i\infty}
    \e^{\bm{u}^\top \bm{y}}\,
    \hat h(\bm{u})\,\D\bm{u}
    =
    \int_{\mathcal{S}(\bm{R})} \e^{\bm{u}^\top \bm{y}}\,\zeta(\D\bm{u}),
\end{equation}
where $\zeta(\D\bm{u}):=(2\pi i)^{-d}\hat h(\bm{u})\,\D\bm{u}$ is a complex measure on $\mathcal{S}(\bm{R})$.

Let $H_t:=\mathbb{E}[\eta\mid\mathcal{F}_t]$ be the discounted value process. Under a mild measurability
assumption on $(t,\omega,\bm{u})\mapsto M_t(\bm{u})(\omega)$, the linearity of conditional expectation transfers
\eqref{eq:LaplaceMVpayoffs} from payoff to price.

\begin{proposition}[Fourier pricing identity]
\label{propositionfourierpricing}
Assume that $\bm{u}\mapsto M_t(\bm{u})$ is $\mathcal{B}(\mathcal{S}(\bm{R}))$-measurable for each $t$, and that

\begin{equation*}
    \int_{\mathcal{S}(\bm{R})}
    \mathbb{E}\big[|M_T(\bm{u})|^2\big]\,
    |\zeta|(\D\bm{u})<\infty.
\end{equation*}
Then
\begin{equation}
\label{eq:fourierprice}
    H_t
    =
    \int_{\mathcal{S}(\bm{R})} M_t(\bm{u})\,\zeta(\D\bm{u}),
    \qquad t\in[0,T],\quad \mathbb{Q}\text{-a.s.}
\end{equation}
\end{proposition}

The representation \eqref{eq:fourierprice} generalizes classical transform pricing, for example Carr--Madan-type
formulas, to the multivariate setting. Numerically, the high-dimensional integration can be treated with
dimension-adaptive quadrature, for example sparse grids, when $d$ is moderate; see, e.g., \cite{Bayer2023}.

Proposition~\ref{propositionfourierpricing} becomes powerful for hedging once we show that the GKW
projection commutes with Fourier integration. We first recall the GKW decomposition of the basis
martingales $M(\bm{u})$. For each $\bm{u}\in\mathcal{S}(\bm{R})$, let
\begin{equation}\label{eq:GKW-basis}
    M_t(\bm{u})
    =
    M_0(\bm{u})
    +\int_0^t \bm{\theta}_s(\bm{u})^\top\,\D\bm{S}_s
    +L_t(\bm{u}),
\end{equation}
be the multivariate GKW decomposition of $M(\bm{u})$ with respect to $\bm{S}$, where
$L(\bm{u})\perp \bm{S}$. Formally, if we can integrate \eqref{eq:GKW-basis} over $\bm{u}$ and interchange the order of integration, then \eqref{eq:fourierprice} yields the GKW decomposition of $H$ with

\begin{equation*}
    \bm{\theta}_t
    =
    \int_{\mathcal{S}(\bm{R})}\bm{\theta}_t(\bm{u})\,\zeta(\D\bm{u}),
    \qquad
    L_t
    =
    \int_{\mathcal{S}(\bm{R})}L_t(\bm{u})\,\zeta(\D\bm{u}).
\end{equation*}

The theorem below provides a rigorous statement.

\begin{theorem}[Fourier Representation of the Multivariate GKW Decomposition]
\label{thm:FourierGKW}
Let $\bm{S}=(S^1,\dots,S^d)^\top$ be a square-integrable $\mathbb{Q}$-martingale. Let $\eta=h(\bm{Y}_T)$ admit the Laplace representation \eqref{eq:LaplaceMVpayoffs} with complex measure $\zeta$ satisfying
\begin{equation*}
    |\zeta|(\mathcal{S}(\bm{R}))<\infty.
\end{equation*}
Define the predictable increasing scalar process
\begin{equation}
\label{eq:scalar-control-measure}
    A_t
    :=
    t+\sum_{k=1}^d \llangle S^k,S^k\rrangle_t,
    \qquad t\in[0,T].
\end{equation}
Then there exists a predictable Hermitian positive-semidefinite matrix-valued process $\bm{C}=(\bm{C}_t)_{t\in[0,T]}$ such that
\begin{equation}
\label{eq:matrix-density-control}
    \mathrm{d}\llangle \bm{S},\bm{S}\rrangle_t
    =
    \bm{C}_t\,\mathrm{d}A_t
\end{equation}
componentwise. Assume the following integrability and structural conditions hold:
\begin{enumerate}
    \item[i)] For each $\bm{u}\in\mathcal{S}(\bm{R})$, the martingale $M(\bm{u})$ defined in \eqref{eq:cond-mgf} belongs to $\mathcal{H}^2_{\mathbb{C}}$ and admits the GKW decomposition
    \begin{equation}
    \label{eq:GKW-basis-thm1}
        M_t(\bm{u})
        =
        M_0(\bm{u})
        +\int_0^t \bm{\theta}_s(\bm{u})^\top\,\mathrm{d}\bm{S}_s
        +L_t(\bm{u}),
        \qquad t\in[0,T],
    \end{equation}
    where $L(\bm{u})$ is a complex-valued martingale strongly orthogonal to $\bm{S}$, i.e., $L(\bm{u})\perp \bm{S}$.

    \item[ii)] There exist versions of the fields
    $$
        (t,\omega,\bm{u})\mapsto \bm{\theta}_t(\bm{u})(\omega),
        \qquad
        (t,\omega,\bm{u})\mapsto L_t(\bm{u})(\omega),
    $$
    which are $\mathcal{P}\otimes\mathcal{B}(\mathcal{S}(\bm{R}))$-measurable. Let $\bm{\theta}^*$ denote the conjugate transpose. We assume the uniform bounds
    \begin{equation}
    \label{eq:theta-integrability-thm1}
        \int_{\mathcal{S}(\bm{R})}
        \mathbb{E}\left[
            \int_0^T
            \bm{\theta}_t(\bm{u})^*
            \bm{C}_t
            \bm{\theta}_t(\bm{u})\,\mathrm{d}A_t
        \right]
        |\zeta|(\mathrm{d}\bm{u})
        <\infty,
    \end{equation}
    and
    \begin{equation}
    \label{eq:L-integrability-thm1}
        \int_{\mathcal{S}(\bm{R})}
        \sup_{t\le T}\mathbb{E}\bigl[|L_t(\bm{u})|^2\bigr]\,
        |\zeta|(\mathrm{d}\bm{u})
        <\infty.
    \end{equation}
\end{enumerate}
Then the claim value process $H_t=\mathbb{E}[\eta\mid\mathcal{F}_t]$ admits the GKW decomposition
\begin{equation}
\label{eq:representation1-thm1}
    H_t
    =
    H_0
    +\int_0^t \bm{\theta}_s^\top\,\mathrm{d}\bm{S}_s
    +L_t,
    \qquad t\in[0,T],
\end{equation}
where the dynamic hedge and the orthogonal residual are given by
\begin{equation}
\label{eq:representation2-thm1}
    \bm{\theta}_t
    =
    \int_{\mathcal{S}(\bm{R})}\bm{\theta}_t(\bm{u})\,\zeta(\mathrm{d}\bm{u}),
    \qquad
    L_t
    =
    \int_{\mathcal{S}(\bm{R})}L_t(\bm{u})\,\zeta(\mathrm{d}\bm{u}),
\end{equation}
and the stochastic Fubini interchange holds:
\begin{equation}
\label{eq:Fubiniinterchange-thm1}
    \int_0^t
    \Bigg(
        \int_{\mathcal{S}(\bm{R})}\bm{\theta}_s(\bm{u})\,\zeta(\mathrm{d}\bm{u})
    \Bigg)^\top
    \mathrm{d}\bm{S}_s
    =
    \int_{\mathcal{S}(\bm{R})}
    \Bigg(
        \int_0^t \bm{\theta}_s(\bm{u})^\top\,\mathrm{d}\bm{S}_s
    \Bigg)\zeta(\mathrm{d}\bm{u}),
    \qquad \mathbb{Q}\text{-a.s.}
\end{equation}
\end{theorem}

\begin{proof}
The proof extends the univariate methodologies of \cite[Theorems 4.1 and 4.2]{Semi_static_Fourier} to the present multi-asset setting. 

The Radon--Nikodym decomposition of the complex measure $\zeta$ reads
$$
    \mathrm{d}\zeta(\bm{u})=h(\bm{u})\,\mathrm{d}|\zeta|(\bm{u}),
    \qquad |h(\bm{u})|=1.
$$
For each $\bm{u}\in\mathcal{S}(\bm{R})$, define the stochastic integral
$$
    I_t(\bm{u})
    :=
    \int_0^t \bm{\theta}_s(\bm{u})^\top\,\mathrm{d}\bm{S}_s,
    \qquad t\in[0,T].
$$
By \eqref{eq:GKW-basis-thm1}, $I(\bm{u})\in\mathcal{H}^2_{\mathbb{C}}$ for every $\bm{u}\in\mathcal{S}(\bm{R})$. Applying the It\^o isometry for vector stochastic integrals together with the matrix density \eqref{eq:matrix-density-control} yields
\begin{align}
    \mathbb{E}\bigl[|I_T(\bm{u})|^2\bigr]
    &=
    \mathbb{E}\left[
        \int_0^T
        \bm{\theta}_t(\bm{u})^*\,
        \mathrm{d}\llangle \bm{S},\bm{S}\rrangle_t\,
        \bm{\theta}_t(\bm{u})
    \right] \notag\\
    &=
    \mathbb{E}\left[
        \int_0^T
        \bm{\theta}_t(\bm{u})^*\,
        \bm{C}_t\,
        \bm{\theta}_t(\bm{u})\,\mathrm{d}A_t
    \right].
    \label{eq:ito-isometry-thm1}
\end{align}
Substituting this into the integrability condition \eqref{eq:theta-integrability-thm1} gives
\begin{equation}
\label{eq:I-square-integrable-thm1}
    \int_{\mathcal{S}(\bm{R})}
    \mathbb{E}\bigl[|I_T(\bm{u})|^2\bigr]
    |\zeta|(\mathrm{d}\bm{u})
    <\infty.
\end{equation}

We now verify that the aggregated process
$$
    \bm{\theta}_t
    :=
    \int_{\mathcal{S}(\bm{R})}\bm{\theta}_t(\bm{u})\,\zeta(\mathrm{d}\bm{u})
$$
defines an admissible strategy in $L^2(\bm{S})$. Since $\bm{C}_t$ is Hermitian and positive semidefinite, applying the Cauchy--Schwarz inequality to the seminorm induced by $\bm{C}_t$ provides, for $\mathrm{d}A\otimes \mathrm{d}\mathbb{Q}$-a.e. $(t,\omega)$,
\begin{align}
    &\Bigg(
        \int_{\mathcal{S}(\bm{R})}\bm{\theta}_t(\bm{u})\,\zeta(\mathrm{d}\bm{u})
    \Bigg)^*
    \bm{C}_t
    \Bigg(
        \int_{\mathcal{S}(\bm{R})}\bm{\theta}_t(\bm{u})\,\zeta(\mathrm{d}\bm{u})
    \Bigg) \notag\\
    &\qquad=
    \Bigg(
        \int_{\mathcal{S}(\bm{R})}h(\bm{u})\bm{\theta}_t(\bm{u})\,|\zeta|(\mathrm{d}\bm{u})
    \Bigg)^*
    \bm{C}_t
    \Bigg(
        \int_{\mathcal{S}(\bm{R})}h(\bm{u})\bm{\theta}_t(\bm{u})\,|\zeta|(\mathrm{d}\bm{u})
    \Bigg) \notag\\
    &\qquad\le
    |\zeta|(\mathcal{S}(\bm{R}))
    \int_{\mathcal{S}(\bm{R})}
    \bm{\theta}_t(\bm{u})^*\bm{C}_t\bm{\theta}_t(\bm{u})\,
    |\zeta|(\mathrm{d}\bm{u}).
    \label{eq:CS-control-thm1}
\end{align}
Integrating \eqref{eq:CS-control-thm1} over $[0,T]\times\Omega$ and invoking \eqref{eq:theta-integrability-thm1}, we obtain
$$
    \mathbb{E}\left[
        \int_0^T \bm{\theta}_t^*\,\mathrm{d}\llangle \bm{S},\bm{S}\rrangle_t\,\bm{\theta}_t
    \right]
    =
    \mathbb{E}\left[
        \int_0^T \bm{\theta}_t^* \bm{C}_t \bm{\theta}_t\,\mathrm{d}A_t
    \right]
    <\infty.
$$
Therefore, $\bm{\theta}\in L^2(\bm{S})$.

Next, \eqref{eq:I-square-integrable-thm1} and the $\mathcal{P}\otimes\mathcal{B}(\mathcal{S}(\bm{R}))$-measurability of $(t,\omega,\bm{u})\mapsto \bm{\theta}_t(\bm{u})(\omega)$ satisfy the conditions of the stochastic Fubini theorem with scalar control measure $A$. Hence, there exists an optional field $(t,\omega,\bm{u})\mapsto I_t(\bm{u})(\omega)$, serving as a version of $\int_0^\cdot \bm{\theta}_s(\bm{u})^\top\,\mathrm{d}\bm{S}_s$ for every $\bm{u}$, such that
\begin{equation}
\label{eq:stochastic-fubini-proof-thm1}
    \int_0^t
    \Bigg(
        \int_{\mathcal{S}(\bm{R})}\bm{\theta}_s(\bm{u})\,\zeta(\mathrm{d}\bm{u})
    \Bigg)^\top
    \mathrm{d}\bm{S}_s
    =
    \int_{\mathcal{S}(\bm{R})} I_t(\bm{u})\,\zeta(\mathrm{d}\bm{u}),
    \qquad t\in[0,T],
\end{equation}
$\mathbb{Q}$-a.s., which proves \eqref{eq:Fubiniinterchange-thm1}.

Now define the aggregated residual process
$$
    L_t^{*}
    :=
    \int_{\mathcal{S}(\bm{R})}L_t(\bm{u})\,\zeta(\mathrm{d}\bm{u}),
    \qquad t\in[0,T].
$$
By the bound in \eqref{eq:L-integrability-thm1}, for every $t\in[0,T]$,
\begin{align}
    \int_{\mathcal{S}(\bm{R})}
    \mathbb{E}\bigl[|L_t(\bm{u})|\bigr]\,|\zeta|(\mathrm{d}\bm{u})
    &\le
    |\zeta|(\mathcal{S}(\bm{R}))^{1/2}
    \left(
        \int_{\mathcal{S}(\bm{R})}
        \mathbb{E}\bigl[|L_t(\bm{u})|^2\bigr]\,|\zeta|(\mathrm{d}\bm{u})
    \right)^{1/2}
    <\infty.
    \label{eq:L-L1-integrability-thm1}
\end{align}
The integrability condition \eqref{eq:L-L1-integrability-thm1} permits the application of the conditional Fubini theorem. For any $0\le s\le t\le T$,
\begin{align*}
    \mathbb{E}\bigl[L_t^{*}\mid\mathcal{F}_s\bigr]
    &=
    \int_{\mathcal{S}(\bm{R})}
    \mathbb{E}\bigl[L_t(\bm{u})\mid\mathcal{F}_s\bigr]\,
    \zeta(\mathrm{d}\bm{u})\\
    &=
    \int_{\mathcal{S}(\bm{R})}L_s(\bm{u})\,\zeta(\mathrm{d}\bm{u})
    =
    L_s^{*}.
\end{align*}
Thus, $L^{*}$ is a complex-valued square-integrable martingale.

Integrating \eqref{eq:GKW-basis-thm1} with respect to $\zeta(\mathrm{d}\bm{u})$ and using Proposition~\ref{propositionfourierpricing} together with the Fubini interchange \eqref{eq:stochastic-fubini-proof-thm1}, we decompose the target claim:
\begin{align*}
    H_t
    &=
    \int_{\mathcal{S}(\bm{R})}M_t(\bm{u})\,\zeta(\mathrm{d}\bm{u})\\
    &=
    \int_{\mathcal{S}(\bm{R})}M_0(\bm{u})\,\zeta(\mathrm{d}\bm{u})
    +
    \int_{\mathcal{S}(\bm{R})}
    \Bigg(
        \int_0^t \bm{\theta}_s(\bm{u})^\top\,\mathrm{d}\bm{S}_s
    \Bigg)\zeta(\mathrm{d}\bm{u})
    +
    \int_{\mathcal{S}(\bm{R})}L_t(\bm{u})\,\zeta(\mathrm{d}\bm{u})\\
    &=
    H_0
    +
    \int_0^t \bm{\theta}_s^\top\,\mathrm{d}\bm{S}_s
    +
    L_t^{*}.
\end{align*}

It remains to prove strong orthogonality, $L^{*}\perp \bm{S}$. Fix $k\in\{1,\dots,d\}$. Since $L(\bm{u})\perp \bm{S}$, the product $S^k L(\bm{u})$ is a martingale for every $\bm{u}\in\mathcal{S}(\bm{R})$. By the Cauchy-Schwarz inequality, for every $t\in[0,T]$,
\begin{align}
    \int_{\mathcal{S}(\bm{R})}
    \mathbb{E}\bigl[|S_t^kL_t(\bm{u})|\bigr]\,
    |\zeta|(\mathrm{d}\bm{u})
    &\le
    \|S_t^k\|_{L^2(\mathbb{Q})}
    \int_{\mathcal{S}(\bm{R})}\|L_t(\bm{u})\|_{L^2(\mathbb{Q})}\,|\zeta|(\mathrm{d}\bm{u}) \notag\\
    &\le
    \|S_t^k\|_{L^2(\mathbb{Q})}\,
    |\zeta|(\mathcal{S}(\bm{R}))^{1/2}
    \left(
        \int_{\mathcal{S}(\bm{R})}
        \mathbb{E}\bigl[|L_t(\bm{u})|^2\bigr]\,|\zeta|(\mathrm{d}\bm{u})
    \right)^{1/2}
    <\infty.
    \label{eq:orthogonality-integrability-thm1}
\end{align}
Because this bound is finite \eqref{eq:orthogonality-integrability-thm1}, conditional Fubini applies once more. For $0\le s\le t\le T$,
\begin{align*}
    \mathbb{E}\bigl[S_t^kL_t^{*}\mid\mathcal{F}_s\bigr]
    &=
    \int_{\mathcal{S}(\bm{R})}
    \mathbb{E}\bigl[S_t^kL_t(\bm{u})\mid\mathcal{F}_s\bigr]\,
    \zeta(\mathrm{d}\bm{u})\\
    &=
    \int_{\mathcal{S}(\bm{R})}S_s^kL_s(\bm{u})\,\zeta(\mathrm{d}\bm{u})
    =
    S_s^kL_s^{*}.
\end{align*}
Thus, $S^k L^{*}$ is a martingale for every $k$, implying $L^{*}\perp S^k$ for all $k$, and therefore $L^{*}\perp \bm{S}$.

Consequently,
$$
    H_t
    =
    H_0+\int_0^t \bm{\theta}_s^\top\,\mathrm{d}\bm{S}_s+L_t^{*}
$$
is a valid GKW decomposition of $H$ with respect to $\bm{S}$. By the uniqueness of the GKW decomposition in $L^2(\mathbb{Q})$, the integrand and the residual are precisely the processes defined in \eqref{eq:representation2-thm1}. This concludes the proof.
\end{proof}

\subsection{Payoff transforms}\label{subsec:payofftransforms}

The Fourier--GKW machinery requires explicit or numerically stable expressions for $\hat h(\bm{u})$ in
\eqref{eq:LaplaceMVpayoffs}. We collect the transforms used throughout the paper.

\paragraph{\textit{Spread options.}}
For the two-asset spread payoff

\begin{equation*}
    h(y_1,y_2)=(\e^{y_1}-\e^{y_2}-K)^+,
\end{equation*}
one may use a damped Laplace transform on a strip satisfying $R_2<0$ and $R_1+R_2>1$; see, e.g., \cite{Hubalek2006} and Example~\ref{example:LaplaceSpreadBasket} below. For a geometric basket call

\begin{equation*}
    h(\bm{y})
    =
    \Big(\exp\Big(\sum_{i=1}^d \bar w_i y_i\Big)-K\Big)^+,
    \qquad
    \bar w_i:=\frac{w_i}{d},
\end{equation*}
the payoff depends only on the one-dimensional factor $\sum_{i=1}^d \bar w_i Y_T^i$. The transform therefore reduces to the univariate call kernel evaluated along the ray

\begin{equation*}
    \bm{u}=z(\bar w_1,\dots,\bar w_d)^\top.
\end{equation*}

\paragraph{\textit{Polynomial/log payoffs.}}
Quadratic terms such as

\begin{equation*}
    h(\bm{y})=y_i y_j
\end{equation*}
are more efficiently handled by differentiating the conditional transform $\bm{u}\mapsto M_t(\bm{u})$ rather than by explicit Laplace inversion:

\begin{equation*}
    \mathbb{E}[Y_T^iY_T^j\mid\mathcal{F}_t]
    =
    \frac{\partial^2}{\partial u_i\partial u_j}M_t(\bm{u})\Big|_{\bm{u}=0}.
\end{equation*}

\begin{example}[Product and quanto options]
\label{ex:ProductOptions}
A particularly tractable class consists of separable payoffs

\begin{equation*}
    h(\bm{x})
    =
    \prod_{j=1}^d h_j(x_j),
    \qquad
    \bm{x}=(x_1,\dots,x_d)^\top,
\end{equation*}
which naturally arise in quanto and product contracts. Their key advantage is computational: the
multivariate Laplace transform factorizes,

\begin{equation*}
    \hat h(\bm{u})
    =
    \int_{\mathbb{R}^d}
    \e^{\bm{u}^\top \bm{x}}
    \prod_{j=1}^d h_j(x_j)\,\D\bm{x}
    =
    \prod_{j=1}^d
    \Big(
        \int_{\mathbb{R}} \e^{u_j x_j} h_j(x_j)\,\D x_j
    \Big)
    =
    \prod_{j=1}^d \hat h_j(u_j).
\end{equation*}
Thus the multivariate kernel is obtained from standard univariate transforms, for example call or put transforms, thereby avoiding high-dimensional integration at the payoff level. In Section~\ref{sec:replicationtheory}, products of vanilla option payoffs appear as the canonical building blocks for bivariate covariance replication.
\end{example}

\begin{example}[Laplace transforms for multi-asset payoffs]
\label{example:LaplaceSpreadBasket}
To implement \eqref{eq:fourierprice} and Theorem~\ref{thm:FourierGKW}, one needs explicit payoff transforms $\hat h$. Throughout this example, for a payoff $h:\mathbb{R}^M\to\mathbb{R}$ we use the bilateral Laplace transform convention
\begin{align}
    \hat h(\bm u)
    :=
    \int_{\mathbb{R}^M} \e^{-\bm u^\top \bm x} h(\bm x)\,\D\bm x,
    \qquad
    \bm u\in\mathbb{C}^M,
\end{align}
whenever the integral is well defined, so that the inversion formula reads
\begin{align}
    h(\bm x)
    =
    \frac{1}{(2\pi i)^M}
    \int_{\mathcal S(\bm R)}
    \e^{\bm u^\top \bm x}\hat h(\bm u)\,\D\bm u,
    \qquad
    \mathcal S(\bm R):=\{\bm u\in\mathbb{C}^M:\Re(\bm u)=\bm R\}.
\end{align}
Building on \cite{Hurd2010,Bossu2021}, we collect the transforms used throughout the paper.

Fix $K>0$ and write $\Gamma(\cdot)$ for the Gamma function.

\begin{enumerate}
    \item[i)] \textbf{Two-asset spread option.}
    For
    \begin{align}
        h_K^{\mathrm{spr}}(x_1,x_2)
        &:=
        \bigl(\e^{x_1}-\e^{x_2}-K\bigr)^+,
    \end{align}
    choose $\bm R=(R_1,R_2)\in\mathbb{R}^2$ such that $R_2<0$ and $R_1+R_2>1$. Then
    \begin{align}
        h_K^{\mathrm{spr}}(x_1,x_2)
        =
        \frac{1}{(2\pi i)^2}
        \int_{\mathcal S(\bm R)}
        \e^{u_1x_1+u_2x_2}\,
        \hat h_K^{\mathrm{spr}}(u_1,u_2)\,\D u_1\,\D u_2,
    \end{align}
    with
    \begin{align}
        \hat h_K^{\mathrm{spr}}(u_1,u_2)
        =
        K^{\,1-u_1-u_2}
        \frac{\Gamma(u_1+u_2-1)\Gamma(-u_2)}{\Gamma(u_1+1)}.
    \end{align}

    \item[ii)] \textbf{Exchange option.}
    For the Margrabe payoff
    \begin{align}
        h^{\mathrm{ex}}(x_1,x_2)
        &:=
        \bigl(\e^{x_1}-\e^{x_2}\bigr)^+,
    \end{align}
    one may either view it as the limit $K\downarrow 0$ of the spread payoff above, or use the one-dimensional representation
    \begin{align}
        h^{\mathrm{ex}}(x_1,x_2)
        =
        \frac{1}{2\pi i}
        \int_{R-i\infty}^{R+i\infty}
        \e^{ux_1+(1-u)x_2}\,
        \frac{1}{u(u-1)}\,\D u,
        \qquad R>1.
    \end{align}

    \item[iii)] \textbf{$M$-asset basket spread.}
    For $M\geq 2$, define
    \begin{align}
        h_K^{\mathrm{bs}}(\bm x)
        &:=
        \Bigl(\e^{x_1}-\sum_{m=2}^{M} \e^{x_m}-K\Bigr)^+,
        \qquad
        \bm x=(x_1,\dots,x_M)^\top.
    \end{align}
    Choose $\bm R=(R_1,\dots,R_M)^\top\in\mathbb{R}^M$ such that
    \begin{align}
        R_m<0,\qquad m=2,\dots,M,
        \qquad\text{and}\qquad
        \sum_{m=1}^{M} R_m>1.
    \end{align}
    Then
    \begin{align}
        h_K^{\mathrm{bs}}(\bm x)
        =
        \frac{1}{(2\pi i)^M}
        \int_{\mathcal S(\bm R)}
        \e^{\bm u^\top \bm x}\,
        \hat h_K^{\mathrm{bs}}(\bm u)\,\D\bm u,
    \end{align}
    with
    \begin{align}
        \hat h_K^{\mathrm{bs}}(\bm u)
        =
        K^{\,1-\sum_{m=1}^{M}u_m}
        \frac{
            \Gamma\bigl(\sum_{m=1}^{M}u_m-1\bigr)
            \prod_{m=2}^{M}\Gamma(-u_m)
        }{
            \Gamma(u_1+1)
        }.
    \end{align}

    \item[iv)] \textbf{Put on the sum.}
    For
    \begin{align}
        h_K^{\mathrm{sum}}(\bm x)
        &:=
        \Bigl(K-\sum_{m=1}^{M} \e^{x_m}\Bigr)^+,
    \end{align}
    choose $\bm R\in\mathbb{R}^M$ such that $R_m<0$ for all $m=1,\dots,M$. Then
    \begin{align}
        h_K^{\mathrm{sum}}(\bm x)
        =
        \frac{1}{(2\pi i)^M}
        \int_{\mathcal S(\bm R)}
        \e^{\bm u^\top \bm x}\,
        \hat h_K^{\mathrm{sum}}(\bm u)\,\D\bm u,
    \end{align}
    with
    \begin{align}
        \hat h_K^{\mathrm{sum}}(\bm u)
        =
        K^{\,1-\sum_{m=1}^{M}u_m}
        \frac{\prod_{m=1}^{M}\Gamma(-u_m)}{\Gamma\bigl(2-\sum_{m=1}^{M}u_m\bigr)}.
    \end{align}

    \item[v)] \textbf{Worst-of call.}
    Let
    \begin{align}
        h_K^{\wedge}(\bm x)
        &:=
        \Bigl(\min_{1\leq m\leq M} \e^{x_m}-K\Bigr)^+.
    \end{align}
    Then, on any strip $\mathcal S(\bm R)$ such that
    \begin{align}
        R_m>0,\qquad m=1,\dots,M,
        \qquad\text{and}\qquad
        \sum_{m=1}^{M}R_m>1,
    \end{align}
    one has
    \begin{align}
        h_K^{\wedge}(\bm x)
        =
        \frac{1}{(2\pi i)^M}
        \int_{\mathcal S(\bm R)}
        \e^{\bm u^\top \bm x}\,
        \hat h_K^{\wedge}(\bm u)\,\D\bm u,
    \end{align}
    with
    \begin{align}
        \hat h_K^{\wedge}(\bm u)
        =
        \frac{
            K^{\,1-\sum_{m=1}^{M}u_m}
        }{
            \bigl(\sum_{m=1}^{M}u_m-1\bigr)\prod_{m=1}^{M}u_m
        }.
    \end{align}
    In particular, for $M=2$ this is the Laplace kernel of the two-asset worst-of call.

    \item[vi)] \textbf{Two-asset best-of call.}
    In dimension $M=2$, define
    \begin{align}
        h_K^{\vee}(x_1,x_2)
        &:=
        \bigl(\e^{x_1}\vee \e^{x_2}-K\bigr)^+.
    \end{align}
    Using the maximum--minimum identity
    \begin{align}
        \e^{x_1}\vee \e^{x_2}
        +
        \e^{x_1}\wedge \e^{x_2}
        =
        \e^{x_1}+\e^{x_2},
    \end{align}
    one obtains the exact payoff decomposition
    \begin{align}
        h_K^{\vee}(x_1,x_2)
        =
        \bigl(\e^{x_1}-K\bigr)^+
        +
        \bigl(\e^{x_2}-K\bigr)^+
        -
        h_K^{\wedge}(x_1,x_2).
    \end{align}
    Hence the best-of call is implemented in the Fourier--Laplace framework by combining two one-dimensional vanilla call transforms with the two-dimensional worst-of-call transform:
    \begin{align}
        h_K^{\vee}(x_1,x_2)
        &=
        \frac{1}{2\pi i}
        \int_{R_1-i\infty}^{R_1+i\infty}
        \e^{u_1x_1}
        \frac{K^{\,1-u_1}}{u_1(u_1-1)}\,\D u_1
        +
        \frac{1}{2\pi i}
        \int_{R_2-i\infty}^{R_2+i\infty}
        \e^{u_2x_2}
        \frac{K^{\,1-u_2}}{u_2(u_2-1)}\,\D u_2 \nonumber\\
        &\quad
        -
        \frac{1}{(2\pi i)^2}
        \int_{\mathcal S(\bm R)}
        \e^{u_1x_1+u_2x_2}
        \frac{K^{\,1-u_1-u_2}}{(u_1+u_2-1)u_1u_2}\,\D u_1\,\D u_2,
    \end{align}
    where $R_1>1$, $R_2>1$, and $\bm R=(R_1,R_2)$ in the last integral satisfies $R_1>0$, $R_2>0$, and $R_1+R_2>1$. The corresponding put versions follow from put--call parity together with the same maximum--minimum identity. In particular,
    \begin{align}
        \bigl(K-\e^{x_1}\vee \e^{x_2}\bigr)^+
        &=
        \bigl(K-\e^{x_1}\bigr)^+
        +
        \bigl(K-\e^{x_2}\bigr)^+
        -
        \bigl(K-\e^{x_1}\wedge \e^{x_2}\bigr)^+,
    \end{align}
    so the best-of and worst-of puts may be reduced to the corresponding call representations plus the standard one-dimensional put kernels.
\end{enumerate}
\end{example}

\subsection{Implications for variance-optimal hedging}

Theorem~\ref{thm:FourierGKW} provides the analytical backbone for the computations that follow:
\begin{itemize}
    \item \textbf{Dynamic hedges.} The variance-optimal GKW integrand for a general payoff $h(\bm{Y}_T)$ is obtained by integrating
    the basis strategies $\bm{\theta}(\cdot,\bm{u})$ against $\zeta(\D\bm{u})$.

    \item \textbf{Static optimization inputs.} Since the residual $L$ is also represented by a Fourier integral,
    the covariance objects $A$, $B$, and $C$ in \eqref{eq:maincomponenetsstatic} reduce to integrals of the corresponding
    objects for $L(\bm{u})$, which become tractable once $\bm{\theta}(\bm{u})$ is available.
\end{itemize}
In affine and quadratic stochastic covariance models, the conditional transform $M_t(\bm{u})$ admits
closed-form or low-dimensional Riccati or Volterra representations, which we exploit in subsequent sections
to compute $\bm{\theta}(\bm{u})$ and the semi-static quantities efficiently.

\subsection{Fourier Representation of the Static Hedging Components}
\label{subsec:fourier-static}

The semi-static VO solution of Section~\ref{sec:VO-framework} requires, besides the dynamic hedge ratios,
the quadratic inputs $A\in\mathbb{R}$, $B\in\mathbb{R}^n$, and $C\in\mathbb{S}_+^n$ in
\eqref{eq:maincomponenetsstatic}. These objects are second-moment quantities of the orthogonal
residuals in the GKW decompositions of the target and auxiliary claims. In this subsection we express
$A$, $B$, and $C$ as Fourier--Laplace integrals of the corresponding quantities for exponential basis claims. This
reduces the computation of the static optimization problem to evaluating predictable covariations of
fundamental exponential martingales. Let $\eta^0=h^0(\bm{Y}_T)$ be the target claim and let

\begin{equation*}
    \eta^j=h^j(\bm{Y}_T),\qquad j=1,\dots,n,
\end{equation*}
be the auxiliary claims. Assume that each payoff admits a Laplace representation of the form \eqref{eq:LaplaceMVpayoffs} on a
strip

\begin{equation*}
    \mathcal{S}(\bm{R}^{\,j})
    :=
    \bigl\{\bm{u}\in\mathbb{C}^d:\Re(\bm{u})=\bm{R}^{\,j}\bigr\},
\end{equation*}
with associated complex measure $\zeta^j$:

\begin{equation*}
    h^j(\bm{y})
    =
    \int_{\mathcal{S}(\bm{R}^{\,j})}
    \e^{\bm{u}^\top \bm{y}}\,\zeta^j(\D\bm{u}),
    \qquad j=0,1,\dots,n.
\end{equation*}
If $\eta^j$ depends only on a subset of the components of $\bm{Y}_T$, then the integral is taken over the corresponding lower-dimensional strip; cf.~Example~\ref{ex:ProductOptions}.

For each $j$, define the value process

\begin{equation*}
    H_t^j:=\mathbb{E}[\eta^j\mid\mathcal{F}_t]
\end{equation*}
and write its GKW decomposition with respect to $\bm{S}$ as

\begin{equation*}
    H_t^j
    =
    H_0^j+\int_0^t (\bm{\theta}_s^j)^\top\,\D\bm{S}_s+L_t^j,
    \qquad j=0,1,\dots,n,
\end{equation*}
where $L^j\perp \bm{S}$. By Theorem~\ref{thm:FourierGKW}, the GKW integrands and residuals admit Fourier representations obtained by integrating the corresponding objects for the exponential basis martingales $H(\bm{u})$.

\subsubsection*{Residual brackets via Fourier integration}

The entries of $B$ and $C$ can be written in terms of predictable covariations of the residuals:
for square-integrable martingales, $\mathbb{E}[L_T^iL_T^j]=\mathbb{E}[\langle L^i,L^j\rangle_T]$.
Thus we need to evaluate expectations of $\langle L^i,L^j\rangle_T$. Let

\begin{equation*}
    H(\bm{u})_t
    :=
    \mathbb{E}[\e^{\bm{u}^\top \bm{Y}_T}\mid\mathcal{F}_t]
\end{equation*}
and denote by $(\bm{\theta}(\cdot,\bm{u}),L(\cdot,\bm{u}))$
the GKW decomposition of $H(\bm{u})$ with respect to $\bm{S}$. Under the integrability conditions of
Theorem~\ref{thm:FourierGKW}, stochastic Fubini yields the bracket representation
(cf.~\cite[Theorem~4.2]{Semi_static_Fourier})
\begin{equation}\label{eq:FourierBracketResiduals}
    \langle L^i, L^j\rangle_T
    =
    \int_{\mathcal{S}(\bm{R}^{\,i})}
    \int_{\mathcal{S}(\bm{R}^{\,j})}
    \langle L(\cdot,\bm{u}_i),L(\cdot,\bm{u}_j)\rangle_T\,
    \zeta^j(\D\bm{u}_j)\,\zeta^i(\D\bm{u}_i),
    \qquad i,j=0,\dots,n.
\end{equation}

As a direct consequence, the static optimization inputs are obtained by integrating the corresponding
basis brackets.

\begin{corollary}[Fourier formulas for the static optimization inputs]
\label{cor:ABC-fourier}
Let $A$, $B$, and $C$ be defined in \eqref{eq:maincomponenetsstatic}. Under the assumptions of
Theorem~\ref{thm:FourierGKW} and \eqref{eq:FourierBracketResiduals}, we have
\begin{align}
\label{eq:maincomponenetsstaticfourier}
    A
    &=
    \mathbb{E}_{\mathbb{Q}}\big[\langle L^0,L^0\rangle_T\big],\\
    B_j
    &=
    \int_{\mathcal{S}(\bm{R}^{\,j})}
    \mathbb{E}_{\mathbb{Q}}\big[\langle L^0,L(\cdot,\bm{u})\rangle_T\big]\,
    \zeta^j(\D\bm{u}),
    \qquad j=1,\dots,n,\\
    C_{ij}
    &=
    \int_{\mathcal{S}(\bm{R}^{\,i})}
    \int_{\mathcal{S}(\bm{R}^{\,j})}
    \mathbb{E}_{\mathbb{Q}}\big[\langle L(\cdot,\bm{u}_i),L(\cdot,\bm{u}_j)\rangle_T\big]\,
    \zeta^i(\D\bm{u}_i)\,\zeta^j(\D\bm{u}_j),
    \qquad i,j=1,\dots,n.
\end{align}
\end{corollary}

\subsubsection*{Explicit dynamics of pairwise residual brackets}

To implement Corollary~\ref{cor:ABC-fourier}, it remains to compute $\langle L(\cdot,\bm{u}_1),L(\cdot,\bm{u}_2)\rangle$
for $\bm{u}_1,\bm{u}_2$ in the relevant strips. For each $\bm{u}$, the basis martingale admits the GKW decomposition

\begin{equation*}
    H(\bm{u})_t
    =
    H(\bm{u})_0+\int_0^t \bm{\theta}_s(\bm{u})^\top\,\D\bm{S}_s + L(\bm{u})_t.
\end{equation*}
By bilinearity of predictable covariation,

\begin{align*}
    \langle L(\bm{u}_1),L(\bm{u}_2)\rangle
    &=
    \langle H(\bm{u}_1),H(\bm{u}_2)\rangle
    -\Big\langle H(\bm{u}_1),\int \bm{\theta}(\bm{u}_2)^\top\,\D\bm{S}\Big\rangle
    -\Big\langle \int \bm{\theta}(\bm{u}_1)^\top\,\D\bm{S}, H(\bm{u}_2)\Big\rangle
    \\ 
    &\quad +\Big\langle \int \bm{\theta}(\bm{u}_1)^\top\,\D\bm{S}, \int \bm{\theta}(\bm{u}_2)^\top\,\D\bm{S}\Big\rangle.
\end{align*}
Differentiating and using the defining projection identity of the GKW integrand,

\begin{equation*}
    \D\langle \bm{S},H(\bm{u})\rangle_t
    =
    \D\llangle \bm{S},\bm{S}\rrangle_t\,\bm{\theta}_t(\bm{u}),
\end{equation*}
we obtain the compact form
\begin{equation}\label{eq:bracket-residual-diff}
    \D\langle L(\bm{u}_1),L(\bm{u}_2)\rangle_t
    =
    \D\langle H(\bm{u}_1),H(\bm{u}_2)\rangle_t
    -
    \bm{\theta}_t(\bm{u}_1)^\top\,\D\llangle \bm{S},\bm{S}\rrangle_t\,\bm{\theta}_t(\bm{u}_2),
    \qquad t\in[0,T].
\end{equation}
This identity holds for continuous square-integrable martingales and is the multivariate analogue of the
univariate formula used in \cite{Semi_static_Fourier}. It is convenient to package the relevant residual brackets into finite-variation processes.
For $\bm{u},\bm{u}_1,\bm{u}_2\in\mathbb{C}^d$, define predictable processes $\mathcal{A}$, $\mathcal{B}(\bm{u})$, and
$\mathcal{C}(\bm{u}_1,\bm{u}_2)$ by
\begin{align}
\label{eq:defABCprocesses}
    \D\mathcal{A}_t
    &:=
    \D\langle H^0,H^0\rangle_t-(\bm{\theta}_t^0)^\top\,\D\llangle \bm{S},\bm{S}\rrangle_t\,\bm{\theta}_t^0,\\
    \D\mathcal{B}(\bm{u})_t
    &:=
    \D\langle H^0,H(\bm{u})\rangle_t-(\bm{\theta}_t^0)^\top\,\D\llangle \bm{S},\bm{S}\rrangle_t\,\bm{\theta}_t(\bm{u}),\\
    \D\mathcal{C}(\bm{u}_1,\bm{u}_2)_t
    &:=
    \D\langle H(\bm{u}_1),H(\bm{u}_2)\rangle_t
    -\bm{\theta}_t(\bm{u}_1)^\top\,\D\llangle \bm{S},\bm{S}\rrangle_t\,\bm{\theta}_t(\bm{u}_2),
\end{align}
with

\begin{equation*}
    \mathcal{A}_0=\mathcal{B}(\bm{u})_0=\mathcal{C}(\bm{u}_1,\bm{u}_2)_0=0.
\end{equation*}
By \eqref{eq:bracket-residual-diff},

\begin{equation*}
    \mathcal{A}_T=\langle L^0,L^0\rangle_T,\qquad
    \mathcal{B}(\bm{u})_T=\langle L^0,L(\bm{u})\rangle_T,\qquad
    \mathcal{C}(\bm{u}_1,\bm{u}_2)_T=\langle L(\bm{u}_1),L(\bm{u}_2)\rangle_T.
\end{equation*}
Therefore, Corollary~\ref{cor:ABC-fourier} can be written equivalently as
\begin{align}
\label{eq:importantFourier}
    A
    &=
    \mathbb{E}_{\mathbb{Q}}\big[\mathcal{A}_T\big],\\
    B_j
    &=
    \int_{\mathcal{S}(\bm{R}^{\,j})}
    \mathbb{E}_{\mathbb{Q}}\big[\mathcal{B}(\bm{u})_T\big]\,
    \zeta^j(\D\bm{u}),
    \qquad j=1,\dots,n,\\
    C_{ij}
    &=
    \int_{\mathcal{S}(\bm{R}^{\,i})}
    \int_{\mathcal{S}(\bm{R}^{\,j})}
    \mathbb{E}_{\mathbb{Q}}\big[\mathcal{C}(\bm{u}_i,\bm{u}_j)_T\big]\,
    \zeta^i(\D\bm{u}_i)\,\zeta^j(\D\bm{u}_j),
    \qquad i,j=1,\dots,n.
\end{align}

Equations~\eqref{eq:bracket-residual-diff}--\eqref{eq:importantFourier} reduce the computation of the
static optimization inputs to evaluating predictable covariations of the exponential basis martingales
$H(\bm{u})$. In the next section we show that, in both affine and quadratic stochastic covariance models, the
required covariations can be obtained from low-dimensional ODEs.

\section{Affine Stochastic Covariance Models}
\label{sec:affinecase}

The purpose of this section is to introduce a class of multivariate stochastic covariance models for which the Fourier--Laplace representations derived in Section~\ref{sec:fourierrepresentation} become analytically tractable. The central class is that of \emph{affine stochastic covariance models}, where the relevant conditional transforms depend exponentially-affinely on the state variables. This property reduces the computation of pricing, Galtchouk--Kunita--Watanabe hedge ratios, and the static quadratic quantities $A$, $\bm B$, and $\bm C$ to systems of matrix Riccati equations and, in several important cases, to explicit matrix-exponential formulas.

The mathematical foundation for affine processes on the cone of symmetric positive semidefinite matrices was developed in \cite{Cuchiero2011}, extending the general affine framework of Duffie, Filipovi\'c, and Schachermayer to matrix-valued state spaces. Multivariate stochastic covariance models based on this theory, including Wishart-type specifications and affine jump extensions, have subsequently been studied in \cite{Fonseca2007,MAYERHOFER2011568,GNOATTO2012,Ahdida2013,Muhle-Karbe2012}. In the present context, their main relevance is that they provide closed-form or low-dimensional representations for the exponential basis martingales $H(\bm u)$, which in turn allows one to compute the semi-static hedging objects appearing in Section~\ref{sec:fourierrepresentation}.

\subsection{Affine processes on $\mathbb{S}_+^d$}

Let $M_d(\mathbb{R})$ denote the space of $d\times d$ real matrices and let $\bm I_d$ be the identity matrix. We write $\mathbb{S}^d$ for the space of symmetric $d\times d$ matrices, equipped with the Hilbert--Schmidt inner product

\begin{equation*}
    \langle \bm x,\bm y\rangle := \Tr(\bm x \bm y),
    \qquad
    \bm x,\bm y\in\mathbb{S}^d,
\end{equation*}
and the induced Frobenius norm $\|\bm x\|:=\sqrt{\Tr(\bm x^2)}$. The cone of symmetric positive semidefinite matrices is denoted by $\mathbb{S}_+^d$, and its interior by $\mathbb{S}_{++}^d$. We use the Loewner order: $\bm x\preceq \bm y$ means $\bm y-\bm x\in\mathbb{S}_+^d$, while $\bm x\prec \bm y$ means $\bm y-\bm x\in\mathbb{S}_{++}^d$. For $i,j\in\{1,\dots,d\}$, we denote by $\bm \e^{ij}$ the matrix unit with entry $1$ at $(i,j)$ and $0$ elsewhere.

Let $(\Omega,\mathcal F,\mathbb F,\mathbb Q)$ be a filtered probability space satisfying the usual conditions. We consider a time-homogeneous Markov process $\bm\Sigma=(\bm\Sigma_t)_{t\in[0,T]}$ with values in $\mathbb{S}_+^d$.

\begin{definition}[Affine process on $\mathbb{S}_+^d$]
\label{def:affine_process}
The process $\bm\Sigma$ is called an \emph{affine process} on $\mathbb{S}_+^d$ if it is stochastically continuous and if there exist functions

\begin{equation*}
    \phi:[0,T]\times \mathbb{S}_+^d\to\mathbb{R},
    \qquad
    \bm\psi:[0,T]\times \mathbb{S}_+^d\to\mathbb{S}_+^d,
\end{equation*}
such that, for every $t\in[0,T]$ and all $\bm u,\bm x\in\mathbb{S}_+^d$,
\begin{equation}
\label{eq:exponentialaffineform}
    \mathbb E_{\bm x}\left[\exp\big(-\Tr(\bm u \bm\Sigma_t)\big)\right]
    =
    \exp\big(-\phi(t,\bm u)-\Tr(\bm\psi(t,\bm u)\bm x)\big).
\end{equation}
\end{definition}

In the conservative case relevant for stochastic covariance modeling in finance, the infinitesimal generator of $\bm\Sigma$ is determined by an admissible parameter set $(\bm\alpha,\bm b,\mathcal B,m,\mu)$, where $\bm\alpha,\bm b\in\mathbb{S}_+^d$, $\mathcal B:\mathbb{S}^d\to\mathbb{S}^d$ is linear, $m$ is a state-independent jump measure on $\mathbb{S}_+^d\setminus\{\bm 0\}$, and $\mu$ is a state-dependent jump kernel on $\mathbb{S}_+^d\setminus\{\bm 0\}$; see \cite[Section~2]{Cuchiero2011}. We restrict throughout to the no-killing case.

\begin{theorem}[Affine transform formula on $\mathbb{S}_+^d$, cf. \cite{Cuchiero2011}]
\label{thm:AffineRiccati}
Let $\bm\Sigma$ be a conservative affine process on $\mathbb{S}_+^d$ with admissible parameters $(\bm\alpha,\bm b,\mathcal B,m,\mu)$. Then the functions $\phi$ and $\bm\psi$ in \eqref{eq:exponentialaffineform} are the unique solutions of the generalized Riccati system
\begin{align}
\label{eq:RiccatiODE_affinecone}
    \partial_t \phi(t,\bm u)
    &=
    \mathrm F(\bm\psi(t,\bm u)),
    & \phi(0,\bm u)&=0, \\
    \partial_t \bm\psi(t,\bm u)
    &=
    \mathrm R(\bm\psi(t,\bm u)),
    & \bm\psi(0,\bm u)&=\bm u,
\end{align}
where, in the no-killing case,
\begin{align}
\label{eq:F_affinecone}
    \mathrm F(\bm v)
    &=
    \Tr(\bm b\,\bm v)
    -
    \int_{\mathbb{S}_+^d\setminus\{\bm 0\}}
    \big(\e^{-\Tr(\bm v\bm\xi)}-1\big)\,m(\D\bm\xi), \\
\label{eq:R_affinecone}
    \mathrm R(\bm v)
    &=
    -2\bm v \bm\alpha \bm v
    +\mathcal B^\ast(\bm v)
    -
    \int_{\mathbb{S}_+^d\setminus\{\bm 0\}}
    \frac{\e^{-\Tr(\bm v\bm\xi)}-1+\Tr(\chi(\bm\xi)\bm v)}{1\wedge \|\bm\xi\|^2}\,\mu(\D\bm\xi),
\end{align}
and $\mathcal B^\ast$ denotes the adjoint of $\mathcal B$ with respect to the Hilbert--Schmidt inner product. Conversely, every admissible parameter set defines a unique conservative affine process on $\mathbb{S}_+^d$.
\end{theorem}

In particular, if $m=\mu=0$ and $\mathcal B(\bm x)=\bm M\bm x+\bm x\bm M^\top$ for some $\bm M\in M_d(\mathbb R)$, then $\bm\Sigma$ is a continuous affine diffusion of Wishart type. Writing $\bm\alpha=\bm A^\top\bm A$ with $\bm A\in M_d(\mathbb R)$, its dynamics take the form
\begin{equation}
\label{eq:wishartSigma}
    \D\bm\Sigma_t
    =
    \big(\bm\Omega+\bm M\bm\Sigma_t+\bm\Sigma_t\bm M^\top\big)\,\D t
    +
    \sqrt{\bm\Sigma_t}\,\D\bm W_t\,\bm A
    +
    \bm A^\top \D\bm W_t^\top \sqrt{\bm\Sigma_t},
\end{equation}
where $\bm\Omega=\bm b\in\mathbb{S}_+^d$ and $\bm W$ is a $d\times d$ matrix Brownian motion.

\subsection{Affine stochastic covariance models}

We now embed the affine covariance process $\bm\Sigma$ into a multivariate asset-pricing model. Let $\bm S=(S^1,\dots,S^d)^\top$ denote the vector of discounted asset prices and let $\bm Y=\log \bm S$ be the log-price process. We say that $(\bm Y,\bm\Sigma)$ is an \emph{Affine Stochastic Covariance} model if the joint process is affine on the state space

\begin{equation*}
    D:=\mathbb{R}^d\times\mathbb{S}_+^d.
\end{equation*}

\begin{assumption}[ASC transform structure]
\label{ass:affineSV}
There exist functions

\begin{equation*}
    \phi:[0,T]\times\mathbb{C}^d\times(\mathbb{S}^d+i\mathbb{S}^d)\to\mathbb{C},
    \qquad
    \bm\Psi:[0,T]\times\mathbb{C}^d\times(\mathbb{S}^d+i\mathbb{S}^d)\to \mathbb{S}^d+i\mathbb{S}^d,
\end{equation*}
such that, for every $0\le t\le T$, every $\bm u\in\mathbb{C}^d$, every $\bm V\in \mathbb{S}^d+i\mathbb{S}^d$ in the domain of the transform, and $\tau:=T-t$,
\begin{equation}
\label{exponentialaffine}
    \mathbb E\left[
        \exp\big(
            \bm u^\top \bm Y_T + \Tr(\bm V\bm\Sigma_T)
        \big)
        \,\middle|\,
        \mathcal F_t
    \right]
    =
    \exp\big(
        \phi(\tau,\bm u,\bm V)
        + \bm u^\top \bm Y_t
        + \Tr(\bm\Psi(\tau,\bm u,\bm V)\bm\Sigma_t)
    \big).
\end{equation}
\end{assumption}

The exponential basis martingales introduced in Section~\ref{sec:fourierrepresentation} are therefore given by
\begin{equation}
\label{affinerepresentation}
    H_t(\bm u)
    :=
    \mathbb E\left[\e^{\bm u^\top \bm Y_T}\mid\mathcal F_t\right]
    =
    \exp\big(
        \phi(T-t,\bm u,\bm 0)
        + \bm u^\top \bm Y_t
        + \Tr(\bm\Psi(T-t,\bm u,\bm 0)\bm\Sigma_t)
    \big).
\end{equation}
Hence, once $\phi$ and $\bm\Psi$ are known, the Fourier objects of Section~\ref{sec:fourierrepresentation} become explicit.

Two important examples are the continuous Wishart model and the jump-driven multivariate BNS model.

\begin{example}[Wishart affine stochastic covariance model]
\label{WASCExample}
The continuous Wishart affine stochastic covariance model is defined by
\begin{align}
\label{eq:WASC_Y}
    \D\bm Y_t
    &=
    -\frac12 \diag(\bm\Sigma_t)\,\D t
    +
    \sqrt{\bm\Sigma_t}\,\D\bm B_t, \\
\label{eq:WASC_Sigma}
    \D\bm\Sigma_t
    &=
    \big(\bm\Omega+\bm M\bm\Sigma_t+\bm\Sigma_t\bm M^\top\big)\,\D t
    +
    \sqrt{\bm\Sigma_t}\,\D\bm W_t\,\bm A
    +
    \bm A^\top \D\bm W_t^\top \sqrt{\bm\Sigma_t},
\end{align}
where $\bm\Omega\in\mathbb{S}_+^d$, $\bm M,\bm A\in M_d(\mathbb R)$, and the Brownian drivers satisfy the vector-correlation specification
\begin{equation}
\label{eq:WASCcorr}
    \D\bm B_t
    =
    \D\bm W_t\,\bm\rho
    +
    \sqrt{1-\bm\rho^\top\bm\rho}\,\D\bm Z_t,
    \qquad
    \bm\rho\in\mathbb R^d,
    \qquad
    \bm\rho^\top\bm\rho\le 1,
\end{equation}
with $\bm Z$ a $d$-dimensional Brownian motion independent of $\bm W$. In order to ensure that $\bm\Sigma_t\in\mathbb{S}_{++}^d$ for all $t$, one imposes the standard Wishart admissibility condition $\bm\Omega\succeq (d-1)\bm A^\top\bm A$; see \cite{Fonseca2007}.

For this model, the transform coefficients in \eqref{exponentialaffine} solve the Riccati system
\begin{align}
\label{eq:WASC_Riccati_Psi}
    \partial_\tau \bm\Psi(\tau,\bm u,\bm V)
    &=
    \bm\Psi(\tau,\bm u,\bm V)\,\bm A^\top\bm A\,\bm\Psi(\tau,\bm u,\bm V)
    +\bm\Psi(\tau,\bm u,\bm V)\big(\bm M+\bm A^\top\bm\rho\,\bm u^\top\big) \nonumber\\
    &\quad
    +\big(\bm M^\top+\bm u\,\bm\rho^\top\bm A\big)\bm\Psi(\tau,\bm u,\bm V)
    +\frac12\big(\bm u\bm u^\top-\diag(\bm u)\big), \\
\label{eq:WASC_Riccati_phi}
    \partial_\tau \phi(\tau,\bm u,\bm V)
    &=
    \Tr\big(\bm\Omega\,\bm\Psi(\tau,\bm u,\bm V)\big),
\end{align}
with initial conditions $\bm\Psi(0,\bm u,\bm V)=\bm V$ and $\phi(0,\bm u,\bm V)=0$.

A key advantage of the Wishart model is that \eqref{eq:WASC_Riccati_Psi} can be linearized. Define the $2d\times 2d$ Hamiltonian matrix
\begin{equation}
\label{eq:WASC_Hamiltonian}
    \bm{\mathcal H}(\bm u)
    :=
    \begin{pmatrix}
        \bm M+\bm A^\top\bm\rho\,\bm u^\top
        &
        -2\bm A^\top\bm A
        \\[1mm]
        \frac12\big(\bm u\bm u^\top-\diag(\bm u)\big)
        &
        -\big(\bm M^\top+\bm u\,\bm\rho^\top\bm A\big)
    \end{pmatrix}.
\end{equation}
If

\begin{align*}
    \exp\big(\tau \bm{\mathcal H}(\bm u)\big)
    =
    \begin{pmatrix}
        \bm\Theta_{11}(\tau,\bm u) & \bm\Theta_{12}(\tau,\bm u) \\
        \bm\Theta_{21}(\tau,\bm u) & \bm\Theta_{22}(\tau,\bm u)
    \end{pmatrix},
\end{align*}
then
\begin{equation}
\label{eq:WASC_Psi_closedform}
    \bm\Psi(\tau,\bm u,\bm V)
    =
    \big(
        \bm\Theta_{21}(\tau,\bm u)+\bm\Theta_{22}(\tau,\bm u)\bm V
    \big)
    \big(
        \bm\Theta_{11}(\tau,\bm u)+\bm\Theta_{12}(\tau,\bm u)\bm V
    \big)^{-1},
\end{equation}
whenever the inverse exists. In particular, for the basis martingales $H(\bm u)$, one has

\begin{equation*}
    \bm\Psi(\tau,\bm u,\bm 0)
    =
    \bm\Theta_{21}(\tau,\bm u)\bm\Theta_{11}(\tau,\bm u)^{-1}.
\end{equation*}
\end{example}

\begin{example}[Multivariate BNS model]
\label{BNSExample}
An affine jump specification is provided by the multivariate Barndorff--Nielsen and Shephard model. Let $\bm L$ be a matrix subordinator on $\mathbb{S}_+^d$, let $\mathcal R:\mathbb{S}^d\to\mathbb R^d$ be a linear leverage map, and let $\bm B$ be a $d$-dimensional Brownian motion independent of $\bm L$. The model is defined by
\begin{align}
\label{eq:BNS_Y}
    \D\bm Y_t
    &=
    -\frac12 \diag(\bm\Sigma_t)\,\D t
    -\bm\kappa\,\D t
    +
    \sqrt{\bm\Sigma_t}\,\D\bm B_t
    +
    \mathcal R(\D\bm L_t), \\
\label{eq:BNS_Sigma}
    \D\bm\Sigma_t
    &=
    \big(\bm M\bm\Sigma_t+\bm\Sigma_t\bm M^\top\big)\,\D t
    +
    \D\bm L_t,
\end{align}
where $\bm\kappa\in\mathbb R^d$ is chosen so that $\bm S=\e^{\bm Y}$ is a local martingale.
For the variance-optimal hedging problem we strengthen this requirement and assume throughout the BNS applications that $\bm S$ is in fact a square-integrable true $\mathbb Q$-martingale on $[0,T]$, i.e.
\[
    \mathbb E_{\mathbb Q}\!\left[\sup_{0\le t\le T}\|\bm S_t\|^2\right]<\infty,
    \qquad
    \mathbb E_{\mathbb Q}\!\left[\|\bm S_T\|^2\right]<\infty,
\]
and that the Laplace strips used later lie in the domain where the affine transform \eqref{eq:BNS_phi} is finite. Without these $L^2$-martingale assumptions the GKW decomposition and the variance-optimal objective are not well posed.

In the common compound-Poisson Wishart specification, $\bm L$ has Lévy measure

\begin{equation*}
    m_{\bm L}(\D\bm X)
    =
    \lambda\,F_{W_d(n,\bm\Theta)}(\D\bm X),
\end{equation*}
where $F_{W_d(n,\bm\Theta)}$ denotes the Wishart law with shape parameter $n$ and scale matrix $\bm\Theta\in\mathbb S_{++}^d$. Then, for every matrix $\bm R$ such that $\bm I_d-2\bm R\bm\Theta\in\mathbb S_{++}^d$,
\begin{equation}
\label{eq:WishartMGF}
    \int_{\mathbb S_+^d} \e^{\Tr(\bm R\bm X)}\,m_{\bm L}(\D\bm X)
    =
    \lambda\,
    \det(\bm I_d-2\bm R\bm\Theta)^{-n/2}.
\end{equation}

In this model, the matrix Riccati equation is linear:
\begin{equation}
\label{eq:BNS_Psi}
    \partial_\tau \bm\Psi(\tau,\bm u,\bm V)
    =
    \bm\Psi(\tau,\bm u,\bm V)\bm M
    +
    \bm M^\top \bm\Psi(\tau,\bm u,\bm V)
    +
    \frac12\big(\bm u\bm u^\top-\diag(\bm u)\big),
    \qquad
    \bm\Psi(0,\bm u,\bm V)=\bm V.
\end{equation}
The scalar function $\phi$ is then obtained from the Lévy exponent of $\bm L$:
\begin{equation}
\label{eq:BNS_phi}
    \partial_\tau \phi(\tau,\bm u,\bm V)
    =
    \int_{\mathbb S_+^d}
    \left(
        \e^{\bm u^\top \mathcal R(\bm X) + \Tr(\bm\Psi(\tau,\bm u,\bm V)\bm X)}
        -1
        -\bm u^\top \mathcal R(\bm X)
    \right)
    m_{\bm L}(\D\bm X)
    -
    \bm u^\top \bm\kappa.
\end{equation}
Hence the transform remains fully explicit once \eqref{eq:WishartMGF} is available.
\end{example}

\subsection{Semi-static hedging in affine stochastic covariance models}

We now specialize the semi-static variance-optimal hedging formulas of Section~\ref{sec:fourierrepresentation} to affine stochastic covariance models. Let

\begin{equation*}
    \bm X_t
    :=
    \big(\bm Y_t,\operatorname{vec}(\bm\Sigma_t)\big)
    \in
    \mathbb R^d\times\mathbb R^{d^2},
\end{equation*}
and let $\tau:=T-t$. For $\bm u\in\mathcal S(\bm R)$, define the basis martingale

\begin{equation*}
    H_t(\bm u)
    =
    \exp\big(
        \phi(\tau,\bm u,\bm 0)
        + \bm u^\top \bm Y_t
        + \Tr(\bm\Psi(\tau,\bm u,\bm 0)\bm\Sigma_t)
    \big).
\end{equation*}
In order to compute the GKW integrands $\bm\theta(\bm u)$ and the residual brackets entering $A$, $\bm B$, and $\bm C$, it is enough to identify the local covariation kernels of $\bm S$, $H(\bm u)$, and $\bm\Sigma$.

Let $\bm C_t$ denote the predictable covariance matrix of the continuous martingale part of $\bm X$, written in block form as
\begin{equation}
\label{eq:Ct_block}
    \bm C_t
    =
    \begin{pmatrix}
        \bm C_t^{\bm Y\bm Y} & \bm C_t^{\bm Y\Sigma} \\
        (\bm C_t^{\bm Y\Sigma})^\top & \bm C_t^{\Sigma\Sigma}
    \end{pmatrix},
\end{equation}
where the blocks are defined by
\begin{align}
\label{eq:QV_Y}
    \D\llangle \bm Y,\bm Y\rrangle_t
    &=
    \bm C_t^{\bm Y\bm Y}\,\D t, \\
\label{eq:QV_Cross}
    \D\langle Y_t^r,(\operatorname{vec}\bm\Sigma_t)_{(ij)}\rangle
    &=
    \big(\bm C_t^{\bm Y\Sigma}\big)_{r,(ij)}\,\D t,
    \qquad r,i,j\in\{1,\dots,d\}, \\
\label{eq:QV_Sigma}
    \D\llangle \operatorname{vec}(\bm\Sigma),\operatorname{vec}(\bm\Sigma)\rrangle_t
    &=
    \bm C_t^{\Sigma\Sigma}\,\D t.
\end{align}
If jumps are present, let $\nu_t(\D\bm\xi_Y,\D\bm\xi_\Sigma)\,\D t$ denote the predictable compensator of the jump measure of $(\bm Y,\bm\Sigma)$.

For $\bm u,\bm u_1,\bm u_2\in\mathcal S(\bm R)$, define

\begin{equation*}
    \mathcal E_t(\bm u;\bm\xi_Y,\bm\xi_\Sigma)
    :=
    \bm u^\top \bm\xi_Y
    +
    \Tr\big(\bm\Psi(\tau,\bm u,\bm 0)\bm\xi_\Sigma\big).
\end{equation*}
Then It\^o's formula yields the following local covariation kernels:
\begin{align}
\label{eq:K_HH_general}
    \mathcal K_t^{HH}(\bm u_1,\bm u_2)
    &:=
    \frac{\D}{\D t}\langle H(\bm u_1),H(\bm u_2)\rangle_t \nonumber\\
    &=
    H_t(\bm u_1)H_t(\bm u_2)
    \Bigg[
        \Gamma_t^c(\bm u_1,\bm u_2)
        +
        \int_D
        \big(\e^{\mathcal E_t(\bm u_1;\bm\xi_Y,\bm\xi_\Sigma)}-1\big)
        \big(\e^{\mathcal E_t(\bm u_2;\bm\xi_Y,\bm\xi_\Sigma)}-1\big)
        \nu_t(\D\bm\xi_Y,\D\bm\xi_\Sigma)
    \Bigg], \\
\label{eq:Gamma_c_general}
    \Gamma_t^c(\bm u_1,\bm u_2)
    &=
    \begin{pmatrix}
        \bm u_1 \\
        \operatorname{vec}\big(\bm\Psi(\tau,\bm u_1,\bm 0)\big)
    \end{pmatrix}^{\top}
    \bm C_t
    \begin{pmatrix}
        \bm u_2 \\
        \operatorname{vec}\big(\bm\Psi(\tau,\bm u_2,\bm 0)\big)
    \end{pmatrix}, \\
\label{eq:K_SH_general}
    \bm{\mathcal K}_t^{\bm S H}(\bm u)
    &:=
    \frac{\D}{\D t}\langle \bm S,H(\bm u)\rangle_t \nonumber\\
    &=
    \diag(\bm S_{t-})\,H_{t-}(\bm u)
    \Bigg[
        \bm C_t^{\bm Y\bm Y}\bm u
        +
        \bm C_t^{\bm Y\Sigma}\operatorname{vec}\big(\bm\Psi(\tau,\bm u,\bm 0)\big)
        \\
        &\quad +
        \int_D
        \big(\e^{\bm\xi_Y}-\bm 1\big)
        \big(\e^{\mathcal E_t(\bm u;\bm\xi_Y,\bm\xi_\Sigma)}-1\big)
        \nu_t(\D\bm\xi_Y,\D\bm\xi_\Sigma)
    \Bigg], \\
\label{eq:K_SS_general}
    \bm{\mathcal K}_t^{\bm S\bm S}
    &:=
    \frac{\D}{\D t}\llangle \bm S,\bm S\rrangle_t \nonumber\\
    &=
    \diag(\bm S_{t-})\,\bm C_t^{\bm Y\bm Y}\,\diag(\bm S_{t-})
    +
    \diag(\bm S_{t-})
    \left(
        \int_D\!
        \big(\!\e^{\bm\xi_Y}-\!\bm 1\big)\big(\e^{\bm\xi_Y}-\bm 1\!\big)^\top
        \nu_t(\D\bm\xi_Y,\D\bm\xi_\Sigma)
    \right)\!
    \diag(\bm S_{t-}).
\end{align}
Consequently, the GKW hedge ratio of the basis martingale $H(\bm u)$ is
\begin{equation}
\label{eq:theta_u_explicit}
    \bm\theta_t(\bm u)
    =
    \big(\bm{\mathcal K}_t^{\bm S\bm S}\big)^{\dagger}\,
    \bm{\mathcal K}_t^{\bm S H}(\bm u).
\end{equation}
The corresponding residual-bracket density is therefore
\begin{equation}
\label{eq:C_explicit}
    \mathcal C_t(\bm u_1,\bm u_2)
    :=
    \frac{\D}{\D t}\langle L(\bm u_1),L(\bm u_2)\rangle_t
    =
    \mathcal K_t^{HH}(\bm u_1,\bm u_2)
    -
    \bm\theta_t(\bm u_1)^\top
    \bm{\mathcal K}_t^{\bm S\bm S}
    \bm\theta_t(\bm u_2).
\end{equation}
Integrating $\mathcal C_t(\bm u_1,\bm u_2)$ over time and over the Fourier contours yields exactly the quantities appearing in \eqref{eq:importantFourier}.

We now record the resulting formulas in the two benchmark affine models introduced above.

\begin{example}[Wishart affine stochastic covariance: explicit hedging kernels]
\label{ex:WASC_Hedging}
Consider the continuous Wishart model \eqref{eq:WASC_Y}--\eqref{eq:WASCcorr}. Since the model is continuous, the jump kernel vanishes identically. The continuous covariance blocks in \eqref{eq:Ct_block} are then given by
\begin{align}
\label{eq:WASC_CYY}
    \bm C_t^{\bm Y\bm Y}
    &=
    \bm\Sigma_t, \\
\label{eq:WASC_CYSigma}
    \big(\bm C_t^{\bm Y\Sigma}\big)_{r,(ij)}
    &=
    (\bm\Sigma_t)_{ri}(\bm A^\top\bm\rho)_j
    +
    (\bm\Sigma_t)_{rj}(\bm A^\top\bm\rho)_i,
    \qquad r,i,j\in\{1,\dots,d\}, \\
\label{eq:WASC_CSigmaSigma}
    \big(\bm C_t^{\Sigma\Sigma}\big)_{(ij),(kl)}
    &=
    (\bm\Sigma_t)_{ik}\bm Q_{jl}
    +
    (\bm\Sigma_t)_{il}\bm Q_{jk}
    +
    (\bm\Sigma_t)_{jk}\bm Q_{il}
    +
    (\bm\Sigma_t)_{jl}\bm Q_{ik},
    \qquad
    \bm Q:=\bm A^\top\bm A.
\end{align}

Let $\bm\Psi_t(\bm u):=\bm\Psi(T-t,\bm u,\bm 0)$. Since $\bm\Psi_t(\bm u)\in\mathbb S^d$, the cross block simplifies to
\begin{equation}
\label{eq:WASC_cross_simplified}
    \bm C_t^{\bm Y\Sigma}\operatorname{vec}\big(\bm\Psi_t(\bm u)\big)
    =
    2\bm\Sigma_t \bm\Psi_t(\bm u)\bm A^\top\bm\rho.
\end{equation}
Hence
\begin{equation}
\label{eq:WASC_KSH}
    \bm{\mathcal K}_t^{\bm S H}(\bm u)
    =
    \diag(\bm S_t)\,H_t(\bm u)\,
    \bm\Sigma_t\big(\bm u+2\bm\Psi_t(\bm u)\bm A^\top\bm\rho\big),
\end{equation}
while
\begin{equation}
\label{eq:WASC_KSS}
    \bm{\mathcal K}_t^{\bm S\bm S}
    =
    \diag(\bm S_t)\,\bm\Sigma_t\,\diag(\bm S_t).
\end{equation}
Assuming $\bm\Sigma_t\in\mathbb S_{++}^d$, \eqref{eq:theta_u_explicit} yields
\begin{equation}
\label{eq:WASC_theta}
    \bm\theta_t(\bm u)
    =
    H_t(\bm u)\,\diag(\bm S_t)^{-1}
    \big(\bm u+2\bm\Psi_t(\bm u)\bm A^\top\bm\rho\big).
\end{equation}

For $\bm u_1,\bm u_2\in\mathcal S(\bm R)$, let $\bm\Psi_{1,t}:=\bm\Psi_t(\bm u_1)$ and $\bm\Psi_{2,t}:=\bm\Psi_t(\bm u_2)$. Then
\begin{align}
\label{eq:WASC_Gamma}
    \Gamma_t^c(\bm u_1,\bm u_2)
    &=
    \bm u_1^\top \bm\Sigma_t \bm u_2
    +
    2\bm u_1^\top \bm\Sigma_t \bm\Psi_{2,t}\bm A^\top\bm\rho
    +
    2\bm u_2^\top \bm\Sigma_t \bm\Psi_{1,t}\bm A^\top\bm\rho \nonumber\\
    &\quad
    +
    4\Tr\big(\bm\Psi_{1,t}\bm\Sigma_t\bm\Psi_{2,t}\bm A^\top\bm A\big).
\end{align}
Substituting \eqref{eq:WASC_theta} into \eqref{eq:C_explicit}, all terms involving only the asset-spanned component cancel, and one obtains
\begin{equation}
\label{eq:WASC_residual_density}
    \mathcal C_t(\bm u_1,\bm u_2)
    =
    4 H_t(\bm u_1)H_t(\bm u_2)\,
    \Tr\Big(
        \bm\Psi_{1,t}\bm\Sigma_t\bm\Psi_{2,t}\,
        \bm A^\top(\bm I_d-\bm\rho\bm\rho^\top)\bm A
    \Big).
\end{equation}
Thus the unhedgeable component is precisely the part of the covariance noise carried by directions orthogonal to the return driver $\bm\rho$.
\end{example}

\begin{example}[Multivariate BNS model: explicit hedging kernels]
\label{ex:BNS_Hedging}
Consider the multivariate BNS model \eqref{eq:BNS_Y}--\eqref{eq:BNS_Sigma}. In this case, $\bm\Sigma$ has no continuous martingale part, so the continuous covariance blocks satisfy
\begin{equation}
\label{eq:BNS_cont_blocks}
    \bm C_t^{\bm Y\bm Y}=\bm\Sigma_t,
    \qquad
    \bm C_t^{\bm Y\Sigma}=\bm 0,
    \qquad
    \bm C_t^{\Sigma\Sigma}=\bm 0.
\end{equation}
All additional covariance contributions therefore come from jumps. Applying It\^o's formula with jumps and the predictable compensator of the jump measure yields explicit local predictable-covariation kernels in the BNS model.

Let

\begin{equation*}
    \mathcal E_t(\bm u;\bm X)
    :=
    \bm u^\top \mathcal R(\bm X)
    +
    \Tr\big(\bm\Psi(T-t,\bm u,\bm 0)\bm X\big),
    \qquad
    \bm X\in\mathbb S_+^d.
\end{equation*}
Since the joint jump measure of $(\bm Y,\bm\Sigma)$ is supported on the graph of $\mathcal R$, the kernel formulas reduce to integrals with respect to $m_{\bm L}$. Define
\begin{align}
\label{eq:BNS_Vjump}
    \bm V_{\mathrm{jump}}
    &:=
    \int_{\mathbb S_+^d}
    \big(\e^{\mathcal R(\bm X)}-\bm 1\big)\big(\e^{\mathcal R(\bm X)}-\bm 1\big)^\top
    \,m_{\bm L}(\D\bm X), \\
\label{eq:BNS_Ju}
    \bm J_t(\bm u)
    &:=
    \int_{\mathbb S_+^d}
    \big(\e^{\mathcal R(\bm X)}-\bm 1\big)
    \big(\e^{\mathcal E_t(\bm u;\bm X)}-1\big)
    \,m_{\bm L}(\D\bm X).
\end{align}
Then
\begin{align}
\label{eq:BNS_KSS}
    \bm{\mathcal K}_t^{\bm S\bm S}
    &=
    \diag(\bm S_{t-})
    \big(\bm\Sigma_t+\bm V_{\mathrm{jump}}\big)
    \diag(\bm S_{t-}), \\
\label{eq:BNS_KSH}
    \bm{\mathcal K}_t^{\bm S H}(\bm u)
    &=
    \diag(\bm S_{t-})\,H_{t-}(\bm u)
    \big(\bm\Sigma_t\bm u+\bm J_t(\bm u)\big), \\
\label{eq:BNS_KHH}
    \mathcal K_t^{HH}(\bm u_1,\bm u_2)
    &=
    H_{t-}(\bm u_1)H_{t-}(\bm u_2)
    \Bigg[
        \bm u_1^\top \bm\Sigma_t \bm u_2
        +
        \int_{\mathbb S_+^d}
        \big(\e^{\mathcal E_t(\bm u_1;\bm X)}-1\big)
        \big(\e^{\mathcal E_t(\bm u_2;\bm X)}-1\big)
        m_{\bm L}(\D\bm X)
    \Bigg].
\end{align}
Therefore,
\begin{equation}
\label{eq:BNS_theta}
    \bm\theta_t(\bm u)
    =
    H_{t-}(\bm u)\,
    \diag(\bm S_{t-})^{-1}
    \big(\bm\Sigma_t+\bm V_{\mathrm{jump}}\big)^{-1}
    \big(\bm\Sigma_t\bm u+\bm J_t(\bm u)\big),
\end{equation}
and the residual covariance density is
\begin{equation}
\label{eq:BNS_residual}
    \mathcal C_t(\bm u_1,\bm u_2)
    =
    \mathcal K_t^{HH}(\bm u_1,\bm u_2)
    -
    \bm\theta_t(\bm u_1)^\top
    \bm{\mathcal K}_t^{\bm S\bm S}
    \bm\theta_t(\bm u_2).
\end{equation}
All jump integrals in \eqref{eq:BNS_Vjump}--\eqref{eq:BNS_residual} are computable in closed form under the Wishart jump law by repeated use of \eqref{eq:WishartMGF}, evaluated at appropriately shifted complex arguments. Hence the Fourier representation of $A$, $\bm B$, and $\bm C$ remains fully tractable in the BNS setting as well.
\end{example}

Combining \eqref{eq:theta_u_explicit} and \eqref{eq:C_explicit} with the Fourier formulas of Section~\ref{sec:fourierrepresentation} yields rigorous explicit semi-static hedging formulas in affine stochastic covariance models. In the continuous Wishart case, once $\phi$ and $\bm\Psi$ are known, the basis hedges $\bm\theta(\bm u)$ and the residual-bracket kernel $\mathcal C_t(\bm u_1,\bm u_2)$ are explicit, and the objects $A$, $\bm B$, and $\bm C$ follow by the contour integrations in \eqref{eq:importantFourier}. In jump-driven affine models, the same affine structure yields explicit predictable bracket kernels and therefore rigorous representations of the outer inputs, although some of the resulting expectations may still require numerical evaluation.

\section{Applications to Semi-Static Hedging of Covariance Swaps}
\label{sec:covswap-semistatic}

The preliminary numerical results of Section~\ref{preliminaryresults} already indicate the central structural message of the present paper: once the auxiliary instruments are chosen in a manner consistent with the spanning formulas of Section~\ref{sec:replicationtheory}, the residual hedging error can be reduced substantially relative to the purely dynamic benchmark. We now turn to explicit stochastic covariance models and show how the general Fourier--GKW machinery developed in Sections~\ref{sec:VO-framework}--\ref{sec:fourierrepresentation} yields tractable formulas for covariance swaps. In particular, we specialize the abstract semi-static variance-optimal hedging problem to the two model classes introduced in Section~\ref{sec:affinecase}, namely the continuous Wishart affine stochastic covariance model and the multivariate Barndorff--Nielsen and Shephard model with matrix-valued jump covariance. Let $i,j\in\{1,\dots,d\}$, and let $\bm Y=\log \bm S$. A covariance swap written on the pair $(S^i,S^j)$ pays at maturity $T$
\begin{equation}
\label{eq:app_covswap_payoff}
    H_T^{\mathrm{cov}}
    :=
    \langle Y^i,Y^j\rangle_{0,T}
    -
    K_{\mathrm{cov}},
\end{equation}
where $K_{\mathrm{cov}}$ is the strike. Under the pricing measure $\mathbb Q$, the fair strike is
\begin{equation}
\label{eq:app_covswap_strike}
    K_{\mathrm{cov}}
    :=
    \mathbb E\big[\langle Y^i,Y^j\rangle_{0,T}\big].
\end{equation}
As shown in Example~\ref{cov-swapreplicationQuanto} and Example~\ref{cov-swapreplicationlogspreads}, the floating leg $\langle Y^i,Y^j\rangle_{0,T}$ admits semi-static decompositions involving dynamic trading in the underlying assets together with static positions in log-contracts, product options, or spread-type instruments. The present subsection complements those structural replication identities by computing the corresponding variance-optimal hedging quantities in explicit stochastic covariance models.

For notational convenience, we introduce the symmetric matrix
\begin{equation}
\label{eq:app_Eij}
    \bm \e^{ij}
    :=
    \frac{1}{2}\big(\bm \e^{ij}+\bm \e^{ji}\big)\in\mathbb S^d,
\end{equation}
so that

\begin{equation*}
    \langle Y^i,Y^j\rangle_{0,T}
    =
    \int_0^T \Tr\big(\bm \e^{ij}\bm\Sigma_t\big)\,\D t
\end{equation*}
whenever the model is continuous. We also write

\begin{equation*}
    \bm a^{ij}:=\operatorname{vec}(\bm \e^{ij})\in\mathbb R^{d^2}.
\end{equation*}

\subsubsection{Covariance swaps in the Wishart affine stochastic covariance model}
\label{app:wasc-covswap}

We first consider the continuous Wishart affine stochastic covariance model of Example~\ref{WASCExample}. Under discounted prices, the dynamics are
\begin{align}
\label{eq:app_WASC_dynamics}
    \D\bm Y_t
    &=
    -\frac12 \diag(\bm\Sigma_t)\,\D t
    +
    \sqrt{\bm\Sigma_t}\,\D\widetilde{\bm B}_t, \nonumber\\
    \D\bm\Sigma_t
    &=
    \big(
        \bm\Omega+\bm M\bm\Sigma_t+\bm\Sigma_t\bm M^\top
    \big)\,\D t
    +
    \sqrt{\bm\Sigma_t}\,\D\bm W_t\,\bm A
    +
    \bm A^\top \D\bm W_t^\top \sqrt{\bm\Sigma_t},
\end{align}
with
\begin{equation}
\label{eq:app_WASC_corr}
    \D\widetilde{\bm B}_t
    =
    \D\bm W_t\,\bm\rho
    +
    \sqrt{1-\bm\rho^\top\bm\rho}\,\D\bm B_t,
    \qquad
    \bm\rho\in\mathbb R^d,
    \qquad
    \bm\rho^\top\bm\rho\le 1.
\end{equation}
Since the model is continuous,
\begin{equation}
\label{eq:app_WASC_qcov}
    \D\langle Y^i,Y^j\rangle_t
    =
    (\bm\Sigma_t)_{ij}\,\D t
    =
    \Tr\big(\bm \e^{ij}\bm\Sigma_t\big)\,\D t.
\end{equation}
Usually, $K_{\mathrm{cov}}$ is chosen such that the value
of the contract is zero at inception (par mean), and hence the corresponding strike is chosen as
\begin{equation}
\label{eq:app_WASC_strike}
    K_{\mathrm{cov}}
    =
    \int_0^T \mathbb E\left[\Tr\big(\bm \e^{ij}\bm\Sigma_t\big)\right]\,\D t.
\end{equation}

Let
\begin{equation}
\label{eq:app_WASC_M}
    \bm{\mathcal M}
    :=
    \bm I_d\otimes \bm M + \bm M\otimes \bm I_d
    \in \mathbb R^{d^2\times d^2}.
\end{equation}
For $s\ge t$, the conditional mean $\bm\mu_s:=\mathbb E[\bm\Sigma_s\mid\mathcal F_t]$ solves

\begin{equation*}
    \frac{\D}{\D s}\bm\mu_s
    =
    \bm\Omega+\bm M\bm\mu_s+\bm\mu_s\bm M^\top,
    \qquad
    \bm\mu_t=\bm\Sigma_t,
\end{equation*}
and therefore
\begin{equation}
\label{eq:app_WASC_ESigma}
    \operatorname{vec}(\bm\mu_s)
    =
    \e^{\bm{\mathcal M}(s-t)}\operatorname{vec}(\bm\Sigma_t)
    +
    \int_t^s \e^{\bm{\mathcal M}(s-u)}\operatorname{vec}(\bm\Omega)\,\D u.
\end{equation}
Assuming $\bm{\mathcal M}$ is invertible, define
\begin{align}
\label{eq:app_WASC_Acal}
    \bm{\mathcal A}(t,T)
    &:=
    \int_t^T \e^{\bm{\mathcal M}(s-t)}\,\D s
    =
    \bm{\mathcal M}^{-1}\big(\e^{\bm{\mathcal M}(T-t)}-\bm I_{d^2}\big), \\
\label{eq:app_WASC_bcal}
    \bm{ b}(t,T)
    &:=
    \int_t^T \int_t^s \e^{\bm{\mathcal M}(s-u)}\,\D u\,\D s \, \operatorname{vec}(\bm\Omega) \nonumber\\
    &=
    \Big(
        \bm{\mathcal M}^{-2}\big(\e^{\bm{\mathcal M}(T-t)}-\bm I_{d^2}\big)
        -(T-t)\bm{\mathcal M}^{-1}
    \Big)\operatorname{vec}(\bm\Omega).
\end{align}
Then
\begin{equation}
\label{eq:app_WASC_future_integrated_cov}
    \int_t^T \mathbb E\left[\operatorname{vec}(\bm\Sigma_s)\mid\mathcal F_t\right]\,\D s
    =
    \bm{\mathcal A}(t,T)\operatorname{vec}(\bm\Sigma_t)
    +
    \bm{ b}(t,T).
\end{equation}
We now introduce the deterministic matrix-valued coefficient
\begin{equation}
\label{eq:app_WASC_Gij}
    \bm G_{ij}(t)
    :=
    \operatorname{mat}\Big(
        \bm{\mathcal A}(t,T)^\top \bm a^{ij}
    \Big)\in\mathbb S^d,
\end{equation}
and the deterministic scalar
\begin{equation}
\label{eq:app_WASC_cij}
    c_{ij}(t)
    :=
    (\bm a^{ij})^\top \bm{ b}(t,T).
\end{equation}
It follows that the value process of the covariance swap can be written as
\begin{equation}
\label{eq:app_WASC_H0}
    H_t^0
    :=
    \mathbb E\left[H_T^{\mathrm{cov}}\mid\mathcal F_t\right]
    =
    \int_0^t \Tr\big(\bm \e^{ij}\bm\Sigma_s\big)\,\D s
    +
    \Tr\big(\bm G_{ij}(t)\bm\Sigma_t\big)
    +
    c_{ij}(t)
    -
    K_{\mathrm{cov}}.
\end{equation}
In particular, $H^0$ is affine in the state variable $\bm\Sigma_t$, up to the already-realized covariance.

The continuous martingale part of $H^0$ is generated entirely by the continuous martingale part of $\bm\Sigma$. Using the explicit Wishart covariance structure from Example~\ref{ex:WASC_Hedging}, one obtains
\begin{align}
\label{eq:app_WASC_H0_qv}
    \frac{\D}{\D t}\langle H^0,H^0\rangle_t
    &=
    4\Tr\big(
        \bm G_{ij}(t)\bm\Sigma_t\bm G_{ij}(t)\bm A^\top\bm A
    \big), \\
\label{eq:app_WASC_YH0}
    \frac{\D}{\D t}\langle \bm Y,H^0\rangle_t
    &=
    2\bm\Sigma_t \bm G_{ij}(t)\bm A^\top\bm\rho.
\end{align}
Consequently,
\begin{equation}
\label{eq:app_WASC_SH0}
    \frac{\D}{\D t}\langle \bm S,H^0\rangle_t
    =
    2\diag(\bm S_t)\bm\Sigma_t \bm G_{ij}(t)\bm A^\top\bm\rho,
\end{equation}
while
\begin{equation}
\label{eq:app_WASC_SS}
    \frac{\D}{\D t}\llangle \bm S,\bm S\rrangle_t
    =
    \diag(\bm S_t)\bm\Sigma_t\diag(\bm S_t).
\end{equation}
Therefore, the variance-optimal dynamic hedge of the covariance swap is
\begin{equation}
\label{eq:app_WASC_theta0}
    \bm\theta_t^0
    =
    \left(
        \frac{\D}{\D t}\llangle \bm S,\bm S\rrangle_t
    \right)^{-1}
    \frac{\D}{\D t}\langle \bm S,H^0\rangle_t
    =
    2\diag(\bm S_t)^{-1}\bm G_{ij}(t)\bm A^\top\bm\rho.
\end{equation}
The crucial feature of \eqref{eq:app_WASC_theta0} is that the hedge does not depend on $\bm\Sigma_t$: the covariance state enters the claim value process $H^0$, but the GKW projection eliminates it from the final hedge ratio because both $\langle \bm S,H^0\rangle$ and $\llangle \bm S,\bm S\rrangle$ carry the common factor $\bm\Sigma_t$.

To identify the unhedgeable component, define
\begin{equation}
\label{eq:app_WASC_Vperp}
    \bm V^\perp
    :=
    \bm A^\top
    \big(\bm I_d-\bm\rho\bm\rho^\top\big)
    \bm A.
\end{equation}
Then
\begin{equation}
\label{eq:app_WASC_residual_density}
    \frac{\D}{\D t}\langle L^0,L^0\rangle_t
    =
    4\Tr\big(
        \bm G_{ij}(t)\bm\Sigma_t\bm G_{ij}(t)\bm V^\perp
    \big),
\end{equation}
and therefore
\begin{equation}
\label{eq:app_WASC_A}
    A
    =
    \mathbb E\big[\langle L^0,L^0\rangle_T\big]
    =
    4\int_0^T
    \Tr\big(
        \bm G_{ij}(t)\mathbb E[\bm\Sigma_t]\bm G_{ij}(t)\bm V^\perp
    \big)\,\D t.
\end{equation}
Since $\mathbb E[\bm\Sigma_t]$ is explicitly available from \eqref{eq:app_WASC_ESigma}, the baseline mean-squared hedging error reduces to a deterministic time integral.

\subsubsection{Static portfolio inputs in the Wishart model}
\label{app:wasc-static-inputs}

Let $\eta^k=h^k(\bm Y_T)$, $k=1,\dots,n$, be auxiliary European claims admitting Laplace representations
\begin{equation}
\label{eq:app_WASC_aux_payoff}
    \eta^k
    =
    \int_{\mathcal S(\bm R^k)}
    \e^{\bm u^\top \bm Y_T}\,
    \zeta^k(\D\bm u),
\end{equation}
and let

\begin{equation*}
    H_t(\bm u)
    :=
    \mathbb E\left[\e^{\bm u^\top \bm Y_T}\mid\mathcal F_t\right],
    \qquad
    \bm\Psi_t(\bm u)
    :=
    \bm\Psi(T-t,\bm u,\bm 0).
\end{equation*}
From Example~\ref{ex:WASC_Hedging}, the basis GKW integrand is
\begin{equation}
\label{eq:app_WASC_theta_basis}
    \bm\theta_t(\bm u)
    =
    H_t(\bm u)\,\diag(\bm S_t)^{-1}
    \big(
        \bm u + 2\bm\Psi_t(\bm u)\bm A^\top\bm\rho
    \big).
\end{equation}
Therefore, the dynamic hedge of the $k$-th auxiliary claim is
\begin{equation}
\label{eq:app_WASC_theta_k}
    \bm\theta_t^k
    =
    \diag(\bm S_t)^{-1}
    \int_{\mathcal S(\bm R^k)}
    H_t(\bm u)
    \big(
        \bm u + 2\bm\Psi_t(\bm u)\bm A^\top\bm\rho
    \big)
    \zeta^k(\D\bm u).
\end{equation}

The residual cross-covariation between the covariance swap and a basis claim $H(\bm u)$ is obtained exactly as in \eqref{eq:app_WASC_residual_density}. One finds
\begin{equation}
\label{eq:app_WASC_L0Lu}
    \frac{\D}{\D t}\langle L^0,L(\bm u)\rangle_t
    =
    4 H_t(\bm u)\,
    \Tr\big(
        \bm G_{ij}(t)\bm\Sigma_t\bm\Psi_t(\bm u)\bm V^\perp
    \big).
\end{equation}
Hence the $k$-th component of $\bm B$ is
\begin{equation}
\label{eq:app_WASC_B}
    B_k
    =
    4\int_{\mathcal S(\bm R^k)}\int_0^T
    \mathbb E\left[
        H_t(\bm u)\,
        \Tr\big(
            \bm G_{ij}(t)\bm\Sigma_t\bm\Psi_t(\bm u)\bm V^\perp
        \big)
    \right]
    \D t\,
    \zeta^k(\D\bm u).
\end{equation}
Similarly, for two basis claims $H(\bm u_1)$ and $H(\bm u_2)$,
\begin{equation}
\label{eq:app_WASC_LuLv}
    \frac{\D}{\D t}\langle L(\bm u_1),L(\bm u_2)\rangle_t
    =
    4H_t(\bm u_1)H_t(\bm u_2)\,
    \Tr\big(
        \bm\Psi_t(\bm u_1)\bm\Sigma_t\bm\Psi_t(\bm u_2)\bm V^\perp
    \big),
\end{equation}
and therefore
\begin{equation}
\label{eq:app_WASC_C}
    C_{k\ell}
    =
    4\int_{\mathcal S(\bm R^k)}
    \int_{\mathcal S(\bm R^\ell)}
    \int_0^T
    \mathbb E\left[
        H_t(\bm u_1)H_t(\bm u_2)\,
        \Tr\big(
            \bm\Psi_t(\bm u_1)\bm\Sigma_t\bm\Psi_t(\bm u_2)\bm V^\perp
        \big)
    \right]
    \D t\,
    \zeta^k(\D\bm u_1)\zeta^\ell(\D\bm u_2).
\end{equation}
Equations~\eqref{eq:app_WASC_theta0}, \eqref{eq:app_WASC_A}, \eqref{eq:app_WASC_B}, and \eqref{eq:app_WASC_C} provide the complete variance-optimal semi-static hedging system for covariance swaps in the Wishart model.

A useful structural consequence is immediate: if $\bm V^\perp=\bm 0$, equivalently if the covariance noise is fully spanned by the traded asset noise, then $A=0$, $\bm B=\bm 0$, and $\bm C=\bm 0$. In that case, the market is dynamically complete for the covariance swap, and the semi-static problem collapses to exact dynamic replication.

\subsubsection{Covariance swaps in the multivariate BNS model}
\label{app:bns-covswap}

We now turn to the multivariate Barndorff--Nielsen and Shephard model of Example~\ref{BNSExample}. The dynamics are
\begin{align}
\label{eq:app_BNS_dyn}
    \D\bm Y_t
    &=
    -\frac12 \diag(\bm\Sigma_t)\,\D t
    -\bm\kappa\,\D t
    +
    \sqrt{\bm\Sigma_t}\,\D\bm B_t
    +
    \mathcal R(\D\bm L_t), \nonumber\\
    \D\bm\Sigma_t
    &=
    \big(
        \bm b\bm\Sigma_t+\bm\Sigma_t\bm b^\top
    \big)\,\D t
    +
    \D\bm L_t,
\end{align}
where $\bm L$ is a matrix subordinator on $\mathbb S_+^d$ with Lévy measure $m_{\bm L}$, and $\mathcal R:\mathbb S^d\to\mathbb R^d$ is a linear leverage map. In contrast to the Wishart diffusion, the realized covariance now contains both a continuous and a jump contribution:
\begin{equation}
\label{eq:app_BNS_realized_cov}
    [Y^i,Y^j]_T
    =
    \int_0^T (\bm\Sigma_t)_{ij}\,\D t
    +
    \sum_{0<s\le T}\Delta Y_s^i\,\Delta Y_s^j.
\end{equation}
Since $\Delta\bm Y_s=\mathcal R(\Delta\bm L_s)$, this can be written as
\begin{equation}
\label{eq:app_BNS_realized_cov_measure}
    [Y^i,Y^j]_T
    =
    \int_0^T \Tr\big(\bm \e^{ij}\bm\Sigma_t\big)\,\D t
    +
    \int_0^T \int_{\mathbb S_+^d}
    \mathcal R_i(\bm X)\mathcal R_j(\bm X)\,
    \mu_{\bm L}(dt,\D\bm X).
\end{equation}
Accordingly, the fair covariance swap strike is
\begin{equation}
\label{eq:app_BNS_strike_general}
    K_{\mathrm{cov}}
    =
    \mathbb E\left[\int_0^T \Tr\big(\bm \e^{ij}\bm\Sigma_t\big)\,\D t\right]
    +
    T\int_{\mathbb S_+^d}\mathcal R_i(\bm X)\mathcal R_j(\bm X)\,m_{\bm L}(\D\bm X).
\end{equation}

To make the first term explicit, define the linear operator
\begin{equation}
\label{eq:app_BNS_Bop}
    \mathcal B(\bm X)
    :=
    \bm b\bm X+\bm X\bm b^\top,
    \qquad
    \bm X\in\mathbb S^d.
\end{equation}
Let
\begin{equation}
\label{eq:app_BNS_ML}
    \bm M_{\bm L}
    :=
    \int_{\mathbb S_+^d} \bm X\,m_{\bm L}(\D\bm X).
\end{equation}
Then
\begin{equation}
\label{eq:app_BNS_ESigmaT}
    \mathbb E[\bm\Sigma_T]
    =
    \e^{T\bm b}\bm\Sigma_0 \e^{T\bm b^\top}
    +
    \int_0^T \e^{s\bm b}\bm M_{\bm L}\e^{s\bm b^\top}\,\D s,
\end{equation}
and, assuming $\mathcal B$ is invertible,
\begin{equation}
\label{eq:app_BNS_intSigma}
    \mathbb E\left[\int_0^T \bm\Sigma_s\,\D s\right]
    =
    \mathcal B^{-1}\big(
        \mathbb E[\bm\Sigma_T]-\bm\Sigma_0-T\bm M_{\bm L}
    \big).
\end{equation}
Substituting into \eqref{eq:app_BNS_strike_general} yields
\begin{equation}
\label{eq:app_BNS_strike}
    K_{\mathrm{cov}}
    =
    \Tr\left(
        \bm \e^{ij}
        \mathcal B^{-1}\big(
            \mathbb E[\bm\Sigma_T]-\bm\Sigma_0-T\bm M_{\bm L}
        \big)
    \right)
    +
    T\int_{\mathbb S_+^d}\mathcal R_i(\bm X)\mathcal R_j(\bm X)\,m_{\bm L}(\D\bm X).
\end{equation}

Define the deterministic matrix
\begin{equation}
\label{eq:app_BNS_Gij}
    \bm G_{ij}^{\mathrm{BNS}}(t)
    :=
    \int_0^{T-t}
    \e^{u\bm b^\top}\bm \e^{ij}\e^{u\bm b}\,\D u
    \in \mathbb S^d,
\end{equation}
and the scalar
\begin{equation}
\label{eq:app_BNS_Jij}
    \mathcal J_{ij}
    :=
    \int_{\mathbb S_+^d}\mathcal R_i(\bm X)\mathcal R_j(\bm X)\,m_{\bm L}(\D\bm X).
\end{equation}
Then the value process of the covariance swap admits the affine decomposition
\begin{equation}
\label{eq:app_BNS_H0}
    H_t^0
    =
    [Y^i,Y^j]_t
    +
    \Tr\big(
        \bm G_{ij}^{\mathrm{BNS}}(t)\bm\Sigma_t
    \big)
    +
    c_{ij}^{\mathrm{BNS}}(t)
    +
    (T-t)\mathcal J_{ij}
    -
    K_{\mathrm{cov}},
\end{equation}
where $c_{ij}^{\mathrm{BNS}}(t)$ is a deterministic scalar collecting the contribution of the future drift of $\bm\Sigma$.

Since the BNS covariance state has no continuous martingale part, the baseline hedge is driven by jumps. A jump $\Delta \bm L_t=\bm X\in\mathbb S_+^d$ produces the claim jump
\begin{equation}
\label{eq:app_BNS_DeltaH0}
    \Delta H_t^0(\bm X)
    =
    \mathcal R_i(\bm X)\mathcal R_j(\bm X)
    +
    \Tr\big(
        \bm G_{ij}^{\mathrm{BNS}}(t)\bm X
    \big).
\end{equation}
We further define
\begin{align}
\label{eq:app_BNS_Xi}
    \bm\Xi_t
    &:=
    \bm\Sigma_t
    +
    \int_{\mathbb S_+^d}
    \big(\e^{\mathcal R(\bm X)}-\bm 1\big)
    \big(\e^{\mathcal R(\bm X)}-\bm 1\big)^\top
    m_{\bm L}(\D\bm X), \\
\label{eq:app_BNS_kappa0}
    \bm\kappa_t^0
    &:=
    \int_{\mathbb S_+^d}
    \big(\e^{\mathcal R(\bm X)}-\bm 1\big)
    \Delta H_t^0(\bm X)\,
    m_{\bm L}(\D\bm X).
\end{align}
Here $\bm\Xi_t$ is the instantaneous covariance matrix of the traded asset returns, including both diffusive and jump contributions, and $\bm\kappa_t^0$ is the instantaneous covariance between the asset returns and the covariance swap jump.

The relevant GKW kernels are therefore
\begin{align}
\label{eq:app_BNS_H0H0kernel}
    \mathcal K_t^{H^0H^0}
    &:=
    \int_{\mathbb S_+^d}
    \big(\Delta H_t^0(\bm X)\big)^2\,
    m_{\bm L}(\D\bm X), \\
\label{eq:app_BNS_SH0kernel}
    \frac{\D}{\D t}\langle \bm S,H^0\rangle_t
    &=
    \diag(\bm S_{t-})\bm\kappa_t^0, \\
\label{eq:app_BNS_SSkernel}
    \frac{\D}{\D t}\llangle \bm S,\bm S\rrangle_t
    &=
    \diag(\bm S_{t-})\bm\Xi_t\diag(\bm S_{t-}).
\end{align}
Assuming $\bm\Xi_t$ is invertible, the variance-optimal dynamic hedge is
\begin{equation}
\label{eq:app_BNS_theta0}
    \bm\theta_t^0
    =
    \diag(\bm S_{t-})^{-1}\bm\Xi_t^{-1}\bm\kappa_t^0.
\end{equation}
The minimal baseline hedging error is
\begin{equation}
\label{eq:app_BNS_A}
    A
    =
    \int_0^T
    \mathbb E\left[
        \mathcal K_t^{H^0H^0}
        -
        (\bm\kappa_t^0)^\top \bm\Xi_t^{-1}\bm\kappa_t^0
    \right]dt.
\end{equation}

\subsubsection{Explicit Wishart-jump formulas in the OU--Wishart BNS specification}
\label{app:bns-ou-wishart}

We now specialize to the compound-Poisson Wishart specification of Example~\ref{BNSExample}. Assume that $\bm L$ has Lévy measure

\begin{equation*}
    m_{\bm L}(\D\bm X)
    =
    \lambda\,F_{W_d(n,\bm\Theta)}(\D\bm X),
\end{equation*}
and that the leverage map is diagonal:
\begin{equation}
\label{eq:app_BNS_diag_leverage}
    \mathcal R(\bm X)
    =
    (\rho_1 X_{11},\dots,\rho_d X_{dd})^\top.
\end{equation}
For $k\in\{1,\dots,d\}$, define
\begin{equation}
\label{eq:app_BNS_Rk}
    \bm R_k
    :=
    \rho_k \bm \e^{kk},
\end{equation}
and let
\begin{equation}
\label{eq:app_BNS_MTheta}
    \mathcal M_{\bm\Theta}(\bm R)
    :=
    \det(\bm I_d-2\bm R\bm\Theta)^{-n/2},
    \qquad
    \bm I_d-2\bm R\bm\Theta\in\mathbb S_{++}^d.
\end{equation}
Then the jump covariance correction $\bm J:=\bm\Xi_t-\bm\Sigma_t$ is deterministic, and its entries are
\begin{equation}
\label{eq:app_BNS_Jmatrix}
    J_{k\ell}
    =
    \lambda\Big(
        \mathcal M_{\bm\Theta}(\bm R_k+\bm R_\ell)
        -
        \mathcal M_{\bm\Theta}(\bm R_k)
        -
        \mathcal M_{\bm\Theta}(\bm R_\ell)
        +
        1
    \Big).
\end{equation}
Hence
\begin{equation}
\label{eq:app_BNS_Xi_explicit}
    \bm\Xi_t
    =
    \bm\Sigma_t+\bm J.
\end{equation}

To evaluate $\bm\kappa_t^0$, define
\begin{equation}
\label{eq:app_BNS_ThetaR}
    \bm\Theta_{\bm R}
    :=
    \big(\bm\Theta^{-1}-2\bm R\big)^{-1},
\end{equation}
and let
\begin{equation}
\label{eq:app_BNS_aij}
    a_{ij}
    :=
    \rho_i\rho_j.
\end{equation}
Using tilted Wishart moments, one obtains for the $k$-th component of $\bm\kappa_t^0$
\begin{align}
\label{eq:app_BNS_kappa0_explicit}
    (\bm\kappa_t^0)_k
    &=
    \lambda a_{ij} n
    \Big[
        \mathcal M_{\bm\Theta}(\bm R_k)
        \Big(
            n(\bm\Theta_{\bm R_k})_{ii}(\bm\Theta_{\bm R_k})_{jj}
            +
            2(\bm\Theta_{\bm R_k})_{ij}^2
        \Big)
        -
        \Big(
            n\Theta_{ii}\Theta_{jj}
            +
            2\Theta_{ij}^2
        \Big)
    \Big] \nonumber\\
    &\quad
    +
    \lambda n
    \Big[
        \mathcal M_{\bm\Theta}(\bm R_k)
        \Tr\big(
            \bm G_{ij}^{\mathrm{BNS}}(t)\bm\Theta_{\bm R_k}
        \big)
        -
        \Tr\big(
            \bm G_{ij}^{\mathrm{BNS}}(t)\bm\Theta
        \big)
    \Big].
\end{align}
Similarly, the jump-variance kernel $\mathcal K_t^{H^0H^0}$ reduces to explicit Wishart moments. Writing $\bm G_t:=\bm G_{ij}^{\mathrm{BNS}}(t)$,
\begin{equation}
\label{eq:app_BNS_H0H0structure}
    \mathcal K_t^{H^0H^0}
    =
    \lambda\,
    \mathbb E\left[
        \big(
            a_{ij}J_{ii}J_{jj}
            +
            \Tr(\bm G_t J)
        \big)^2
    \right],
    \qquad
    J\sim W_d(n,\bm\Theta).
\end{equation}
The required moments using Isserlis theorem \cite{Isserlis} are
\begin{align}
\label{eq:app_BNS_quartic}
    \mathbb E\big[(J_{ii}J_{jj})^2\big]
    &=
    n(n+2)
    \Big(
        n(n+2)\Theta_{ii}^2\Theta_{jj}^2
        +
        8(n+2)\Theta_{ii}\Theta_{jj}\Theta_{ij}^2
        +
        8\Theta_{ij}^4
    \Big), \\
\label{eq:app_BNS_cubic}
    \mathbb E\big[J_{ii}J_{jj}\Tr(\bm G_t J)\big]
    &=
    n^3\Theta_{ii}\Theta_{jj}\Tr(\bm G_t\bm\Theta)
    +
    8n\Theta_{ij}(\bm\Theta\bm G_t\bm\Theta)_{ij} \nonumber\\
    &\quad
    +
    2n^2
    \Big(
        \Theta_{ij}^2\Tr(\bm G_t\bm\Theta)
        +
        \Theta_{ii}(\bm\Theta\bm G_t\bm\Theta)_{jj}
        +
        \Theta_{jj}(\bm\Theta\bm G_t\bm\Theta)_{ii}
    \Big), \\
\label{eq:app_BNS_quadratic_trace}
    \mathbb E\big[(\Tr(\bm G_t J))^2\big]
    &=
    2n\Tr(\bm\Theta\bm G_t\bm\Theta\bm G_t)
    +
    n^2\big(\Tr(\bm G_t\bm\Theta)\big)^2.
\end{align}
Substituting \eqref{eq:app_BNS_quartic}--\eqref{eq:app_BNS_quadratic_trace} into \eqref{eq:app_BNS_H0H0structure} yields a fully explicit form of the baseline BNS hedging error $A$.

\subsubsection{Fourier--Laplace auxiliary claims in the BNS model}
\label{app:bns-fourier-static}

Let $\eta^k=h^k(\bm Y_T)$, $k=1,\dots,n$, be auxiliary claims with Laplace representations

\begin{equation*}
    \eta^k
    =
    \int_{\mathcal S(\bm R^k)}
    \e^{\bm u^\top \bm Y_T}\,
    \zeta^k(\D\bm u).
\end{equation*}
Write

\begin{equation*}
    H_t(\bm u)
    :=
    \mathbb E\left[\e^{\bm u^\top \bm Y_T}\mid\mathcal F_t\right]
    =
    \exp\Big(
        \phi(T-t,\bm u,\bm 0)
        +
        \Tr\big(\bm\Psi_t(\bm u)\bm\Sigma_t\big)
        +
        \bm u^\top \bm Y_t
    \Big),
\end{equation*}
with $\bm\Psi_t(\bm u):=\bm\Psi(T-t,\bm u,\bm 0)$. Under the diagonal leverage specification \eqref{eq:app_BNS_diag_leverage}, define the effective jump-transform matrix
\begin{equation}
\label{eq:app_BNS_Ru}
    \bm R_t(\bm u)
    :=
    \bm\Psi_t(\bm u) + \operatorname{Diag}(\bm\rho\odot \bm u),
\end{equation}
so that

\begin{equation*}
    \bm u^\top \mathcal R(\bm X)
    +
    \Tr\big(
        \bm\Psi_t(\bm u)\bm X
    \big)
    =
    \Tr\big(
        \bm R_t(\bm u)\bm X
    \big).
\end{equation*}

The asset--basis covariance vector is
\begin{equation}
\label{eq:app_BNS_kappau}
    \bm\kappa_t(\bm u)
    :=
    \bm\Sigma_t\bm u
    +
    \int_{\mathbb S_+^d}
    \big(
        \e^{\mathcal R(\bm X)}-\bm 1
    \big)
    \big(
        \e^{\Tr(\bm R_t(\bm u)\bm X)}-1
    \big)
    m_{\bm L}(\D\bm X),
\end{equation}
and therefore
\begin{equation}
\label{eq:app_BNS_theta_u}
    \bm\theta_t(\bm u)
    =
    \diag(\bm S_{t-})^{-1}
    H_{t-}(\bm u)\,
    \bm\Xi_t^{-1}\bm\kappa_t(\bm u).
\end{equation}
The dynamic hedge of the $k$-th auxiliary claim is then
\begin{equation}
\label{eq:app_BNS_theta_k}
    \bm\theta_t^k
    =
    \diag(\bm S_{t-})^{-1}
    \int_{\mathcal S(\bm R^k)}
    H_{t-}(\bm u)\,
    \bm\Xi_t^{-1}\bm\kappa_t(\bm u)\,
    \zeta^k(\D\bm u).
\end{equation}

Next, the covariance between the covariance swap and a basis claim is
\begin{equation}
\label{eq:app_BNS_KH0Hu}
    \mathcal K_t^{H^0H}(\bm u)
    :=
    H_{t-}(\bm u)
    \int_{\mathbb S_+^d}
    \Delta H_t^0(\bm X)
    \big(
        \e^{\Tr(\bm R_t(\bm u)\bm X)}-1
    \big)
    m_{\bm L}(\D\bm X),
\end{equation}
and the residual cross-covariation density is
\begin{equation}
\label{eq:app_BNS_L0Lu}
    \frac{\D}{\D t}\langle L^0,L(\bm u)\rangle_t
    =
    \mathcal K_t^{H^0H}(\bm u)
    -
    H_{t-}(\bm u)\,
    (\bm\kappa_t^0)^\top \bm\Xi_t^{-1}\bm\kappa_t(\bm u).
\end{equation}
Therefore,
\begin{equation}
\label{eq:app_BNS_B}
    B_k
    =
    \int_0^T
    \int_{\mathcal S(\bm R^k)}
    \mathbb E\left[
        \mathcal K_t^{H^0H}(\bm u)
        -
        H_{t-}(\bm u)\,
        (\bm\kappa_t^0)^\top \bm\Xi_t^{-1}\bm\kappa_t(\bm u)
    \right]
    \zeta^k(\D\bm u)\,\D t.
\end{equation}

Similarly, for two basis claims $H(\bm u_1)$ and $H(\bm u_2)$,
\begin{align}
\label{eq:app_BNS_KHH}
    \mathcal K_t^{HH}(\bm u_1,\bm u_2)
    &=
    H_{t-}(\bm u_1)H_{t-}(\bm u_2)
    \Bigg[
        \bm u_1^\top \bm\Sigma_t \bm u_2 \nonumber\\
    &\qquad\qquad
        +
        \int_{\mathbb S_+^d}
        \big(
            \e^{\Tr(\bm R_t(\bm u_1)\bm X)}-1
        \big)
        \big(
            \e^{\Tr(\bm R_t(\bm u_2)\bm X)}-1
        \big)
        m_{\bm L}(\D\bm X)
    \Bigg],
\end{align}
and the corresponding residual kernel is
\begin{equation}
\label{eq:app_BNS_Ckernel}
    \mathcal C_t(\bm u_1,\bm u_2)
    :=
    \mathcal K_t^{HH}(\bm u_1,\bm u_2)
    -
    H_{t-}(\bm u_1)H_{t-}(\bm u_2)\,
    \bm\kappa_t(\bm u_1)^\top \bm\Xi_t^{-1}\bm\kappa_t(\bm u_2).
\end{equation}
Hence
\begin{equation}
\label{eq:app_BNS_C}
    C_{k\ell}
    =
    \int_0^T
    \int_{\mathcal S(\bm R^k)}
    \int_{\mathcal S(\bm R^\ell)}
    \mathbb E\left[
        \mathcal C_t(\bm u_1,\bm u_2)
    \right]
    \zeta^k(\D\bm u_1)\zeta^\ell(\D\bm u_2)\,\D t.
\end{equation}

In the OU--Wishart jump specification, all terms in \eqref{eq:app_BNS_kappau} and \eqref{eq:app_BNS_KH0Hu} admit closed forms. Indeed, for each auxiliary claim $k$ we get,
\begin{equation}
\label{eq:app_BNS_kappau_explicit}
    \big(\bm\kappa_t(\bm u)\big)_k
    =
    (\bm\Sigma_t\bm u)_k
    +
    \lambda\Big(
        \mathcal M_{\bm\Theta}\big(\bm R_t(\bm u)+\bm R_k\big)
        -
        \mathcal M_{\bm\Theta}\big(\bm R_t(\bm u)\big)
        -
        \mathcal M_{\bm\Theta}(\bm R_k)
        +
        1
    \Big),
\end{equation}
while, with $\bm\Theta_{\bm R_t(\bm u)}:=(\bm\Theta^{-1}-2\bm R_t(\bm u))^{-1}$,
\begin{align}
\label{eq:app_BNS_KH0Hu_explicit}
    \int_{\mathbb S_+^d}
    &\Delta H_t^0(\bm X)
    \big(
        \e^{\Tr(\bm R_t(\bm u)\bm X)}-1
    \big)
    m_{\bm L}(\D\bm X) \nonumber\\
    &=
    \lambda
    \Big[
        \mathcal M_{\bm\Theta}\big(\bm R_t(\bm u)\big)\,
        n\Tr\big(
            \bm G_{ij}^{\mathrm{BNS}}(t)\bm\Theta_{\bm R_t(\bm u)}
        \big)
        -
        n\Tr\big(
            \bm G_{ij}^{\mathrm{BNS}}(t)\bm\Theta
        \big)
    \Big] \nonumber\\
    &\quad
    +
    \lambda a_{ij}
    \Big[
        \mathcal M_{\bm\Theta}\big(\bm R_t(\bm u)\big)\,
        n\Big(
            2(\bm\Theta_{\bm R_t(\bm u)})_{ij}^2
            +
            n(\bm\Theta_{\bm R_t(\bm u)})_{ii}(\bm\Theta_{\bm R_t(\bm u)})_{jj}
        \Big)
        -
        n\Big(
            2\Theta_{ij}^2+n\Theta_{ii}\Theta_{jj}
        \Big)
    \Big].
\end{align}
Thus the complete semi-static hedging problem for covariance swaps in the BNS model reduces, exactly as in the Wishart case, to deterministic time integrals and Fourier contour integrals of explicit affine transform quantities.

\section{Numerical results and robustness}
\label{sec:numerics-and-examples}

In this section we report numerical experiments for variance--optimal hedging of multi--asset derivatives in the two--dimensional Wishart affine stochastic covariance model introduced in Section~\ref{sec:affinecase}, with particular emphasis on the continuous WASC specification of Example~\ref{WASCExample}. This section complements the preliminary portfolio diagnostics in Section~\ref{sec:replicationtheory}: there the focus was on the qualitative structure of the semi--static hedge and on the role of the spanning formulas, whereas here we work under a fully specified stochastic covariance model and quantify hedging performance under Monte Carlo simulation.

The numerical study is organized around two blocks of experiments. 

\begin{itemize}
    \item First, we study dynamic variance--optimal hedging of European product/quanto options and compare against a misspecified bivariate GBM $\Delta$--hedge benchmark.
    \item Second, we study semi--static hedging of covariance swaps, where the dynamic component is the variance--optimal strategy of Section~\ref{sec:covswap-semistatic} and the static component is a buy--and--hold portfolio of auxiliary European options.
\end{itemize}
  In both blocks we measure performance through terminal hedging errors and variance reduction statistics, and we keep the numerical implementation aligned with the Fourier--Laplace methodology developed in Section~\ref{sec:fourierrepresentation}.

\subsection{Model specification}
\label{subsec:model-numerics}

Under the risk--neutral measure $\mathbb{Q}$ we consider two assets

\begin{equation*}
    \bm{S}_t=(S_t^1,S_t^2)^\top,
    \qquad
    \bm{Y}_t=\log \bm{S}_t=(\log S_t^1,\log S_t^2)^\top,
\end{equation*}
with stochastic covariance matrix $\bm{\Sigma}_t\in\mathbb{S}_+^2$. The joint dynamics follow the continuous WASC model described in Section~\ref{sec:affinecase}. Since the present section is purely numerical, we only recall the ingredients needed for implementation. The covariance process satisfies

\begin{equation*}
    \D\bm{\Sigma}_t
    =
    \big(
        \alpha\,\bm{A}^\top\bm{A}
        +
        \bm{b}\bm{\Sigma}_t
        +
        \bm{\Sigma}_t\bm{b}^\top
    \big)\,\D t
    +
    \bm{\Sigma}_t^{1/2}\,\D\bm{W}_t\,\bm{A}
    +
    \bm{A}^\top \D\bm{W}_t^\top \bm{\Sigma}_t^{1/2},
\end{equation*}
where $\bm{W}$ is a $2\times2$ matrix Brownian motion, $\bm{A}\in\mathbb{R}^{2\times2}$ is the volatility--of--volatility matrix, $\bm{b}\in\mathbb{R}^{2\times2}$ is the mean--reversion matrix, and $\alpha>d-1$ is the Wishart shape parameter. The log--prices evolve according to

\begin{equation*}
    \D\bm{Y}_t
    =
    \Big(
        r\bm{1}
        -
        \tfrac12 \diag(\bm{\Sigma}_t)
    \Big)\,\D t
    +
    \bm{\Sigma}_t^{1/2}\,\D\bm{Z}_t,
\end{equation*}
with leverage specification

\begin{equation*}
    \D\bm{Z}_t
    =
    \sqrt{1-\|\bm{\rho}\|^2}\,\D\bm{B}_t
    +
    \D\bm{W}_t\,\bm{\rho},
\end{equation*}
where $r$ is the constant short rate, $\bm{B}$ is a two--dimensional Brownian motion independent of $\bm{W}$, and $\bm{\rho}\in\mathbb{R}^2$ controls the instantaneous correlation between return shocks and covariance shocks.

Throughout the experiments we fix
\begin{align}
\label{paramswasc}
    \bm{A}
    &=
    \begin{pmatrix}
        0.21 & 0.14\\
        0.14 & 0.21
    \end{pmatrix},
    &
    \bm{b}
    &=
    \begin{pmatrix}
        -2.5 & -1.5\\
        -1.5 & -2.5
    \end{pmatrix},\\
    \bm{\rho}
    &=
    \begin{pmatrix}
        -0.6\\
        -0.3
    \end{pmatrix},
    &
    \alpha &= 7.14283,
    &
    r &= 0.
\end{align}
The initial asset levels are $S_0^1=S_0^2=100$, and the initial covariance matrix is

\begin{align*}
    \bm{\Sigma}_0
    =
    \begin{pmatrix}
        0.10 & 0.07\\
        0.07 & 0.10
    \end{pmatrix}.
\end{align*}

We work on the one--year horizon $T=1$ and discretize $[0,T]$ on the equidistant grid $t_k=k\Delta t$, $k=0,\dots,N$, with $N=250$ and $\Delta t = 1/250$, corresponding to daily rebalancing on a $250$--day trading year. All hedging statistics are estimated by Monte Carlo using $P=50{,}000$ independent paths. Paths for $(\bm{\Sigma},\bm{Y})$ are generated with the splitting scheme of \cite[Section~5.5]{Alfonsi2015}.

\subsection{Variance--optimal hedging of product and spread options}

We consider European payoffs with maturity $T=1$. For product/quanto options we use
\begin{equation}
\label{eq:prod-payoff-num}
    H^{\mathrm{prod}}_{K_1,K_2}
    :=
    (S_T^1-K_1)^+(S_T^2-K_2)^+,
    \qquad K_1,K_2>0,
\end{equation}

and all the variations of them (e.g product of call and puts). Throughout this subsection the hedging instruments are restricted to the two underlyings. A predictable strategy $\bm{\theta}_t=(\theta_t^1,\theta_t^2)^\top$ generates the self--financing wealth process
\begin{equation}
\label{eq:hedge-portfolio}
    V_t
    =
    V_0 + \int_0^t \bm{\theta}_u^\top\,\D\bm{S}_u,
\end{equation}
implemented on the grid by

\begin{equation*}
    V_{t_{k+1}}
    =
    V_{t_k}
    +
    \bm{\theta}_{t_k}^\top(\bm{S}_{t_{k+1}}-\bm{S}_{t_k}).
\end{equation*}

Let $H$ denote the discounted payoff. The variance--optimal (VO) strategy $\bm{\theta}^{\mathrm{VO}}$ minimizes

\begin{equation*}
    \mathbb{E}_{\mathbb{Q}}
    \Big[
        \big(
            H - V_0 - \textstyle\int_0^T \bm{\theta}_u^\top \D\bm{S}_u
        \big)^2
    \Big]
\end{equation*}
over all admissible strategies. In our implementation $V_0^{\mathrm{VO}}=H_0=\mathbb{E}_{\mathbb{Q}}[H]$, and the continuous--time VO hedge ratios are evaluated from the Fourier--Laplace formulas of Section~\ref{sec:fourierrepresentation} specialized to the WASC model in Section~\ref{sec:affinecase}, then applied pathwise on the discrete grid:

\begin{equation*}
    V^{\mathrm{VO}}_{t_{k+1}}
    =
    V^{\mathrm{VO}}_{t_k}
    +
    \big(\bm{\theta}^{\mathrm{VO}}_{t_k}\big)^\top(\bm{S}_{t_{k+1}}-\bm{S}_{t_k}),
    \qquad
    V_0^{\mathrm{VO}} = H_0.
\end{equation*}
The terminal hedging error on path $p$ is $L_T^{\mathrm{VO},(p)}:=H^{(p)}-V_T^{\mathrm{VO},(p)}$.

All Fourier integrals entering prices and hedge ratios are evaluated numerically by tensorized Gauss--Laguerre quadrature on the imaginary part of the relevant complex strips, combined with orthant symmetrization. The same quadrature nodes are used consistently across times and simulation paths. In the experiments we use between $16$ and $20$ quadrature points per dimension, which proved sufficient for stable prices and hedge ratios.

For comparison we also implement a misspecified bivariate GBM $\Delta$--hedging benchmark. Under the proxy model,

\begin{equation*}
    \frac{\D S_t^{(m)}}{S_t^{(m)}}
    =
    r\,\D t + \sigma_m\,\D W_t^{(m)},
    \qquad
    \D\langle W^{(1)},W^{(2)}\rangle_t
    =
    \rho_{\mathrm{GBM}}\,\D t,
\end{equation*}
with constants $(\sigma_1,\sigma_2,\rho_{\mathrm{GBM}})$. The proxy volatilities are chosen as Black--Scholes implied volatilities matching near--ATM WASC vanilla prices at maturity $T$, and $\rho_{\mathrm{GBM}}$ is chosen to match the empirical correlation of WASC log--returns. In the experiments this leads to

\begin{equation*}
    \sigma_1=0.27,
    \qquad
    \sigma_2=0.27,
    \qquad
    \rho_{\mathrm{GBM}}=0.69.
\end{equation*}
For product options we compute proxy prices and deltas from the closed--form bivariate lognormal formulas; for spread options we compute proxy prices and deltas by Fourier inversion under the GBM characteristic function. The resulting proxy hedge is

\begin{equation*}
    \bm{\Delta}_t^{\mathrm{GBM}}
    :=
    \nabla_{\bm{s}} C^{\mathrm{GBM}}(t,\bm{S}_t),
\end{equation*}
and its self--financing implementation on the WASC paths is\begin{equation}
\label{eq:gbm-delta-recursion}
    V^{\Delta,\mathrm{GBM}}_{t_{k+1}}
    =
    V^{\Delta,\mathrm{GBM}}_{t_k}
    +
    \big(\bm{\Delta}^{\mathrm{GBM}}_{t_k}\big)^\top
    \big(\bm{S}_{t_{k+1}}-\bm{S}_{t_k}\big),
    \qquad k=0,\dots,N-1,
\end{equation}
with initial capital
\begin{equation}
\label{eq:gbm-delta-initial}
    V^{\Delta,\mathrm{GBM}}_{0}
    =
    C^{\mathrm{WASC}}(0,\bm{S}_0,\bm{\Sigma}_0).
\end{equation}
We thus center all hedging P\&L comparisons on the same reference prices, namely the WASC model values at inception. For brevity, we refer to \eqref{eq:gbm-delta-recursion}--\eqref{eq:gbm-delta-initial} as the $\Delta$--hedge. We also report the unhedged benchmark with initial capital $H_0$, whose terminal error is $L_T^{\mathrm{unh},(p)}:=H^{(p)}-H_0$. For any strategy $X\in\{\mathrm{unh},\Delta,\mathrm{VO}\}$ we define

\begin{equation*}
    \mathrm{RMSE}^X
    :=
    \Bigg(
        \frac1P
        \sum_{p=1}^P
        \big(L_T^{X,(p)}\big)^2
    \Bigg)^{1/2},
    \qquad
    \mathrm{RHE}^X
    :=
    100\cdot
    \frac{\mathrm{RMSE}^X}{\mathrm{RMSE}^{\mathrm{unh}}},
\end{equation*}
The gain of the VO hedge over the $\Delta$--hedge is measured by

\begin{equation*}
    \mathrm{Gain}
    :=
    100\cdot\Bigg(1-\frac{\mathrm{SSE}^{\mathrm{VO}}}{\mathrm{SSE}^{\Delta}}\Bigg).
\end{equation*}

Tables~\ref{tab:prod-allkinds-combined} reports representative variance--optimal hedging results for bivariate product options across payoff types and moneyness levels. Across all contracts shown, dynamic trading in the two underlyings produces a substantial reduction in quadratic hedging risk relative to the unhedged benchmark, and the variance--optimal (VO) strategy uniformly improves upon the misspecified bivariate GBM $\Delta$--hedge in terms of $\widehat{\mathrm{Var}}(L_T^X)$, with variance reductions ranging from $94$--$98\%$ for PP contracts and $94$--$96\%$ for CC contracts. The distributional diagnostics reveal that unhedged terminal errors are markedly non--Gaussian, with pronounced asymmetry and heavy tails — excess kurtosis exceeding $100$ for CP and PC products and skewness above $20$ in several cases. Both dynamic hedges materially compress dispersion and typically attenuate tail risk, while the VO hedge produces the tightest terminal error distributions across all payoff types, most notably halving the excess kurtosis of the GBM $\Delta$--hedge for PP contracts and reducing skewness toward zero.

\begin{table}[!htbp]
  \centering
  \caption{Representative hedging performance and distributional diagnostics for bivariate product options across payoff types and moneyness levels. Variance reduction is computed relative to the unhedged benchmark. All quantities have been computed on terminal hedging errors $L_T^X = H_T - V_T^X$.}
  \label{tab:prod-allkinds-combined}
  \begingroup
  \small
  \setlength{\tabcolsep}{3.8pt}
  \renewcommand{\arraystretch}{1.12}
  \resizebox{\textwidth}{!}{%
  \begin{tabular}{@{}ccccc l rrrr rr@{}}
    \toprule
    & & & & & &
      \multicolumn{4}{c}{$\mathrm{Hedging \ performance}$} &
      \multicolumn{2}{c}{$\mathrm{Tail\  Risk}$} \\
    \cmidrule(lr){7-10}\cmidrule(lr){11-12}
    kind & $K_1$ & $K_2$ & $m_1$ & $m_2$ & Strategy $X$
      & $\widehat{\mathrm{Var}}(L_T^X)$
      & $\mathrm{RMSE}^X$
      & $\mathrm{SSE}^X$
      & Variance Reduction [\%]
      & $\mathrm{Skew}(L_T^X)$
      & $\mathrm{ExKurt}(L_T^X)$ \\
    \midrule

    \multicolumn{12}{@{}l}{\textbf{CC (call--call) products}}\\
    \midrule
    CC & 116 & 128 & 1.16 & 1.28 & Unhedged
      & $2.19\times 10^{4}$ & $1.48\times 10^{2}$ & --                  & --      & $12.014$ & $248.441$ \\
    CC & 116 & 128 & 1.16 & 1.28 & GBM $\Delta$--hedge
      & $4.58\times 10^{3}$ & $6.77\times 10^{1}$ & $2.32\times 10^{7}$ & $79.11$ & $-2.591$ & $8.552$   \\
    CC & 116 & 128 & 1.16 & 1.28 & VO hedge
      & $9.24\times 10^{2}$ & $3.04\times 10^{1}$ & $4.79\times 10^{6}$ & $95.78$ & $-1.430$ & $8.710$   \\
    \addlinespace
    CC & 123 & 128 & 1.23 & 1.28 & Unhedged
      & $1.54\times 10^{4}$ & $1.24\times 10^{2}$ & --                  & --      & $13.913$ & $325.847$ \\
    CC & 123 & 128 & 1.23 & 1.28 & GBM $\Delta$--hedge
      & $3.83\times 10^{3}$ & $6.19\times 10^{1}$ & $1.94\times 10^{7}$ & $75.11$ & $-2.838$ & $10.452$  \\
    CC & 123 & 128 & 1.23 & 1.28 & VO hedge
      & $7.22\times 10^{2}$ & $2.69\times 10^{1}$ & $3.74\times 10^{6}$ & $95.31$ & $-1.554$ & $10.759$  \\

    \midrule
    \multicolumn{12}{@{}l}{\textbf{CP (call--put) products}}\\
    \midrule
    CP & 110 & 81 & 1.10 & 0.81 & Unhedged
      & $1.36\times 10^{2}$ & $1.17\times 10^{1}$ & --                  & --      & $30.975$ & $1103.308$ \\
    CP & 110 & 81 & 1.10 & 0.81 & GBM $\Delta$--hedge
      & $1.22\times 10^{1}$ & $3.50\times 10^{0}$ & $6.12\times 10^{4}$ & $91.01$ & $9.560$  & $248.446$  \\
    CP & 110 & 81 & 1.10 & 0.81 & VO hedge
      & $5.87\times 10^{0}$ & $2.42\times 10^{0}$ & $2.94\times 10^{4}$ & $95.69$ & $5.338$  & $119.511$  \\
    \addlinespace
    CP & 116 & 81 & 1.16 & 0.81 & Unhedged
      & $6.06\times 10^{1}$ & $7.78\times 10^{0}$ & --                  & --      & $34.651$ & $1319.605$ \\
    CP & 116 & 81 & 1.16 & 0.81 & GBM $\Delta$--hedge
      & $5.78\times 10^{0}$ & $2.40\times 10^{0}$ & $2.89\times 10^{4}$ & $90.47$ & $12.740$ & $346.580$  \\
    CP & 116 & 81 & 1.16 & 0.81 & VO hedge
      & $2.68\times 10^{0}$ & $1.64\times 10^{0}$ & $1.34\times 10^{4}$ & $95.58$ & $6.008$  & $172.104$  \\

    \midrule
    \multicolumn{12}{@{}l}{\textbf{PC (put--call) products}}\\
    \midrule
    PC & 88 & 122 & 0.88 & 1.22 & Unhedged
      & $5.71\times 10^{1}$ & $7.56\times 10^{0}$ & --                  & --      & $29.030$ & $922.223$ \\
    PC & 88 & 122 & 0.88 & 1.22 & GBM $\Delta$--hedge
      & $3.96\times 10^{0}$ & $1.99\times 10^{0}$ & $1.99\times 10^{4}$ & $93.07$ & $0.238$  & $165.112$ \\
    PC & 88 & 122 & 0.88 & 1.22 & VO hedge
      & $3.23\times 10^{0}$ & $1.80\times 10^{0}$ & $1.62\times 10^{4}$ & $94.34$ & $-5.981$ & $81.927$  \\
    \addlinespace
    PC & 94 & 122 & 0.94 & 1.22 & Unhedged
      & $1.79\times 10^{2}$ & $1.34\times 10^{1}$ & --                  & --      & $22.631$ & $581.141$ \\
    PC & 94 & 122 & 0.94 & 1.22 & GBM $\Delta$--hedge
      & $1.19\times 10^{1}$ & $3.44\times 10^{0}$ & $5.98\times 10^{4}$ & $93.37$ & $-1.554$ & $77.940$  \\
    PC & 94 & 122 & 0.94 & 1.22 & VO hedge
      & $9.76\times 10^{0}$ & $3.12\times 10^{0}$ & $4.92\times 10^{4}$ & $94.55$ & $-4.430$ & $70.384$  \\

    \midrule
    \multicolumn{12}{@{}l}{\textbf{PP (put--put) products}}\\
    \midrule
    PP & 69 & 69 & 0.69 & 0.69 & Unhedged
      & $3.62\times 10^{3}$ & $6.02\times 10^{1}$ & --                  & --      & $9.919$ & $116.746$ \\
    PP & 69 & 69 & 0.69 & 0.69 & GBM $\Delta$--hedge
      & $1.25\times 10^{2}$ & $1.12\times 10^{1}$ & $6.34\times 10^{5}$ & $96.54$ & $5.338$ & $41.160$  \\
    PP & 69 & 69 & 0.69 & 0.69 & VO hedge
      & $7.20\times 10^{1}$ & $8.48\times 10^{0}$ & $3.73\times 10^{5}$ & $98.01$ & $-0.470$ & $15.606$ \\
    \addlinespace
    PP & 76 & 69 & 0.76 & 0.69 & Unhedged
      & $6.14\times 10^{3}$ & $7.83\times 10^{1}$ & --                  & --      & $8.747$ & $91.610$  \\
    PP & 76 & 69 & 0.76 & 0.69 & GBM $\Delta$--hedge
      & $1.85\times 10^{2}$ & $1.36\times 10^{1}$ & $9.40\times 10^{5}$ & $96.99$ & $4.502$ & $30.009$  \\
    PP & 76 & 69 & 0.76 & 0.69 & VO hedge
      & $1.10\times 10^{2}$ & $1.05\times 10^{1}$ & $5.78\times 10^{5}$ & $98.20$ & $-0.348$ & $11.270$  \\

    \bottomrule
  \end{tabular}%
  }
  \endgroup
\end{table}

Figure~\ref{fig:prod-CC} illustrates these findings for a representative call--call product option. Relative to the unhedged benchmark, both dynamic strategies markedly tighten the error distribution and reduce the probability of large replication losses. The misspecified GBM $\Delta$--hedge still exhibits a visible location shift and heavier tails, whereas the variance--optimal hedge remains most sharply concentrated around zero.

\begin{figure}[!htbp]
    \centering
    \includegraphics[width=0.9\linewidth]{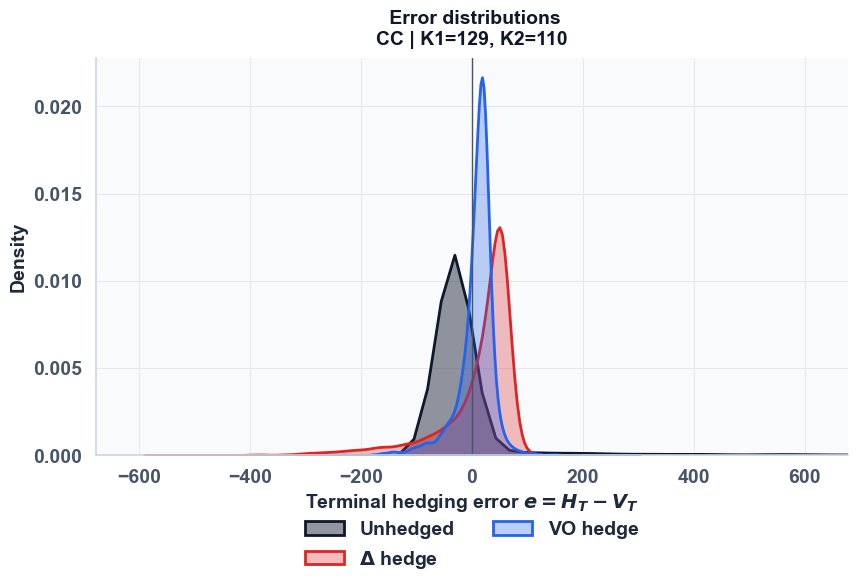}
    \caption{Terminal hedging error distributions for a representative call--call product option.}
    \label{fig:prod-CC}
\end{figure}

\begin{figure}[!htbp]
    \centering
    \includegraphics[width=0.9\linewidth]{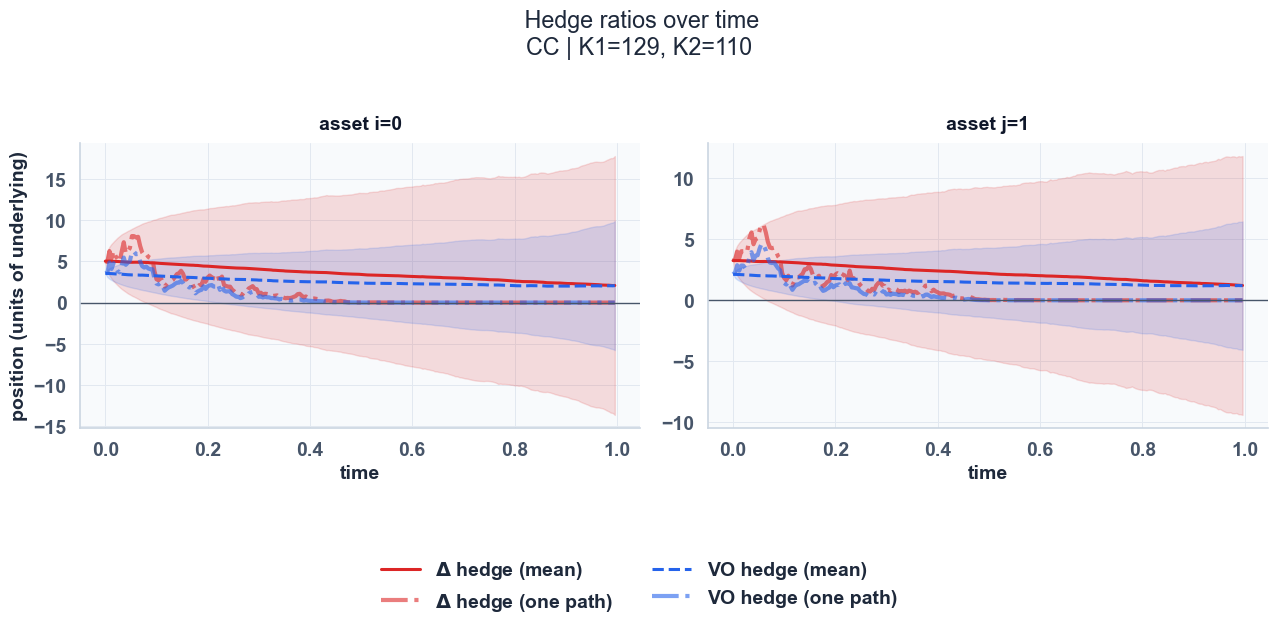}
    \caption{Representative mark-to-market trajectories for the dynamically hedged product-option portfolio under the WASC simulation setup.}
    \label{fig:MTM-product}
    \end{figure}

\begin{figure}[!htbp]
  \centering
  \begin{subfigure}[t]{0.98\linewidth}
    \centering
    \includegraphics[width=0.9\linewidth]{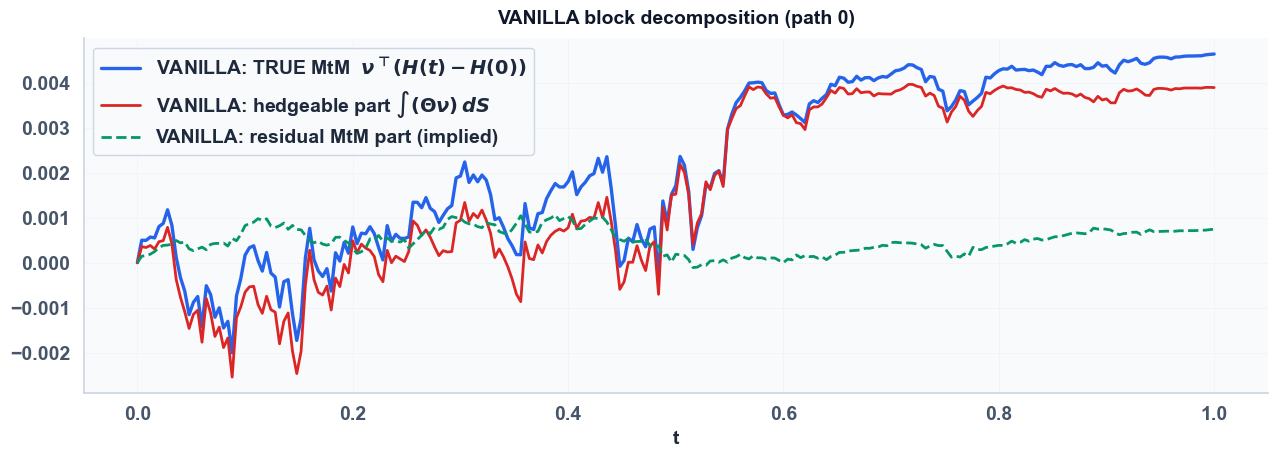}
    \subcaption{GKW decomposition Vanilla options}
    \label{fig:gkw-vanilla}
  \end{subfigure}

  \vspace{0.9em}

  \begin{subfigure}[t]{0.98\linewidth}
    \centering
    \includegraphics[width=0.9\linewidth]{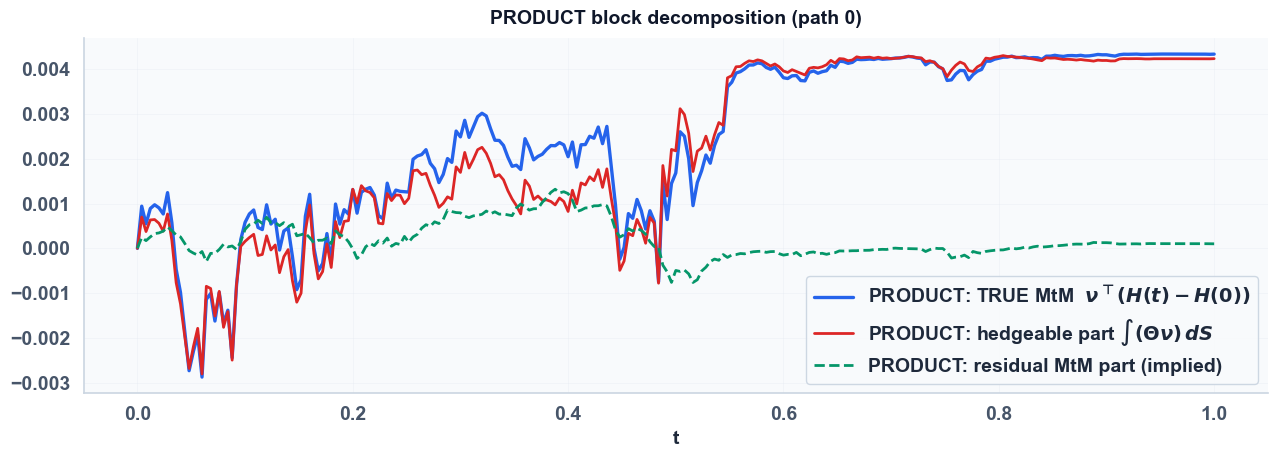}
    \subcaption{GKW decomposition Quanto options}
    \label{fig:gkw-quanto}
  \end{subfigure}
  \caption{Representative mark-to-market trajectories for the dynamically hedged product-option portfolio under the WASC simulation setup.}
  \label{fig:GKW}
\end{figure}

\vspace{0.1cm}
We emphasise that all numerical results reported above are obtained under an \emph{idealised} information set: the hedger is assumed to know the true data-generating dynamics of $(\bm S,\bm\Sigma)$ and to compute prices and hedge ratios under the correct model. In practice, this assumption is rarely justified. Calibration delivers at best an approximate fit to observed option prices, and the resulting hedge is inevitably exposed to \emph{model risk}. In particular, it is exposed to misspecification of key dependence parameters such as instantaneous covariance and correlation. To assess the robustness of our conclusions, we therefore study hedging performance in a misspecified setting in which the data are generated under the WASC dynamics, but the hedger constructs a proxy hedge without knowledge of the true model.

A classical benchmark for volatility misspecification is provided by the one-sided super-/sub-hedging results of \cite{ElKaroui}: for sufficiently regular markets, a hedging strategy computed under an assumed volatility that \emph{dominates} the true volatility yields a one-sided hedge for European and American options; analogous statements hold when the true volatility dominates the misspecified volatility. In our multivariate setting, an equally important channel of model risk is \emph{correlation} misspecification. Dispersion trading intuition expresses this risk through the spread between \emph{realised} and \emph{implied} (average) correlation: the P\&L of a zero-cost variance-dispersion position is, to leading order, proportional to the realised--implied correlation spread times an average realised variance of the constituents. 

In the same spirit, the misspecified bivariate GBM $\Delta$--hedge for the product options considered here exhibits an analogous mechanism: the leading contribution to the tracking error is \emph{linear} in the correlation misspecification and \emph{proportional} to an integrated \emph{cross-$\Gamma$ exposure}\footnote{In the misspecified setting, the option is priced under the true WASC dynamics whereas the hedge is computed from the proxy bivariate GBM. The leading contribution to the delta-hedged portfolio P\&L is
\begin{equation}
\label{eq:misspecified-corr-cross-gamma}
\D \Pi_t^{\mathrm{corr}}
\approx
\,S_t^1S_t^2\,\partial_{s_1s_2}^2 C^{\mathrm{WASC}}(t,\bm S_t,\bm\Sigma_t)\,
\Big((\bm\Sigma_t)_{12}-\rho_{\mathrm{GBM}}\sigma_1^{\mathrm{GBM}}\sigma_2^{\mathrm{GBM}}\Big)\,\D t,
\end{equation}
so that the tracking error is first order in the covariance (equivalently, correlation) misspecification and weighted by the true cross-$\Gamma$ exposure.} \cite{JacquierSlaoui}. This explains the systematic deterioration of the proxy $\Delta$--hedge as the correlation input is moved away from the effective dependence generated by the WASC model, as illustrated in Figure~\ref{fig:miss-specified-correlation}. Moreover, within the semi-static framework the mixed second-order terms, namely the cross-$\Gamma$ exposures, can be controlled by including covariance swaps among the auxiliary hedging instruments.

\begin{figure}[!htbp]
  \centering
  \includegraphics[width=0.9\linewidth]{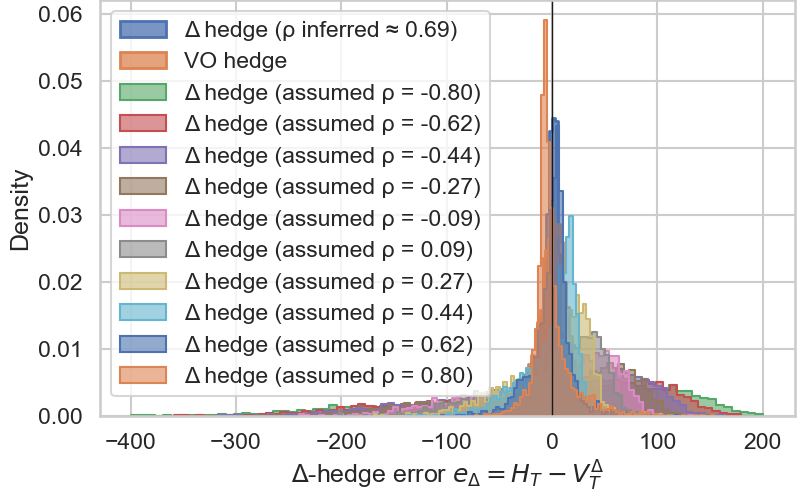}
  \caption{Sensitivity of the proxy GBM $\Delta$--hedge to correlation misspecification under the WASC data-generating dynamics.}
  \label{fig:miss-specified-correlation}
\end{figure}
\FloatBarrier

\subsection{Semi-static hedging of covariance swaps}

We now apply the semi--static variance--optimal framework of Section~\ref{sec:covswap-semistatic} to geometric covariance swaps under the same WASC dynamics. The target claim is the discounted covariance swap payoff

\begin{equation*}
    H_T^0
    =
    \langle \log S^1,\log S^2\rangle_T
    -
    K_{\mathrm{cov}},
    \qquad
    K_{\mathrm{cov}}
    =
    \mathbb{E}_{\mathbb{Q}}\big[\langle \log S^1,\log S^2\rangle_T\big],
\end{equation*}
and the dynamic component of the hedge is the variance-optimal trading strategy in the underlyings derived in Section~\ref{sec:covswap-semistatic}. We augment this dynamic hedge by a static portfolio of auxiliary European options with maturity $T$, chosen from the three instrument families motivated by the spanning identities in Section~\ref{sec:replicationtheory}: (i) vanilla options on $S^1$ and $S^2$, (ii) options on the log-ratio $\log(S^1/S^2)$ or simple ratio $S^1/S^2$, and (iii) product/quanto options. 

For a given auxiliary family $\{\eta^1,\dots,\eta^n\}$, the optimal static weights $\bm{\nu}^\star\in\mathbb{R}^n$ are computed from the finite-dimensional outer problem defined in Section~\ref{sec:VO-framework}. We emphasize a critical computational detail here: while Section~\ref{sec:fourierrepresentation} established explicit Fourier-analytic representations for the GKW components, computing the $n \times n$ residual covariance matrix $\bm{C}$ and the cross-covariance vector $\bm{B}$ via multidimensional Fourier inversion becomes computationally prohibitive and numerically unstable for large static portfolios due to highly oscillatory integrands. Consequently, we employ a hybrid numerical scheme. We utilize the exact Fourier formulas to compute the continuous-time dynamic hedge ratios (which depend only on low-dimensional marginals), while employing standard Monte Carlo simulation to efficiently and robustly estimate the residual covariances $A$, $\bm{B}$, and $\bm{C}$ for the outer static optimization. 

For each portfolio size $n$ and each auxiliary family we report the terminal semi--static hedging error

\begin{equation*}
    L_T^{\nu}
    :=
    H_T^0
    -
    \Big(
        H_0^0
        +
        \int_0^T (\bm{\theta}_t^{0}-\sum_{k=1}^n \nu_k \bm{\theta}_t^{k})^\top \D\bm{S}_t
        +
        \sum_{k=1}^n \nu_k \eta^k
    \Big),
\end{equation*}
implemented on the discrete grid, and we summarize performance by the variance reduction relative to the dynamic--only benchmark. Consistent with the preliminary diagnostics in Section~\ref{sec:replicationtheory}, adding static instruments chosen in line with the spanning formulas produces substantial additional variance reduction beyond the purely dynamic hedge, and product/quanto families typically deliver the strongest improvements among the tested auxiliary sets. We additionally study the dependence of the variance reduction on the leverage vector $\bm{\rho}$ by re-running the semi--static procedure under different leverage magnitudes while keeping the remaining model parameters fixed; this isolates how the correlation between return shocks and covariance shocks affects the attainable reduction in unhedgeable risk.

\subsection{Leverage vector and the effect on the hedging error}
\label{subsec:leverage_rho_rhe}

We study the dependence of the covariance-swap hedging performance on the leverage vector $\bm{\rho}\in\mathbb{R}^d$, where $d\mathbf Z_t=\sqrt{1-\|\bm{\rho}\|_2^2}\,\D\mathbf B_t+d\mathbf W_t\bm{\rho}$ and $\|\bm{\rho}\|_2<1$ is imposed for admissibility. The structural connection to the inner--outer problem is via the dynamic-only error term $A=\mathbb{E}[\langle L^0,L^0\rangle_T]$: in the affine WASC case, the volatility-direction residual after projection onto traded assets is the Schur complement $Q_{\mathrm{aff}}(t)=C_t^{\Sigma}-C_t^{\Sigma Y}(C_t^Y)^{-1}C_t^{Y\Sigma}$, and in the full-rank regime its dependence on leverage factors through $V^\perp:=A^\top(I-\bm{\rho}\bm{\rho}^\top)A$. In particular, $V^\perp$ is positive semidefinite and decreases in the Loewner order when the rank-one projector $\bm{\rho}\bm{\rho}^\top$ increases along directions that overlap with the column space of $A$; economically, this corresponds to transferring covariance-factor innovations from the orthogonal driver $d\mathbf B_t$ into the traded-asset driver $d\mathbf Z_t$, thereby reducing the non-spanned component of covariance risk.

\begin{figure}[!htbp]
    \centering
    \includegraphics[width=0.8\linewidth]{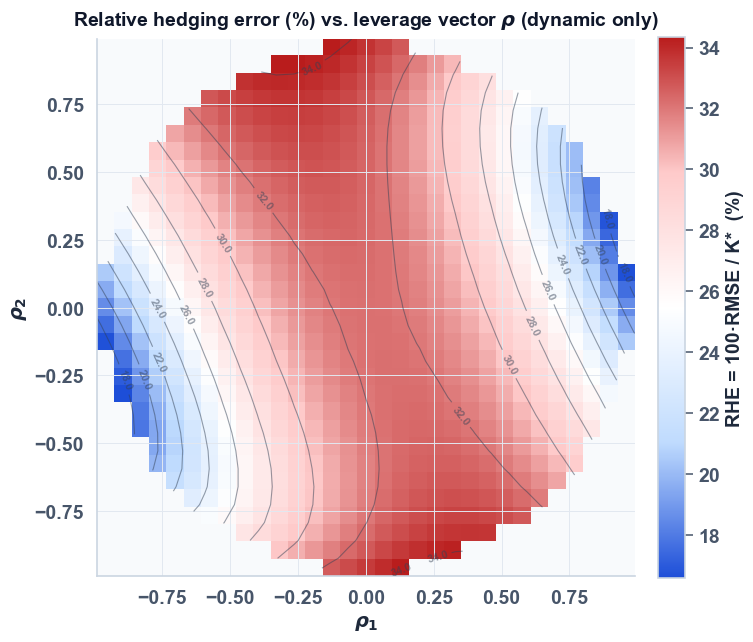}
    \caption{Relative hedging error as a function of the leverage vector $\bm{\rho}\in\mathbb{R}^d$ with $\|\bm{\rho}\|_2\leq 1$.}
    \label{fig:7}
\end{figure}
In the numerical experiment we vary $\bm{\rho}$ while holding fixed all other parameters, the monitoring grid, and the simulation budget, and we reuse common random numbers across $\bm{\rho}$ so that cross-$\bm{\rho}$ differences reflect leverage effects rather than Monte Carlo noise. For each $\bm{\rho}$ we compute the terminal error of the dynamic-only variance-optimal hedge and report the relative hedging error (RHE), i.e.\ the RMSE normalized by the absolute fair swap strike returned by the same experiment. The two-dimensional sweep over the admissible disk (Figure~\ref{fig:7}) exhibits pronounced anisotropy: the RHE is not a function of $\|\bm{\rho}\|_2$ alone but depends materially on the \emph{direction} of $\bm{\rho}$. This is consistent with the theoretical dependence through $V^\perp=A^\top(I-\bm{\rho}\bm{\rho}^\top)A$, because changing direction modifies the projector $\bm{\rho}\bm{\rho}^\top$ and therefore which covariance-factor directions are declared hedgeable via the traded assets.

To isolate the geometry, we consider one-dimensional slices. Figure~\ref{fig:rho_slices_dyn} reports the dynamic-only RHE for (i) a coordinate slice $\bm{\rho}=(x,0)$ with $x\in[-1,1]$ and (ii) a diagonal slice $\bm{\rho}=(t,t)$ with $|t|<1/\sqrt{2}$ (to maintain $\|\bm{\rho}\|_2<1$). Two features are robust. First, the curves are essentially even in the slice parameter (approximately symmetric under $x\mapsto-x$ and $t\mapsto-t$), which is consistent with the fact that the rank-one matrix $\bm{\rho}\bm{\rho}^\top$ (and hence $V^\perp$ and the projected quadratic variation driving $A$) is invariant under $\bm{\rho}\mapsto-\bm{\rho}$; in economic terms, flipping the sign of instantaneous correlation changes the direction of co-movements but not the amount of covariance-factor variance that is transmitted into traded-asset shocks. Second, the diagonal slice dominates the coordinate slice for matched correlation coordinate: for $|x|$ near one the coordinate slice achieves markedly lower RHE, while the diagonal slice cannot reach such extreme values because admissibility restricts $|t|<1/\sqrt{2}$; this is a purely multivariate constraint absent in $d=1$ and illustrates that leverage ``strength'' is limited by the unit-ball condition.

\begin{figure}[!htbp]
    \centering
    \includegraphics[width=0.8\linewidth]{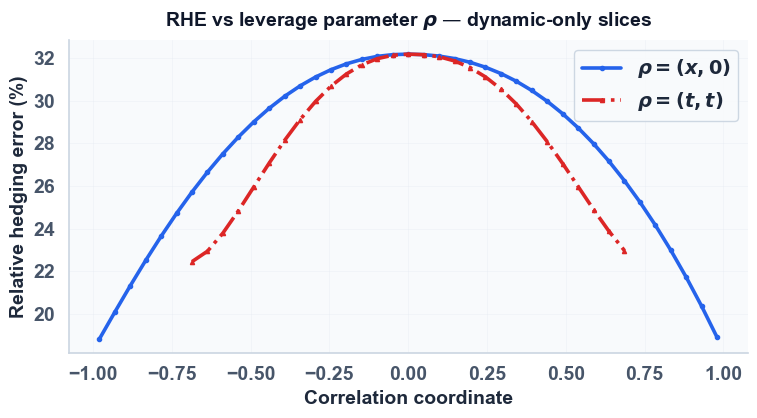}
    \caption{Dynamic-only RHE along leverage slices. The coordinate slice $\bm{\rho}=(x,0)$ is defined on $x\in[-1,1]$, while the diagonal slice $\bm{\rho}=(t,t)$ must satisfy $|t|<1/\sqrt{2}$ to enforce $\|\bm{\rho}\|_2<1$. The near-even symmetry is consistent with invariance of $\bm{\rho}\bm{\rho}^\top$ under $\bm{\rho}\mapsto-\bm{\rho}$, and differences across slices reflect directional dependence through $V^\perp=A^\top(I-\bm{\rho}\bm{\rho}^\top)A$.}
    \label{fig:rho_slices_dyn}
\end{figure}

Finally, we connect these slice observations to the effective leverage parametrization used in the semi-static comparisons. Since both the dynamic hedge $\bm{\theta}^0$ and the residual term $Q_{\mathrm{aff}}(t)$ depend on $\bm{\rho}$ through $A^\top\bm{\rho}$ and $\bm{\rho}\bm{\rho}^\top$ (hence through $A^\top\bm{\rho}\bm{\rho}^\top A$), it is natural to summarize leverage by $\rho_{\mathrm{eff}}=\|A^\top\bm{\rho}\|_2/\|A^\top\|_{\mathrm{op}}$. Along the diagonal $\bm{\rho}=(t,t)$, $\rho_{\mathrm{eff}}$ increases monotonically with $|t|$ (until the admissible boundary), and plotting RHE against $\rho_{\mathrm{eff}}$ collapses part of the directional variation because it measures alignment with the dominant covariance-loading directions encoded by $A$. In the semi-static extension (Figure~\ref{fig:8}), the same structural driver persists: the factor $V^\perp$ enters the quadratic objects $A,B,C$ and thus shifts both the baseline dynamic-only risk and the achievable variance reduction $B^\top C^{-1}B$. Hence, increasing leverage alignment reduces the irreducible orthogonal component (smaller $V^\perp$), while static product payoffs further span residual directions that remain unhedgeable by dynamic trading alone unless the degenerate complete-boundary case of the preceding remark is attained.

\begin{figure}[!htbp]
    \centering
    \includegraphics[width=0.75\linewidth]{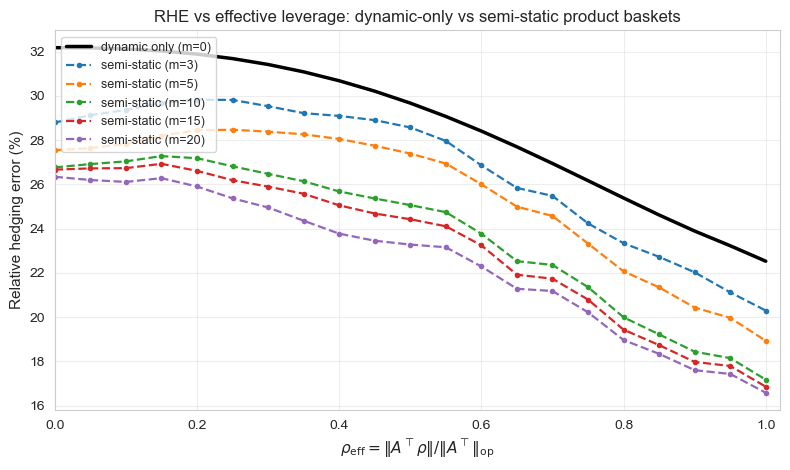}
    \caption{Relative hedging error attainable with portfolios of different effective sizes as a function of the scalar $\rho_{\mathrm{eff}}$.}
    \label{fig:8}
\end{figure}
\FloatBarrier

\section{Conclusion}In this paper, we developed a semi-static variance-optimal hedging framework for multi-asset contingent claims in incomplete markets driven by stochastic covariance processes. The global mean-variance problem was shown to decompose into an inner Galtchouk--Kunita--Watanabe projection onto dynamic trading in the underlying assets and an outer finite-dimensional quadratic optimization over static auxiliary instruments. The multidimensional spanning formulas provided a systematic instrument-selection principle and a continuous-path decomposition for covariance-sensitive claims, clarifying in particular that exact replication of a continuous covariance swap requires additional variance-linked instruments, while jump-inclusive covariance claims must be treated as distinct products. In the continuous affine setting, and especially in the Wishart model, the Fourier--GKW machinery yielded explicit representations for the dynamic hedge and for the residual covariance kernel governing the outer problem. The numerical experiments showed that suitably chosen static cross-asset instruments can materially reduce the mean-squared hedging error beyond the purely dynamic hedge. Taken together, these results support semi-static hedging as a tractable and effective framework for managing correlation and covariance risk in incomplete multi-asset markets.

\bibliographystyle{acm}
\bibliography{Citations}
\end{document}